\pdfoutput=1
\documentclass[twoside]{article}

\usepackage[preprint]{aistats2027}

\usepackage[T1]{fontenc}
\usepackage{textcomp}
\usepackage{times}
\usepackage[scaled]{helvet}
\usepackage{courier}
\usepackage{caption}
\usepackage[breaklinks]{hyperref}
\usepackage{microtype}
\usepackage{xcolor}
\definecolor{colour1}{RGB}{166,206,227}
\definecolor{colour2}{RGB}{31,120,180}
\definecolor{colour3}{RGB}{178,55,250} %
\definecolor{colour4}{RGB}{51,160,44}
\definecolor{purple}{RGB}{178,55,250} %

\newcommand{\mc}[1]{#1}

\usepackage[american]{babel}

\usepackage[round]{natbib}
\usepackage{mathtools}  %
\usepackage{amssymb}
\usepackage{amsthm}

\usepackage{enumitem}
\usepackage{needspace}

\usepackage{booktabs}
\usepackage{multirow}

\usepackage[ruled,vlined,linesnumbered]{algorithm2e}

\usepackage{tikz}
\usetikzlibrary{arrows.meta, positioning}
\usepackage{graphicx}
\usepackage{xcolor}
\definecolor{ForestGreen}{rgb}{0.13, 0.55, 0.13}
\usepackage{fontawesome5}
\hypersetup{colorlinks=true,linkcolor=blue!60!black,citecolor=blue!60!black,urlcolor=blue!60!black}
\hypersetup{pdftitle={Causal Lag Structure Discovery in Confounded Time Series via Orthogonalized Adaptive Estimation},pdfauthor={Hong Kiat Tan, Isaac-Neil Zanoria, James Chen, Haoyang Lyu, Mihai Cucuringu}}

\newtheorem{theorem}{Theorem}
\newtheorem{proposition}[theorem]{Proposition}
\newtheorem{lemma}[theorem]{Lemma}

\theoremstyle{definition}
\newtheorem{definition}{Definition}
\newtheorem{assumption}{Assumption}
\theoremstyle{remark}
\newtheorem{remark}{Remark}

\newcommand{\bA}{\mathbf{A}}
\newcommand{\bB}{\mathbf{B}}

\newcommand{\bW}{\mathbf{W}}
\newcommand{\bT}{\mathbf{T}}
\newcommand{\bY}{\mathbf{Y}}
\newcommand{\bV}{\mathbf{V}}
\newcommand{\bU}{\mathbf{U}}
\newcommand{\bC}{\mathbf{C}}
\newcommand{\bI}{\mathbf{I}}
\newcommand{\beps}{\boldsymbol{\varepsilon}}

\newcommand{\bTheta}{\boldsymbol{\Theta}}
\newcommand{\hbTheta}{\hat{\boldsymbol{\Theta}}}
\newcommand{\bzero}{\mathbf{0}}

\DeclareMathOperator{\Var}{Var}
\DeclareMathOperator{\Cov}{Cov}

\DeclareMathOperator*{\argmin}{arg\,min}

\newcommand{\convd}{\xrightarrow{d}}
\newcommand{\convp}{\xrightarrow{p}}

\newcommand{\E}{\mathbb{E}}
\newcommand{\R}{\mathbb{R}}
\newcommand{\N}{\mathcal{N}}

\newcommand{\hg}{\hat{g}}

\begin{document}

\runningtitle{Causal Lag Structure Discovery in Confounded Time Series}
\runningauthor{Tan, Zanoria, Chen, Lyu, Cucuringu}
\twocolumn[
\aistatstitle{Causal Lag Structure Discovery in Confounded\\Time Series via Orthogonalized Adaptive Estimation}
\aistatsauthor{Hong Kiat Tan$^{1,*}$ \And Isaac-Neil Zanoria$^{2}$ \And James Chen$^{1}$ \And Haoyang Lyu$^{1}$ \And Mihai Cucuringu$^{1,3,4}$}
\aistatsaddress{$^{1}$Department of Mathematics, University of California, Los Angeles\\ $^{2}$Department of Electrical and Computer Engineering, University of California, Los Angeles\\ $^{3}$Oxford-Man Institute of Quantitative Finance, University of Oxford\\ $^{4}$Department of Statistics, University of Oxford}
]
{\renewcommand{\thefootnote}{}\footnotetext{$^{*}$Correspondence to: \texttt{maxtanhk@g.ucla.edu}}}

\begin{abstract}
Finding which variables cause which others in multivariate time series, and at what lags, is central to science and policy, yet existing methods force a choice between flexible confounder adjustment, data-driven lag selection, and inference that controls the false discovery rate (FDR). ORACLE-VARX does all three in one pipeline. First, double/debiased machine learning (DML) removes nonlinear confounder effects from the outcomes and the lagged series. Second, adaptive causal lag estimation (ACLE) picks the lag order at each time step by sequential significance tests, tracking regime changes. Third, entry-wise $z$-tests with Benjamini--Hochberg correction select directed edges at a target FDR. We prove that in each rolling window, the debiased coefficients are asymptotically normal around a window-averaged target, so their $z$-tests are asymptotically valid. On a synthetic benchmark with time-varying structure and nonlinear confounding, ORACLE-VARX (LightGBM) tracks the true lag order best (RMSE $0.96$ vs $1.1$--$1.5$), has edge FDR $0.047$, close to PCMCI ($0.045$) and below VAR ($0.129$) and VAR-LiNGAM ($0.187$), and forecasts better than all three. On nine U.S. sector ETFs with macroeconomic confounders, it yields interpretable causal graphs whose lag order rises in high-volatility regimes.
\end{abstract}
\begin{center}
\faGithub\ \,\href{https://github.com/HK-Tan/ORACLE-VARX}{\texttt{HK-Tan/ORACLE-VARX}}
\end{center}

\section{Introduction}\label{sec:intro}

Causal discovery, learning which variables directly influence
which others from observational data, is a cornerstone of
scientific reasoning~\citep{pearl2009causality}.  For time series, one must determine not only
\emph{which} variables are causally linked, but \mc{also \emph{the lags at which}}
those effects operate.  Such lag-annotated graphs arise across
domains: lead--lag relationships among financial assets signal
information flow and systemic
risk~\citep{lo1990contrarian,billio2012econometric}; climate
teleconnections link weather patterns across distant
regions~\citep{wallace1981teleconnections}; and neural connectivity
studies seek directed interactions at millisecond
resolution~\citep{brovelli2004beta,seth2015granger}.  In each setting, recovering the causal lag structure is the first
step toward reliable prediction, intervention, and policy
evaluation.

Time series pose challenges that generic causal
discovery methods do not address.  Autocorrelation inflates the
effective number of tests, complicating inference.
The true lag order may be unknown and time-varying.  Moreover, exogenous confounders often enter through nonlinear channels.  Classical Granger causality and VAR methods~\citep{granger1969investigating,sims1980macroeconomics,lutkepohl2005new}
assume a single, fixed lag order and either ignore confounders or
include them linearly.  Temporal structure-learning methods---PCMCI~\citep{runge2019detecting},
DYNOTEARS~\citep{pamfil2020dynotears}, CD-NOTS~\citep{sadeghi2025causal},
VAR-LiNGAM~\citep{hyvarinen2010estimation}---each address some of these
issues, but none handles all of (a)~nonlinear confounding,
(b)~adaptive lag order selection, and (c)~FDR-controlled
edge discovery with entry-wise confidence intervals.

As a concrete illustration, consider a system of three endogenous time
series $Y^{(1)}_t, Y^{(2)}_t, Y^{(3)}_t$ driven by a lag-1 cycle
$Y^{(3)} \!\to\! Y^{(1)} \!\to\! Y^{(2)} \!\to\! Y^{(3)}$ (each edge with
coefficient $0.60$), where additional
self-effects at lags~2 and~3 decay linearly to zero over time,
producing three successive regimes (6~edges, then~5, then only the
lag-1 cycle).  Three exogenous confounders enter through distinct
nonlinear functions (tanh, quadratic, sine), creating correlated noise
that cannot be removed by linear regression.  The challenge is: given
observed trajectories of $\bY_t$ and $\bW_t$, can we simultaneously
discover \emph{(i)}~which directed edges exist, \emph{(ii)}~at which
lags they operate, and
\mc{(iii) provide valid statistical inference that controls false discoveries?}
Figure~\ref{fig:running-dag} depicts this data-generating process as a time-unrolled
\mc{acyclic digraph}.

\begin{figure}[t]
\centering
\begin{tikzpicture}[
    scale=0.72,
    every node/.style={font=\small},
    varnode/.style={circle, draw, thick, minimum size=6.5mm, inner sep=0pt},
    confbox/.style={rectangle, draw, thick, rounded corners=3pt,
                    fill=gray!12, minimum width=12mm, minimum height=7mm,
                    inner sep=2pt, font=\small},
    lag1/.style={-{Stealth[length=2.2mm]}, thick, blue!70!black},
    lag2/.style={-{Stealth[length=2.2mm]}, thick, dashed, red!70!black},
    lag3/.style={-{Stealth[length=2.2mm]}, thick, densely dotted,
                 ForestGreen!90!black},
    conf/.style={-{Stealth[length=2mm]}, thick, gray!60, densely dashed},
]

\foreach \col/\lab [count=\c] in {0/{t\!-\!3}, 2.6/{t\!-\!2}, 5.2/{t\!-\!1}, 7.8/{t}} {
    \node[anchor=south, font=\footnotesize\itshape] at (\col, 3.55) {$\lab$};
}

\foreach \col [count=\c] in {0, 2.6, 5.2, 7.8} {
    \node[varnode] (Y1-\c) at (\col, 2.7) {$Y^{(1)}$};
    \node[varnode] (Y2-\c) at (\col, 1.35) {$Y^{(2)}$};
    \node[varnode] (Y3-\c) at (\col, 0)   {$Y^{(3)}$};
}

\foreach \s/\e in {1/2, 2/3, 3/4} {
    \draw[lag1] (Y3-\s) -- (Y1-\e)
        node[midway, above, sloped, font=\scriptsize, blue!70!black] {0.60};
}
\foreach \s/\e in {1/2, 2/3, 3/4} {
    \draw[lag1] (Y1-\s) -- (Y2-\e)
        node[midway, above, sloped, font=\scriptsize, blue!70!black] {0.60};
}
\foreach \s/\e in {1/2, 2/3, 3/4} {
    \draw[lag1] (Y2-\s) -- (Y3-\e)
        node[midway, below, sloped, font=\scriptsize, blue!70!black] {0.60};
}

\draw[lag2] (Y1-1) to[bend left=42] node[above, font=\scriptsize,
    red!70!black, yshift=4pt, fill=white, inner sep=1pt] {$a(t)$} (Y1-3);
\draw[lag2] (Y1-2) to[bend left=42] node[above, font=\scriptsize,
    red!70!black, yshift=4pt, fill=white, inner sep=1pt] {$a(t)$} (Y1-4);
\draw[lag2] (Y2-1) to[bend left=22] node[pos=0.35, below, font=\scriptsize,
    red!70!black, yshift=-5pt, xshift = -5pt, fill=white, inner sep=1pt] {$a(t)$} (Y2-3);
\draw[lag2] (Y2-2) to[bend left=22] node[pos=0.35, below, font=\scriptsize,
    red!70!black, yshift=-5pt, xshift = -5pt, fill=white, inner sep=1pt] {$a(t)$} (Y2-4);

\draw[lag3] (Y3-1) to[bend right=22] node[below, font=\scriptsize,
    ForestGreen!90!black, yshift=-6pt, fill=white, inner sep=1pt] {$b(t)$} (Y3-4);

\node[confbox] (W) at (3.9, -2.35) {$\bW_t$};
\foreach \c in {1,2,3,4} {
    \draw[conf] (W) -- (Y1-\c);
    \draw[conf] (W) -- (Y2-\c);
    \draw[conf] (W) -- (Y3-\c);
}

\matrix[anchor=north west, draw, rounded corners=2pt, inner sep=3pt,
        column sep=4pt, row sep=1pt, font=\scriptsize,
        nodes={anchor=west}]
    at (5, -1.2) {
    \draw[lag1] (0,0) -- (0.5,0); & \node{Lag-1 cycle (0.60)}; \\
    \draw[lag2] (0,0) -- (0.5,0); & \node{Lag-2 self, $a(t)\!\downarrow\!0$}; \\
    \draw[lag3] (0,0) -- (0.5,0); & \node{Lag-3 self, $b(t)\!\downarrow\!0$}; \\
    \draw[conf] (0,0) -- (0.5,0); & \node{Nonlinear confounders}; \\
};
\end{tikzpicture}
\vspace{-2mm}
\caption{An example of a time-unrolled causal graph of a 3-variable system.
Solid blue arrows show the lag-1 cycle (coefficient~$0.60$).
Dashed red and dotted green arrows show lag-2 and lag-3
self-effects $a(t)$ and $b(t)$ that decay to zero
(Section~\ref{sec:synthetic}).
Gray dashed arrows indicate nonlinear confounder effects.}
\label{fig:running-dag}
\end{figure}
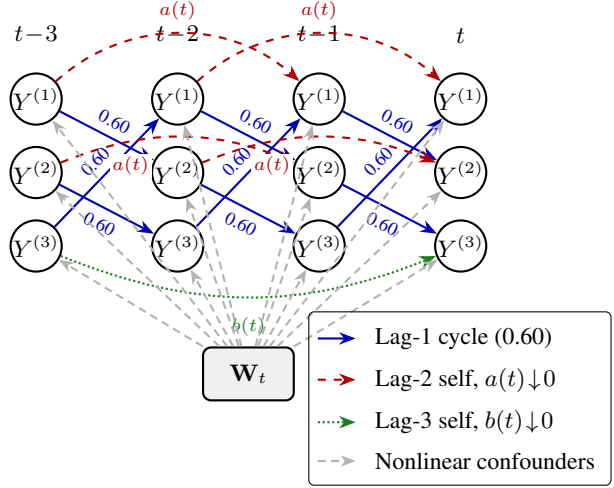

We propose ORACLE-VARX (Orthogonalized Regression for Adaptive Causal
Lag Estimation in VARX), a framework that addresses all three challenges
\mc{within a unified inferential framework}.  Our contributions are:

\vspace{-2mm}
\begin{enumerate}[itemsep=0pt]
    \item We introduce \textbf{DML orthogonalization for VAR coefficients},
    applying double/debiased machine
    learning~\citep{chernozhukov2018double} to residualize both outcomes
    and treatments against nonlinear confounder effects via flexible
    first-stage learners and time-series cross-fitting, yielding
    $\sqrt{T}$-consistent and asymptotically normal coefficient
    estimates (Section~\ref{sec:dml-var}).
    \item We propose \textbf{adaptive causal lag estimation (ACLE)},
    a two-stage procedure that combines sequential significance testing
    across a grid of significance levels with rolling validation-based
    tuning to select the lag order $\hat{p}_t$ at each time step,
    tracking regime changes in the underlying causal structure without
    requiring $p$ to be fixed \emph{a priori}
    (Section~\ref{sec:acle}).
    \item We perform \textbf{FDR-controlled edge discovery} via
    entry-wise $z$-tests on the debiased coefficients, corrected by
    the Benjamini--Hochberg procedure, yielding directed causal graphs
    with false discovery rate control at level~$q$.
    \item We establish \textbf{formal guarantees}, proving
    local asymptotic normality of the debiased coefficients in each
    rolling window, around a window-averaged target and with valid
    standard errors (Theorem~\ref{thm:dml-var}), extending the DML framework
    of~\citet{chernozhukov2018double} to the VARX setting, and
    connecting zero debiased coefficients to Granger non-causality
    (Proposition~\ref{prop:granger}).
    \item We provide \textbf{empirical validation} on both synthetic
    and real data.  On a synthetic benchmark with time-varying structure
    and nonlinear confounding, ORACLE-VARX closely tracks the true lag
    order with near-unit RMSE across five first-stage learners. On nine U.S.\ sector ETFs with macroeconomic confounders, the ACLE adaptive lag  procedure alone achieves an annualized Sharpe ratio of approximately
    1.3 on open-to-close returns, a setting where broad market indices
    earn near-zero
    returns~\citep{haghani2022nightmoves}, and the resulting causal
    graphs yield interpretable lead--lag structure whose adaptive lag
    order rises in high-volatility regimes (Sections~\ref{sec:synthetic}--\ref{sec:etf}).
\end{enumerate}

\vspace{-2mm}
\section{Related Work}\label{sec:related}

\vspace{-2mm}
\paragraph{Granger causality and VAR models.}
Classical Granger causality~\citep{granger1969investigating} tests
whether past values of one series improve the prediction of another
beyond the target's own history, and remains the most widely used
framework for analyzing lead--lag effects in multivariate time series.
\citet{sims1980macroeconomics} popularized vector autoregressive (VAR)
models as a flexible, data-driven tool for macroeconomic analysis, and
\citet{lutkepohl2005new} codified the statistical theory
for estimation, testing, and lag selection---mostly under linear dynamics
and a fixed, globally chosen lag order.
Sparse VAR extensions handle high-dimensional settings:
\citet{basu2015regularized} analyzed $\ell_1$-penalized estimation
with non-asymptotic guarantees under stability conditions, and
\citet{davis2016sparse} fit sparse VARs in two stages via partial
spectral coherence.
\citet{uematsu2025discovering} test each lag coefficient of a large
VAR with debiased-LASSO statistics and control the FDR across edges, but
in a linear, stationary VAR of known order.
All of these methods assume that confounders enter linearly or are
absent, and none adapts the lag order over time.

\vspace{-1mm}
\paragraph{Causal discovery for time series.}
PCMCI~\citep{runge2019detecting} combines PC-style parent
selection with conditional independence tests (optionally
nonlinear) that handle autocorrelation, but it neither adapts the active
lags over time nor gives lag coefficients with confidence intervals
(Table~\ref{tab:toy-results} compares it and VAR-LiNGAM with ours).
VAR-LiNGAM~\citep{hyvarinen2010estimation} uses non-Gaussianity to
identify contemporaneous and lagged structure, but assumes linear effects
and no hidden confounders.
TiMINo~\citep{peters2013causal} establishes identifiability under
independent noise and nonlinear relationships, but its algorithm
rules out feedback loops between series.
DYNOTEARS~\citep{pamfil2020dynotears} learns structure by
continuous optimization but gives no inference on individual edges.
CD-NOTS~\citep{sadeghi2025causal} handles nonstationarity by detecting
mechanism changes, but assumes every hidden confounder is a smooth
function of time.
Recent ML-based Granger and conditional independence tests for time
series~\citep{angelis2025doubly,wiecksosa2025dynamic} adjust for confounders
flexibly, and \citet{shi2026conditional} also search over the lag depth,
one pair of series at a time.
These tests detect links but give no lag coefficients with
confidence intervals.

\vspace{-3mm}
\paragraph{Double/debiased machine learning.}
\citet{robinson1988root} showed that in a partially linear model,
nonparametric nuisance fits (regressions on the confounders) followed by
OLS give $\sqrt{T}$-consistent estimates of the linear part.
\citet{chernozhukov2018double} generalized this into DML: Neyman
orthogonal scores and cross-fitting allow ML first stages that converge fast
enough, while keeping asymptotically normal inference.
Recent works extend DML to dependent data~\citep{ciganovic2026double,yu2026semiparametric}.
For high-dimensional VARs, \citet{hecq2023granger} and
\citet{adamek2023lasso} give LASSO-based inference with a fixed lag order.
Closest to our work, \citet{ballinari2024irfs} prove DML
consistency for impulse response functions under $\alpha$-mixing,
showing that blocked cross-fitting preserves asymptotic validity
(see Remark~\ref{rem:buffer}).
We apply their cross-fitting theory to VAR
coefficient matrices and add ACLE and BH-FDR control.
\citet{yu2026semiparametric} allow vector outcomes and treatments;
with the lags as treatments and instruments, their estimator reduces to our single-window
one. They assume time-invariant nuisance functions;
Theorem~\ref{thm:dml-var} covers drifting ones for our window average.
\mc{ORACLE-VARX is, to our knowledge, the first to combine
nonparametric DML debiasing, adaptive lag discovery, and FDR-controlled edge selection, with a confidence interval per lag coefficient.}

\vspace{-2mm}
\section{Methodology}\label{sec:method}
\vspace{-2mm}

\subsection{The VARX Model and Its Partially Linear Extension}\label{sec:setup}
\vspace{-1mm}

Assume that we have $n$ endogenous time series with
$\bY_t = (Y^{(1)}_t,\ldots,Y^{(n)}_t)^\top \in \R^n$ for
$t=1,\ldots,T$, together with $d$-dimensional exogenous confounders
$\bW_t = (W_{1,t},\ldots,W_{d,t})^\top \in \R^d$.

\paragraph{Linear VARX.}
The standard linear VARX($p$) model is
\begin{equation}\label{eq:varx}
  \bY_t
  = \sum_{k=1}^{p} \bA_k\,\bY_{t-k}
    + \sum_{k=1}^{p} \bB_k\,\bW_{t-k}
    + \beps_t,
\end{equation}
where $\bA_k \in \R^{n\times n}$ are lag-$k$ coefficient matrices
encoding causal effects among the endogenous variables,
$\bB_k \in \R^{n\times d}$ capture linear confounder effects,
and $\beps_t$ is a zero-mean innovation.
For simplicity, we use the same lag order~$p$ for both endogenous and exogenous terms.

\paragraph{Partially linear VARX.}
To allow flexible confounder adjustment, we generalize to a
\emph{partially linear VARX}:
\begin{equation}\label{eq:plvar}
  \bY_t
  = \sum_{k=1}^{p^*} \bA_k\,\bY_{t-k}
    + g(\bW_{t-1},\ldots,\bW_{t-p})
    + \beps_t,
\end{equation}
where $g:\R^{d \cdot p}\!\to\R^n$ is an unknown, potentially nonlinear
nuisance function,
$p^*$ is the true (unknown) maximum causal lag order, $p\ge p^*$ is the lag order used in the model, and
$\E[\beps_t \mid \bY_{t-1},\ldots,\bW_{t-p}]=\bzero$.
The effect of lagged $\bY$ on current $\bY$ remains
linear---yielding interpretable, testable coefficients---while the
confounder effect is left fully flexible.
This is the model underlying OR-VARX and ORACLE-VARX, and we can recover
the standard linear VARX~\eqref{eq:varx} when $g$ is linear.

\begin{definition}[Causal Lag Graph]\label{def:lag-graph}
The \emph{causal lag graph}
$\mathcal{G}=(V,E)$ has vertex set
$V=\{Y^{(1)},\ldots,Y^{(n)}\}$
and directed, lag-annotated edge set
$E=\bigl\{(Y^{(j)}\!\to\! Y^{(i)},\,k): \bA_k[i,j]\neq 0\bigr\}$.
An edge $(Y^{(j)}\!\to\! Y^{(i)},\,k)$ means $Y^{(j)}$ at lag~$k$
Granger-causes $Y^{(i)}$.
\end{definition}

\noindent
\textbf{Goal.}\;
Recover $\mathcal{G}$ from observational data---that is, determine
which entries $\bA_k[i,j]$ are nonzero, at which lags $k$, and for
which directed pairs $(j,i)$ \mc{while providing valid statistical inference for each edge}, while making minimal assumptions on the nuisance function~$g$.

\subsection{DML for VAR: Orthogonalization Procedure}\label{sec:dml-var}

Following the double/debiased machine learning (DML) framework
of~\citet{chernozhukov2018double}, we cast the partially linear
VARX~\eqref{eq:plvar} into a treatment-effect setup.

\paragraph{DML formulation.}
Define the \emph{stacked treatment} and \emph{stacked confounders}:
\begin{equation}\label{eq:stacked}
  \bT_t \;:=\;
  \begin{pmatrix} \bY_{t-1} \\ \vdots \\ \bY_{t-p} \end{pmatrix}
  \!\in \R^{np},
  \quad
  \bC_t \;:=\;
  \begin{pmatrix} \bW_{t-1} \\ \vdots \\ \bW_{t-p} \end{pmatrix}
  \!\in \R^{dp}.
\end{equation}
Let $\bTheta \in \R^{np\times n}$ denote the stacked coefficient matrix
formed by vertically stacking $\bA_1^\top,\ldots,\bA_p^\top$, so that $\bTheta^\top\bT_t=\sum_{k}\bA_k\bY_{t-k}$.
Then~\eqref{eq:plvar} becomes the partially linear model
\begin{align}
  \bY_t &= \bTheta^\top \bT_t + g(\bC_t) + \bU_t, \label{eq:dml-Y}\\
  \bT_t &= g_T(\bC_t) + \bV_t, \label{eq:dml-T}
\end{align}
where $g:\R^{dp}\!\to\R^n$ and $g_T:\R^{dp}\!\to\R^{np}$ are
nuisance functions,
$\bU_t$ is the structural innovation with $\E[\bU_t\mid\bC_t,\bT_t]=\bzero$,
and $\bV_t$ is the treatment residual with $\E[\bV_t\mid\bC_t]=\bzero$.
The residualization below uses the outcome's prediction from the
confounders, $g_Y(\bC_t):=\E[\bY_t\mid\bC_t]=\bTheta^\top g_T(\bC_t)+g(\bC_t)$.
The parameter of interest is $\bTheta$ (equivalently, the individual
$\bA_k$ matrices), and the nuisance functions $g_Y$, $g_T$ are treated as
infinite-dimensional parameters that must be estimated.

\paragraph{Three-step DML procedure.}

\begin{enumerate}[leftmargin=*,itemsep=2pt]

\item \textbf{Cross-fitted nuisance estimation.}
Using flexible machine learning methods (extra trees, gradient boosting,
random forests, or a foundational model), estimate $\hg_Y$ and $\hg_T$ via
\emph{time-series cross-fitting}.

\item \textbf{Residualization.}
Form the residualized outcome and treatment:
\begin{equation}\label{eq:residuals}
  \tilde{\bY}_t := \bY_t - \hg_Y(\bC_t), \quad
  \tilde{\bT}_t := \bT_t - \hg_T(\bC_t).
\end{equation}

\item \textbf{Orthogonalized OLS.}
Estimate the debiased coefficient matrix by ordinary least squares on
the residualized data:
\begin{equation}\label{eq:theta-hat}
  \hbTheta
  = \Bigl(\textstyle\sum_{t}\tilde{\bT}_t\,\tilde{\bT}_t^\top\Bigr)^{-1}
    \Bigl(\textstyle\sum_{t}\tilde{\bT}_t\,\tilde{\bY}_t^\top\Bigr).
\end{equation}
\end{enumerate}

The estimator~\eqref{eq:theta-hat} satisfies Neyman orthogonality
(Section~\ref{sec:theory}), so first-stage estimation errors in
$\hg_Y$ and $\hg_T$ affect $\hbTheta$ only at second order.
This permits flexible ML learners in the first stage without
sacrificing inferential validity.

\paragraph{Time-series cross-fitting.}
Note that standard $K$-fold cross-validation violates temporal ordering and induces
dependence between training and test sets. Hence,
we instead use a \emph{blocked rolling window split}:
fold~$k$ trains on the interval
$[s_k,\;s_k+n_{\mathrm{train}})$ and tests on
$[s_k+n_{\mathrm{train}},\;s_k+n_{\mathrm{train}}+n_{\mathrm{test}})$,
with $s_{k+1}=s_k+n_{\mathrm{test}}$ so that folds tile without overlap.
Nuisance residuals $\tilde{\bY}_t$ and $\tilde{\bT}_t$
are computed and aggregated only on the test folds, for which
the second-stage OLS~\eqref{eq:theta-hat} then uses only these out-of-sample
residuals accordingly.
We \emph{memoize} the first stage: because the blocks lie on one grid shared by all windows, each block is fit once and reused by every window that overlaps it. This needs over $200\times$ fewer first-stage fits than recomputing each window, and a full synthetic run (all 2,595 windows, one tree-based learner) takes 6--17 minutes on a 10-core Apple M5 laptop CPU, with no GPU (Appendix~\ref{app:compute}).

\paragraph{Standard errors and inference.}
After obtaining $\hbTheta$, we compute standard errors using the
classical OLS variance formula on the residualized data:
\begin{equation}\label{eq:ols-se}
  \widehat{\Var}(\hbTheta_{\cdot i})
  = \hat{\sigma}_i^2\,
    \Bigl(\textstyle\sum_{t}\tilde{\bT}_t\,\tilde{\bT}_t^\top\Bigr)^{-1},
\end{equation}
for each equation $i$, where $\hbTheta_{\cdot i}$ is the $i$-th column of $\hbTheta$,
$\hat{\sigma}_i^2 = (T_{\mathrm{eff}}-np)^{-1}\sum_t
(\tilde{Y}_{t,i} - \hbTheta_{\cdot i}^\top\tilde{\bT}_t)^2$
and $T_{\mathrm{eff}}$ is the number of out-of-sample residualized
observations. Under conditionally homoskedastic innovations this is the
correct asymptotic variance, with no HAC correction
(Theorem~\ref{thm:dml-var}).
The diagonal of $\widehat{\Var}(\hbTheta)$ then provides entry-wise standard
errors, enabling $z$-tests for the hypothesis
$H_0:\bA_k[i,j]=0$ on every potential edge in the causal graph.

\begin{remark}[Buffer Zones]\label{rem:buffer}
Time-series DML theory often inserts buffer zones of $k_T$ observations
(e.g.\ $k_T=\lfloor T/K\rfloor$) between training and test folds
\citep{ballinari2024irfs}.
We use adjacent windows ($k_T=0$). The nuisance functions themselves drift
over time, so a buffer would make every nuisance fit staler, and it would drop
$k_T$ observations from every training fold. Our theory
(Section~\ref{sec:theory}) does not use a buffer: it assumes directly that
the cross-fitted nuisance errors are negligible (condition (N),
Appendix~\ref{app:assumptions}).
\end{remark}

\vspace{0mm}
\subsection{Lag Order Selection}\label{sec:acle}

\vspace{0mm}
The lag order $p$ controls the number of coefficient matrices
$\bA_1,\ldots,\bA_p$ to estimate.
Fixing $p$ \emph{a priori} is problematic: setting it too low misses
genuine edges at higher lags, while setting it too high inflates the
parameter count and increases false positives.
The true lag order may also shift over time.
We consider two data-driven approaches.

\vspace{0mm}
\paragraph{RMSE-based lag selection.}
For methods without significance testing (VAR, VARX, OR-VARX),
we iterate over candidate lags $p \in \{1,\ldots,p_{\max}\}$,
compute rolling forecasts $\hat{\bY}_s^{(p)}$ for each,
and select the lag minimizing root mean squared forecast error
on a trailing validation window of $v$ days.
Concretely, define the rolling validation RMSE for any candidate
lag~$p$ as
\begin{equation}\label{eq:rmse-select}
  \mathrm{RMSE}_v(p;\,t)
  = \sqrt{\frac{1}{v}\sum_{s=t-v}^{t-1}
      \bigl\|\bY_s - \hat{\bY}_s^{(p)}\bigr\|^2},
\end{equation}
and set $\hat{p}_t = \argmin_{p \in \{1,\ldots,p_{\max}\}}
\mathrm{RMSE}_v(p;\,t)$.
This is simple and assumption-free but ignores the statistical evidence for individual edges.

\vspace{0mm}
\paragraph{Adaptive Causal Lag Estimation (ACLE).}
ACLE uses statistical significance to
decide when to stop adding lags, so it tracks causal signal
strength, not only forecast accuracy.
It works with both OLS-VAR and DML coefficients.

\vspace{0mm}
\paragraph{Stage~1: Significance-based lag selection.}
For each rolling window centred at time~$t$, we evaluate candidate
lags $p \in \{2,\ldots,p_{\max}\}$ via sequential hypothesis testing.
Figure~\ref{fig:acle-dag} illustrates the procedure.

\begin{figure}[t]
\centering
\begin{tikzpicture}[
    scale=0.45,
    circ/.style={circle, draw, thick, minimum size=0.8cm, align=center, font=\small},
    wcircle/.style={circle, draw, thick, minimum size=0.8cm, align=center, font=\small},
    rect/.style={rectangle, draw, thick, minimum size=0.8cm, rounded corners, align=center, font=\small},
    arrow/.style={->, thick},
    dashed_arrow/.style={->, thick, densely dashed}
]

\begin{scope}
    \node[font=\bfseries] at (3, 7.5) {Lags $1$ to $p{-}1$};

    \node[circ] (T_pre) at (0,0) {$T^{(p-1)}$};
    \node[wcircle] (W_pre) at (3,5) {$W$};
    \node[circ] (Y_pre) at (6,0) {$\bY_t$};

    \draw[arrow] (W_pre) -- (T_pre);
    \draw[arrow] (W_pre) -- (Y_pre);
    \draw[arrow] (T_pre) -- (Y_pre);

    \node[] at (3,0.6) {$\bA_{1:p-1}$};
\end{scope}

\begin{scope}[xshift=10cm]
    \node[font=\bfseries] at (3, 7.5) {Testing lag $p$};

    \node[circ] (T) at (0,0) {$T^{(p-1)}$};
    \node[wcircle] (W) at (3,5) {$W$};
    \node[circ] (Y) at (6,0) {$\bY_t$};
    \node[rect] (R) at (0,-4) {$\bY_{t-p}$};

    \draw[arrow] (W) -- (T);
    \draw[arrow] (W) -- (Y);
    \draw[arrow] (T) -- (Y);

    \node[] at (3,0.6) {$\bA_{1:p-1}$};

    \draw[arrow, densely dotted] (R) -- (T);
    \draw[dashed_arrow] (R) -- (Y);

    \node[] at (4,-2.5) {$\bA_p$};
\end{scope}
\end{tikzpicture}
\caption{Sequential lag testing in ACLE.
Left: the model with lags $1,\ldots,p{-}1$.
Right: testing $H_0\colon \bA_p = 0$ on the lag-$p$ block.
If rejected at level $\alpha$, lag $p$ is retained and we proceed to
$p{+}1$; otherwise we stop and set
$p_\alpha = p{-}1$.}
\label{fig:acle-dag}
\end{figure}
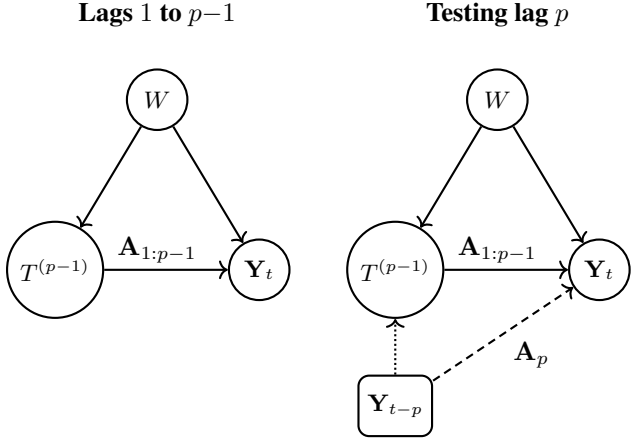

\vspace{0mm}
For each candidate lag $p \in \{2,\ldots,p_{\max}\}$, we test
\[
  H_0\colon \bA_p[i,j] = 0 \quad \text{for all } (i,j) \text{ pairs simultaneously.}
\]
Entry-wise $z$-statistics
$z_{i,j} = \hat{\bA}_p[i,j] \,/\, \hat s(\hat{\bA}_p[i,j])$
are computed from the $p_{\max}$ fit, and Benjamini--Hochberg FDR correction~\citep{benjamini1995controlling} at level~$\alpha$
determines whether any entry of $\bA_p$ is significant.
If none is, we stop and set
$p_\alpha = p - 1$.
This yields a candidate lag $p_\alpha$ for each $\alpha$. Because $\alpha$ governs the sensitivity of
lag detection, we evaluate a grid
$\alpha \in \{\alpha_1,\ldots,\alpha_M\}$, producing candidate lags $\{p_{\alpha_1},\ldots,p_{\alpha_M}\}$.

\vspace{0mm}
\paragraph{Stage~2: Validation-based $\alpha$-selection.}
To select the optimal significance level~$\alpha^*_t$ without an
\emph{a priori} choice, we apply the same rolling validation
criterion~\eqref{eq:rmse-select}, now over the $\alpha$-grid
rather than over raw lag candidates:
\begin{equation}\label{eq:alpha-select}
  \alpha^*_t
  = \argmin_{\alpha \in \{\alpha_1,\ldots,\alpha_M\}}
    \mathrm{RMSE}_v(p_\alpha;\,t),
\end{equation}
where $p_\alpha$ is the lag selected by Stage~1 at level~$\alpha$.
The final adaptively selected lag is
$\hat{p}_t = p_{\alpha^*_t}$.

The selected lag $\hat{p}_t$ may vary as the
rolling window advances: when higher-order self-effects decay to zero,
ACLE tracks the change by reducing $\hat{p}_t$.

\vspace{0mm}
\subsection{ORACLE-VARX: Full Pipeline}\label{sec:oracle-varx}
\vspace{0mm}

The complete ORACLE-VARX pipeline combines the three components above
into a unified framework for causal lag discovery.
For each rolling window at time~$t$:

\begin{enumerate}[leftmargin=*,itemsep=1pt]
  \item \textbf{Input:}\;
  Time series $\{\bY_s,\bW_s\}$ in the current window, maximum lag
  $p_{\max}$, first-stage learner~$\mathcal{L}$.

  \item \textbf{Nuisance estimation:}\;
  For each candidate $p\in\{1,\ldots,p_{\max}\}$, run DML with rolling
  cross-fitting (Section~\ref{sec:dml-var}) to estimate nuisance
  functions $\hg_Y^{(p)}$ and $\hg_T^{(p)}$, then form residuals
  $\tilde{\bY}_t^{(p)}$ and $\tilde{\bT}_t^{(p)}$.

  \item \textbf{Orthogonalized regression:}\;
  Compute $\hbTheta^{(p)}$ and its standard errors
  via~\eqref{eq:theta-hat}--\eqref{eq:ols-se} for each candidate~$p$.

  \item \textbf{Lag selection:}\;
  For non-adaptive methods, select
  $\hat{p}_t$ via RMSE minimization~\eqref{eq:rmse-select}.
  For ACLE methods, run sequential significance testing on the $p_{\max}$ fit across the
  $\alpha$-grid (Stage~1), then select $\alpha^*_t$ via rolling
  validation~\eqref{eq:alpha-select} (Stage~2) to obtain
  $\hat{p}_t = p_{\alpha^*_t}$.

  \item \textbf{BH-FDR edge discovery:}\;
  From the $p_{\max}$ fit $\hbTheta^{(p_{\max})}$, compute $z$-scores
  $z_{k,i,j} = \hat{\bA}_k[i,j]\,/\,\hat s(\hat{\bA}_k[i,j])$
  for all $k\in\{1,\ldots,\hat{p}_t\}$ and $(i,j)$ pairs.
  Apply the Benjamini--Hochberg procedure~\citep{benjamini1995controlling} at level $q$ across
  all $n^2\hat{p}_t$ tests.

  \item \textbf{Output:}\;
  Causal lag graph $\hat{\mathcal{G}}_t$ with lag-annotated edges,
  and day-$t$ forecasts $\hat{\bY}_t$.
\end{enumerate}

The full pseudocode is given in Algorithm~\ref{alg:oracle} (Appendix~\ref{app:algorithm}).

\vspace{0mm}
\subsection{Method Hierarchy}\label{sec:hierarchy}

Table~\ref{tab:hierarchy} displays the nested hierarchy of methods
evaluated in this paper. The experiments in Section~\ref{sec:experiments} evaluate every row of
Table~\ref{tab:hierarchy}, allowing us to attribute performance gains
to each individual component.

\begin{table}[h!]
  \centering
  \caption{Hierarchy of methods. Each row adds one capability
  relative to the row above. ``RMSE'' lag selection picks the lag order
  minimizing rolling validation RMSE~\eqref{eq:rmse-select} over
  $\{1,\ldots,p_{\max}\}$;
  ``ACLE'' uses significance-based sequential testing with
  validation-based $\alpha$-tuning
  (Section~\ref{sec:acle}).}\label{tab:hierarchy}
  \small
  \begin{tabular}{@{}lccc@{}}
    \toprule
    \textbf{Method}
      & \textbf{Lag Sel.}
      & \textbf{Conf.\ Adj.}
      & \textbf{DML} \\
    \midrule
    VAR
      & RMSE
      & None
      & -- \\
    VARX
      & RMSE
      & Linear
      & -- \\
    ACLE-VAR
      & ACLE
      & None
      & -- \\
    ACLE-VARX
      & ACLE
      & Linear
      & -- \\
    OR-VARX
      & RMSE
      & Nonlinear
      & \checkmark \\
    ORACLE-VARX
      & ACLE
      & Nonlinear
      & \checkmark \\
    \bottomrule
  \end{tabular}
\end{table}

\vspace{0mm}
\section{Theoretical Analysis}\label{sec:theory}

The coefficients in our setting drift over time, so we study a single
rolling window. We use the stacked form of Section~\ref{sec:dml-var} with
$p=p_{\max}$ lags and let the coefficients vary slowly:
$\bTheta_t=\bTheta(t/T)$ for a function
$\bTheta:[0,1]\to\R^{np\times n}$, so that
$\bY_t=\bTheta_t^\top\bT_t+g_t(\bC_t)+\beps_t$.
Fix $u_0\in(0,1]$ and the window
$\mathcal W=\{t_0-T_w+1,\ldots,t_0\}$, $t_0=\lfloor u_0T\rfloor$, on which
the second-stage OLS~\eqref{eq:theta-hat} is run.

\paragraph{What a window estimates.}
Let $\bV_t=\bT_t-\E[\bT_t\mid\bC_t]$ be the part of the lags that the
confounders cannot predict, and $J_t=\E[\bV_t\bV_t^\top]$. The window
estimate targets
\begin{equation}\label{eq:target}
  \bar\bTheta=\Bigl(\sum_{t\in\mathcal W}J_t\Bigr)^{-1}
  \sum_{t\in\mathcal W}J_t\,\bTheta_t,
\end{equation}
a weighted average of the coefficients over the window. These weights make
the drift cancel exactly in the estimation error, since
$\sum_{t\in\mathcal W}J_t(\bTheta_t-\bar\bTheta)=0$, and $\bar\bTheta$
lies within $O(T_w/T)$ of $\bTheta(u_0)$ (Lemma~\ref{lem:A-target}).

\begin{assumption}\label{asm:main}
We assume the following conditions:
\begin{itemize}[wide=0pt,labelsep=0.4em,itemsep=0pt,topsep=-\parskip,parsep=0pt,partopsep=0pt]
\item[\textbf{(M)}] $\E[\beps_t\mid\mathcal F_{t-1}]=\bzero$, where
$\mathcal F_{t-1}=\sigma(\bY_s,\bW_s:s<t)$: given the observed past, the
innovation is unpredictable (all confounders are observed and
$p_{\max}\ge p^*$).
\item[\textbf{(S)}] $\bTheta$ is Lipschitz and $\lambda_{\min}(J_t)\ge c_J>0$:
coefficients drift slowly and confounder-adjusted lags are not collinear.
\item[\textbf{(W)}] $T_w\to\infty$ and $T_w=O(T^{2/3})$.
\item[\textbf{(R)}] The regularity conditions (D), (H), (N) of
Appendix~\ref{app:assumptions} hold: bounded moments and short memory,
conditionally homoskedastic innovations, and the usual DML requirement that
the cross-fitted nuisance errors are small (a product rate) and
uncorrelated with $\beps_t$ and $\bV_t$.
\end{itemize}
\end{assumption}

\begin{theorem}[Local asymptotic normality]\label{thm:dml-var}
Under Assumption~\ref{asm:main}, for every entry $(r,i)$ of $\bTheta$,
with $\hat s_{ri}$ the standard error from~\eqref{eq:ols-se},
\[
  \frac{\hbTheta_{ri}-\bar\bTheta_{ri}}{\hat s_{ri}}\convd\N(0,1).
\]
If moreover $T_w=o(T^{2/3})$, then $\bar\bTheta_{ri}$ may be replaced by
$\bTheta_{ri}(u_0)$.
\end{theorem}

\paragraph{Proof idea.}
We first write the error exactly as
$\sqrt{T_w}(\hbTheta-\bar\bTheta)=\hat J^{-1}(S_0+R)$, where
$\hat J=T_w^{-1}\sum_{t\in\mathcal W}\tilde\bT_t\tilde\bT_t^\top$
is close to the average of the $J_t$. The leading term
$S_0=T_w^{-1/2}\sum_{t\in\mathcal W}\bV_t\beps_t^\top$ is a sum of
martingale differences, so a martingale central limit theorem applies with
no autocovariance (HAC) terms. The remainder $R$ collects the drift, whose
mean cancels by~\eqref{eq:target}, and the nuisance errors, which enter
only through products or terms uncorrelated with the innovations (Neyman
orthogonality), so they vanish under~(N) (Appendix~\ref{app:thm1-proof}).
Last, $T_w\hat s_{ri}^2$ is close to the limiting variance, since the residual
variance is consistent and the noise variance drifts slowly
(Appendix~\ref{app:se}).

\paragraph{Lag selection.}
ACLE is a heuristic; we do not prove that it selects the true lag.
Theorem~\ref{thm:dml-var} explains why it tracks the lag order: if
$\bA_k=0$ on the window for all lags above some $p^*$, their
$z$-statistics stay of order one, while every lag block with
$\bA_k(u_0)\neq0$ has an entry whose $\lvert z\rvert$ grows like
$\sqrt{T_w}$. With a fixed Stage~1 level $\alpha$, a zero lag is still
added with probability bounded away from zero, so $\alpha$ must shrink
with $T_w$; if it shrinks slowly enough that its cutoff grows more slowly
than $\sqrt{T_w}$, true lags are still detected. ACLE does not fix
$\alpha$: Stage~2 picks it from a grid by validation, letting the data set
the level. A formal consistency result covering the BH step and this
data-driven $\alpha$ is left to future work.

\paragraph{From coefficients to causal edges.}
Theorem~\ref{thm:dml-var} gives valid $p$-values for each entry of
$\bar\bTheta$ (of $\bA_k(u_0)$ if $T_w=o(T^{2/3})$); BH's FDR guarantee
also needs independent or positively dependent (PRDS)
$p$-values~\citep{benjamini1995controlling,benjamini2001control}.
Under causal sufficiency, meaning every common cause of the series is
observed in $\bW_t$, a zero entry of $\bA_k$ is equivalent to Granger
non-causality at lag $k$ (Proposition~\ref{prop:granger}, given a non-degeneracy condition). A hidden
variable breaks this, because its effect is absorbed into the lag
coefficients (Figure~\ref{fig:causal-suff}): a hidden common cause of two
series creates a spurious edge, and a hidden mediator makes an indirect link look direct.
The coefficients then estimate the best linear lead--lag predictor from
the observed variables, not direct causal effects, and their $p$-values
lose the guarantee of Theorem~\ref{thm:dml-var}. This happens for VAR and ACLE-VAR, which ignore the confounders
in our benchmark, and for the ETFs, where many common causes are
unobserved (Section~\ref{sec:etf}).

\vspace{-1mm}
\section{Experiments}\label{sec:experiments}

We evaluate the full method hierarchy from Table~\ref{tab:hierarchy} on two complementary benchmarks: a synthetic
time-varying system with known ground truth (Section~\ref{sec:synthetic}), and
a financial application to ETF markets with macroeconomic confounders
(Section~\ref{sec:etf}). The synthetic benchmark is our primary experiment,
providing controlled evaluation of lag recovery, coefficient estimation, and
false positive control. The financial application demonstrates that the
framework discovers interpretable causal structure in real data.

\vspace{-1mm}
\subsection{Synthetic Benchmark}\label{sec:synthetic}

\vspace{-1mm}
\paragraph{Data-generating process.}
We construct a three-variable ($n = 3$) partially linear VARX process whose
causal structure is depicted in Figure~\ref{fig:running-dag}. The lag-1
coefficient matrix is
\[
\bA_1 = \begin{pmatrix} 0 & 0 & 0.60 \\ 0.60 & 0 & 0 \\ 0 & 0.60 & 0 \end{pmatrix},
\]
and time-varying self-effects at lags 2 and 3 decay linearly to zero:
$\bA_2[1,1](t) = \bA_2[2,2](t) = 0.30 \cdot \max(0,\, 1 - t/2000)$ and
$\bA_3[3,3](t) = 0.30 \cdot \max(0,\, 1 - t/1000)$, creating three regimes
($p^*(t) = 3, 2, 1$). Three independent AR(1) confounders
($\rho_W = 0.5$) enter through nonlinear nuisance functions:
\begin{align}
g_1(\bW) &= \tfrac{\tanh(2W_1) + W_2^2 + \sin(\pi W_3)}{2.2}, \notag\\
g_2(\bW) &= \tfrac{\mathrm{ReLU}(W_1 - 0.5) + \sin(\pi W_2) + W_3^2}{2.1}, \label{eq:nuisance-fns}\\
g_3(\bW) &= \tfrac{W_1 W_2 + \tanh(W_3) + \cos(\pi W_1)}{1.7}, \notag
\end{align}
each entering the $i$-th structural equation as $\nu\lambda \cdot g_i(\bW_{t-1})$.
Key parameters: $T = 3000$, noise scale $\nu = 0.1$, confounder strength
$\lambda = 1.0$.

\vspace{0mm}
\paragraph{Observability conditions.}
All three confounders always drive the data; the observability level
sets how many of them a method is given: \textbf{All} (3 confounders observed), \textbf{Partial-2}
(2 observed), \textbf{Partial-1} (1 observed), and \textbf{None}
(no confounders; of our methods, only VAR and ACLE-VAR apply). Below All,
confounders that drive the data are left out of the model, which violates
causal sufficiency (Section~\ref{sec:theory}), so the FDR guarantee for
causal edges no longer applies. The main text therefore reports All, and
keeps VAR and ACLE-VAR (None) only as the reference for ignoring the
confounders altogether; Appendix~\ref{app:partial-obs} reports Partial-2
and Partial-1.

\vspace{0mm}
\paragraph{Methods.}
We evaluate all six method variants from Table~\ref{tab:hierarchy},
plus PCMCI (linear test, with an OLS refit) and VAR-LiNGAM, refit in every
window (Appendix~\ref{app:cd-baselines}).
For DML methods, we compare five first-stage learners: Extra Trees,
Random Forest, LightGBM, XGBoost, and TabPFN-2.5~\citep{hollmann2025tabpfn,grinsztajn2025tabpfn}.
TabPFN is a transformer pretrained on millions of synthetic tabular
datasets that performs inference via in-context learning,
achieving strong performance on small tabular datasets without
hyperparameter tuning, which is the regime of our DML cross-fitting
folds. We set $p_{\max} = 5$, a value larger than the true maximum lag of the data generating procedure.

\vspace{-1mm}
\paragraph{Metrics.}
We evaluate five complementary metrics:
\begin{enumerate}
    \item \textbf{Lag RMSE}: $\sqrt{T_{\mathrm{out}}^{-1} \sum_{t}
    (\hat{p}_t - p^*(t))^2}$, measuring how well the method tracks the
    true time-varying lag order. For RMSE-based methods, $\hat{p}_t$ is the
    validation-selected lag.
    \item \textbf{Coefficient MAE} (nonzero edges): average
    $|\hat{\bA}_k[i,j] - \bA_k^*[i,j](t)|$ over true nonzero entries,
    measuring signal recovery.
    \item \textbf{Forecast MAE}: out-of-sample prediction error.
    \item \textbf{FDR} and \textbf{Power}: the share of edges in
    $\hat{\mathcal{G}}_t$ (BH at level $q=0.05$ over all $n^2\hat{p}_t$
    entries) that are false, and the share of true edges found. An edge is
    true if its window-averaged coefficient is nonzero, since this
    average is what the estimates target (Theorem~\ref{thm:dml-var}).
\end{enumerate}

\vspace{-1mm}
\begin{table}[t]
    \centering
    \caption{Synthetic benchmark results (All confounders observed unless
    noted). ET: Extra Trees; LGBM: LightGBM; TP: TabPFN. FDR: BH at
    $q=0.05$. Bold: best per column; lower is better except for Power.}\label{tab:toy-results}
    \small
    \setlength{\tabcolsep}{1.7pt}
    \begin{tabular}{@{}lccccc@{}}
        \toprule
        \textbf{Method} &
        \textbf{Lag} &
        \textbf{Coeff.} &
        \textbf{Fcst.} &
        \textbf{FDR} &
        \textbf{Power} \\
        &
        \textbf{RMSE} &
        \textbf{MAE} &
        \textbf{MAE} &
        &
        $\uparrow$ \\
        \midrule
        \multicolumn{6}{@{}l}{\emph{No confounders}} \\
        VAR
            & 1.475 & 0.060 & 0.110 & 0.129 & 0.655  \\
        ACLE-VAR
            & 1.006 & 0.066 & 0.109 & 0.147 & 0.661  \\
        \addlinespace
        \multicolumn{6}{@{}l}{\emph{All confounders, no DML}} \\
        VARX
            & 1.239 & 0.063 & 0.109 & 0.076 & 0.650  \\
        ACLE-VARX
            & 1.099 & 0.069 & 0.108 & 0.080 & 0.651  \\
        \addlinespace
        \multicolumn{6}{@{}l}{\emph{All confounders, causal discovery}} \\
        PCMCI + OLS
            & 1.214 & 0.065 & 0.108 & \textbf{0.045} & \textbf{0.685}  \\
        VAR-LiNGAM
            & 1.132 & 0.070 & 0.107 & 0.187 & 0.641  \\
        \addlinespace
        \multicolumn{6}{@{}l}{\emph{All confounders, DML}} \\
        OR-VARX, ET
            & 1.212 & 0.063 & 0.095 & 0.094 & 0.649  \\
        OR-VARX, LGBM
            & 1.288 & 0.067 & 0.102 & \textbf{0.045} & 0.666  \\
        OR-VARX, TP
            & 1.159 & 0.067 & \textbf{0.091} & 0.056 & 0.660  \\
        \addlinespace
        ORACLE-VARX, ET
            & 1.073 & 0.063 & 0.094 & 0.101 & 0.652  \\
        ORACLE-VARX, LGBM
            & \textbf{0.962} & \textbf{0.058} & 0.101 & 0.047 & 0.673  \\
        ORACLE-VARX, TP
            & 0.981 & 0.060 & \textbf{0.091} & 0.066 & 0.673  \\
        \bottomrule
    \end{tabular}
\end{table}

\paragraph{Results.}
\emph{Finding 1: ACLE consistently improves lag tracking.}
Adaptive lag selection reduces lag RMSE regardless of the
base method (Table~\ref{tab:toy-results}). The improvement is 32\% for VAR, 11\% for VARX, and 11\% for OR-VARX with Extra Trees.
It holds for every learner and observability level
(Table~\ref{tab:full-results}; Figures~\ref{fig:edge-traj}--\ref{fig:lag-recovery}).

\emph{Finding 2: Adjusting for confounders lowers the FDR, by an amount that depends on the learner.}
Ignoring the confounders, BH edge discovery overshoots its target of
$0.05$: FDR is 0.129 for VAR. Most false edges are lag-1 self-loops,
which the autocorrelated confounders induce: each series appears to
predict itself.
Adjusting for the confounders in VARX and DML cuts FDR to 0.045--0.080 at similar
power (Table~\ref{tab:toy-results}), except with ET (0.094--0.101). Our theory needs the first-stage fits
to be accurate enough (condition~(N), Appendix~\ref{app:assumptions}). We
suspect that ET, with our hyperparameters, falls short of this: the
confounding it leaves in the residuals inflates the edge tests.

\emph{Finding 3 (tentative): TabPFN is a strong first-stage learner without tuning.}
Among the five DML learners at full observability, ORACLE-VARX with TabPFN
has the best forecast MAE (0.091), the best edge power together with LightGBM (both 0.673),
the second-lowest FDR (0.066, after LightGBM's 0.047; the other three
learners are at 0.10--0.12), and the
second-best lag RMSE (0.981, after LightGBM's 0.962;
Table~\ref{tab:full-results}). This validates TabPFN as an effective
``plug-and-play'' nuisance function estimator for DML, requiring no
hyperparameter tuning.

\emph{Finding 4: PCMCI finds edges as well as ours, but tracks lags and forecasts worse.}
PCMCI has the lowest FDR (0.045) and highest power (0.685), close to
ORACLE-VARX with LGBM (0.047, 0.673). ORACLE-VARX (LGBM, TP) is better
on lag RMSE (0.96--0.98 vs 1.21), coefficient MAE (0.058--0.060 vs
0.065) and forecast MAE (0.091--0.101 vs 0.108). This PCMCI uses a
linear test. We did not run its nonlinear version, which suits nonlinear
confounding like ours: it would take about 30 hours per configuration on
9 cores, against about 1 minute for the linear one
(Appendix~\ref{app:compute}).

\begin{figure}[t]
    \centering
    \includegraphics[width=\linewidth]{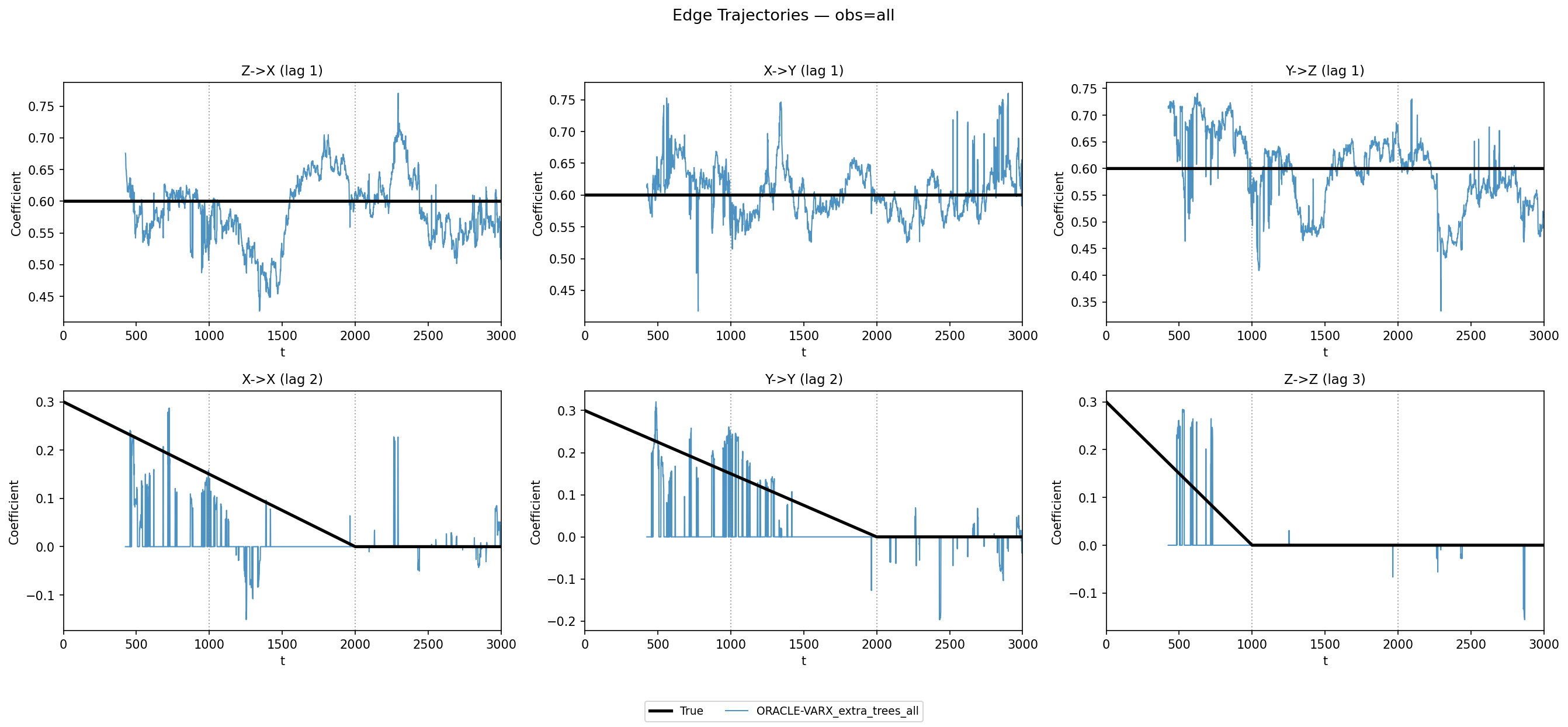}
    \caption{Estimated vs.\ true edge coefficients over time for
    ORACLE-VARX (Extra Trees, all confounders). Each panel shows one
    edge; the black line is the true coefficient $\bA_k^*[i,j](t)$ and
    the colored line is the rolling estimate $\hat{\bA}_k[i,j](t)$.
    Vertical dashed lines mark regime boundaries ($t = 1000$, $t = 2000$).}\label{fig:edge-traj}
\end{figure}

\begin{figure}[t]
    \centering
    \begin{minipage}[t]{0.48\linewidth}
        \centering
        \includegraphics[width=\linewidth]{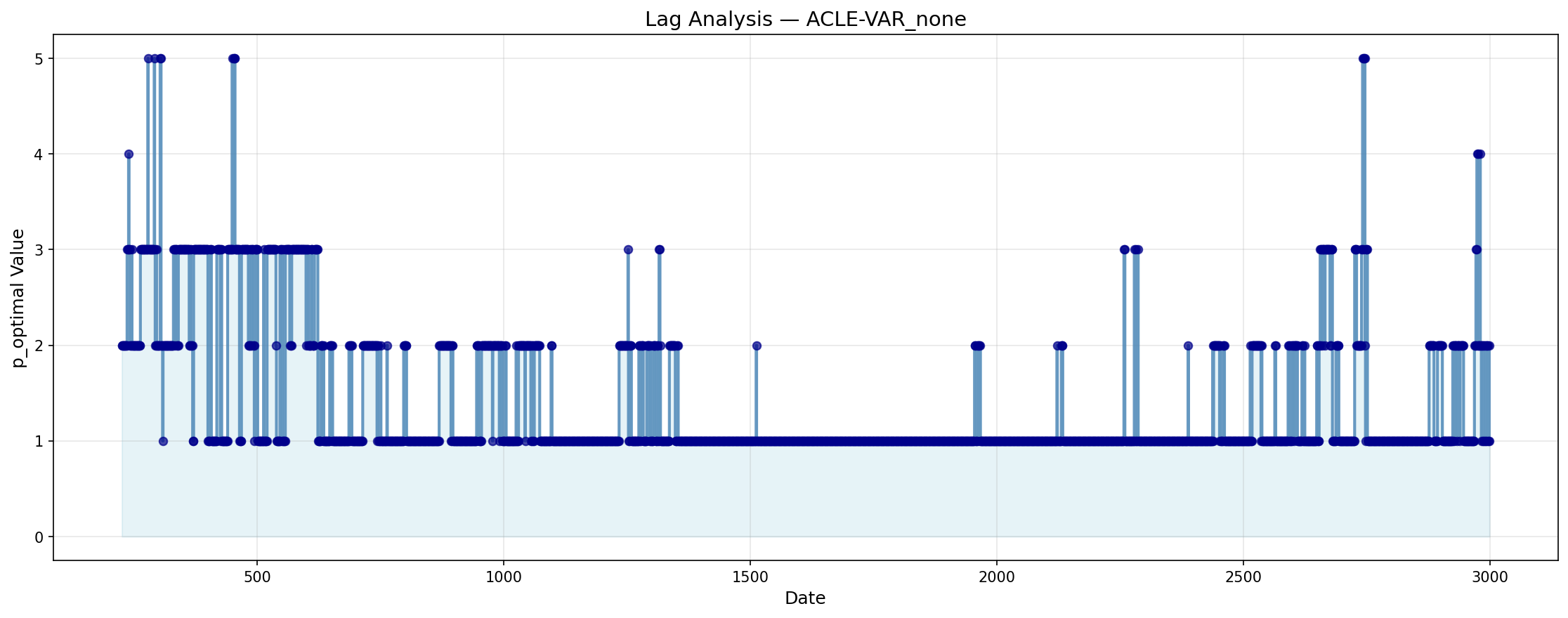}
    \end{minipage}\hfill
    \begin{minipage}[t]{0.48\linewidth}
        \centering
        \includegraphics[width=\linewidth]{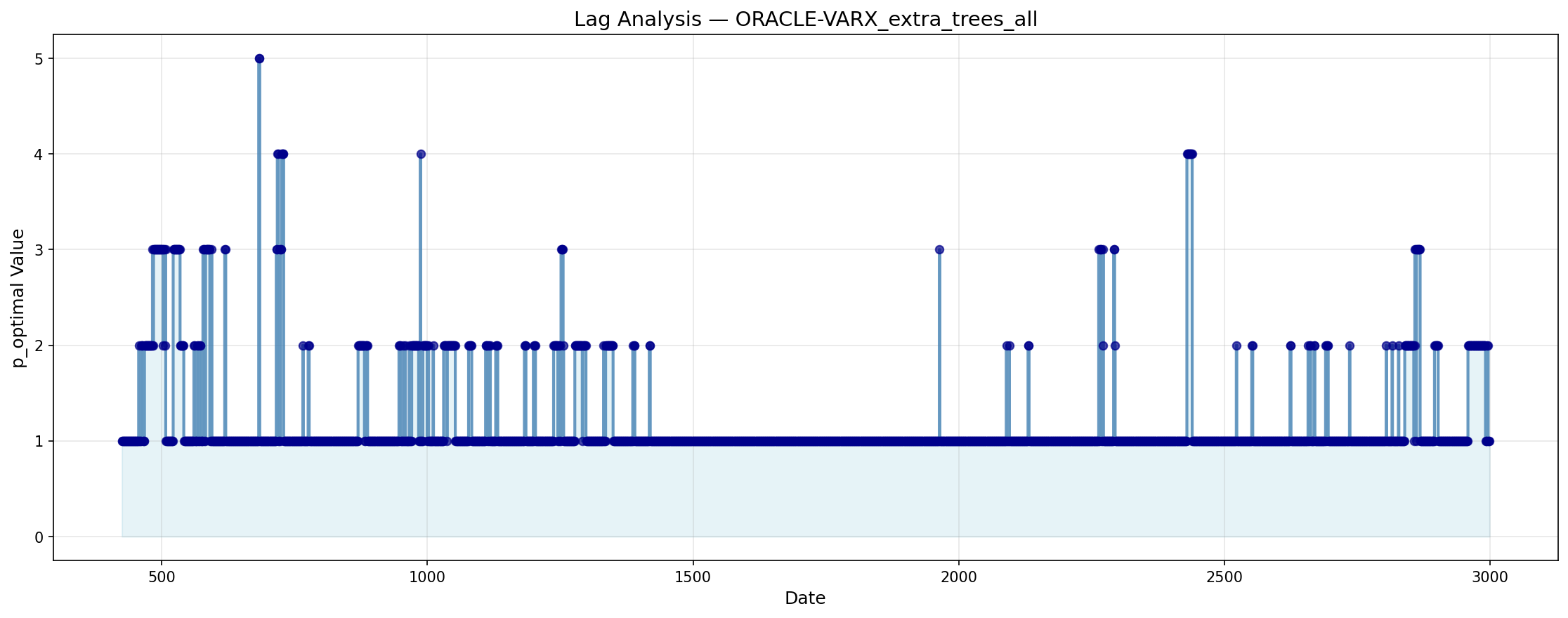}
    \end{minipage}
    \caption{Lag tracking: estimated $\hat{p}_t$ (blue) vs.\ true $p^*(t)$
    (black step function) over time. \emph{Left}: ACLE-VAR (no confounders).
    \emph{Right}: ORACLE-VARX (Extra Trees, all confounders). The lag RMSE of every method is in
    Table~\ref{tab:full-results}.}\label{fig:lag-recovery}
\end{figure}

\vspace{-2mm}

\subsection{Financial ETF Application}\label{sec:etf}

\vspace{-2mm}
\paragraph{Data.}
We apply ORACLE-VARX to daily open-to-close returns of nine U.S.\ sector ETFs
(XLY, XLP, XLE, XLF, XLV, XLI, XLB, XLK, XLU), covering January 2000
to December 2020 ($T \approx 5{,}200$ trading days). Ten macroeconomic
confounders serve as exogenous variables: VIX, the Federal Funds Rate,
5-Year inflation expectations, WTI crude oil, the Economic Policy
Uncertainty index, BBB corporate spreads, 10-Year real yield,
the broad U.S.\ Dollar Index, the emerging market Dollar Index, and
gold volatility. The model is estimated with a rolling window of
504~days ($\approx 2$~years) and $p_{\max} = 10$; see
Table~\ref{tab:config} in Appendix~\ref{app:implementation} for full details.

\vspace{-2mm}

\paragraph{Causal structure discovery.}
Unlike the synthetic benchmark, no ground truth is available for the ETF
application. We therefore focus on qualitative properties of the discovered
causal structure.

\emph{The detected lag order tracks market regimes.}
Figure~\ref{fig:etf-lag} plots the selected lag order $\hat{p}_t$
against the 21-day realized volatility of the S\&P~500 (SPY), the
standard deviation of its daily returns over the past 21 days. In calm markets (e.g., 2004--2006, 2013--2014),
$\hat{p}_t \approx 1$: lead--lag relationships are weak and
short-lived. During the 2008 financial crisis and the COVID-19
shock in early 2020, $\hat{p}_t$ spikes to 3--5: the lag structure
elongates as shocks propagate through the financial system, consistent with the
well-documented increase in cross-market equity correlations in bear
markets~\citep{longin2001extreme}.

\begin{figure}[t]
    \centering
    \includegraphics[width=\linewidth]{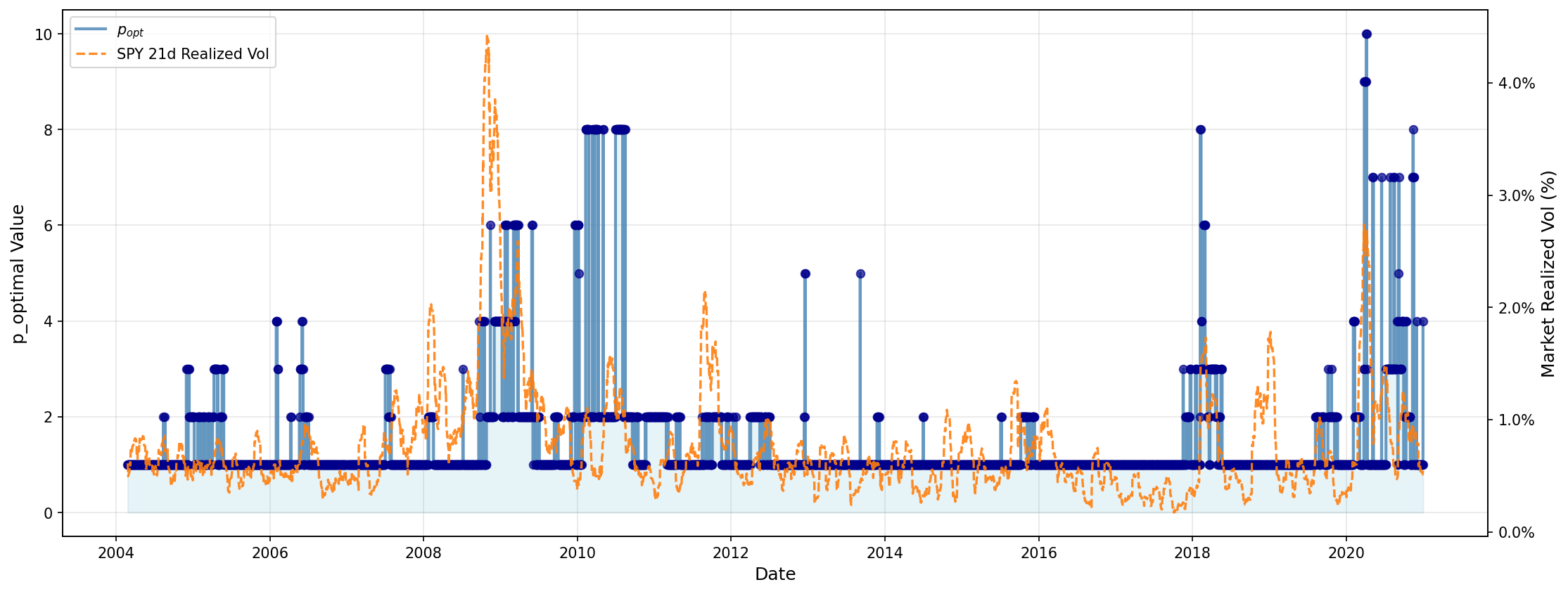}
    \caption{ETF application (ORACLE-VARX, Extra Trees, \texttt{all10}
    confounders). Blue step function: adaptively selected lag order
    $\hat{p}_t$ over time (2004--2020). Orange dashed line: 21-day
    realized volatility of the S\&P~500 (right axis). The lag order spikes
    during the 2008 financial crisis and COVID-2020, rising with the
    volatility.}\label{fig:etf-lag}
\end{figure}

\vspace{-1mm}
\paragraph{Portfolio performance.}
Table~\ref{tab:etf-sharpe} reports Sharpe ratios of a long-short
portfolio weighted by the forecasts (Appendix~\ref{app:etf-trading}).
ACLE generally raises the Sharpe ratio, and the best portfolio comes from
ACLE-VAR, with neither DML nor confounders. On these data, causal
sufficiency is badly violated: each sector ETF pools many stocks driven by
many common factors, most of them unobserved, and, as in Figure~\ref{fig:causal-suff}b, they create lag links that are
not causal but still predict. ACLE tunes its lag order on validation
forecasts, so it keeps these links with a short lag order. DML, in turn,
no longer yields causal coefficients once sufficiency fails, and its
first stage adds parameters to fit from 504 noisy days.

\begin{table}[t]
    \centering
    \caption{ETF Sharpe ratios (weighted strategy, 2004--2020) by
    confounder set; VAR and ACLE-VAR use none.}\label{tab:etf-sharpe}
    \small
    \begin{tabular}{@{}lccc@{}}
        \toprule
        \textbf{Method} & \texttt{vix} & \texttt{macro5} & \texttt{all10} \\
        \midrule
        VAR & \multicolumn{3}{c}{1.08} \\
        ACLE-VAR & \multicolumn{3}{c}{\textbf{1.32}} \\
        \addlinespace
        VARX & 1.11 & 0.43 & 0.75 \\
        ACLE-VARX & 1.07 & 0.81 & 0.86 \\
        \addlinespace
        OR-VARX, ET & 0.71 & 0.71 & 0.67 \\
        ORACLE-VARX, ET & 0.83 & 0.85 & 0.68 \\
        ORACLE-VARX, TP & 0.70 & 0.74 & 0.36 \\
        \bottomrule
    \end{tabular}
\end{table}

\vspace{-1mm}

\Needspace*{6\baselineskip}
\section{Conclusion}\label{sec:conclusion}
\vspace{-1mm}

We introduced ORACLE-VARX, a framework for causal lag discovery in
multivariate time series whose confounders act nonlinearly and whose lag
structure drifts over time. It combines (i) double machine learning, which
removes the confounders' nonlinear effects before the VAR coefficients are
estimated, (ii) ACLE, which picks the lag order in each rolling window by
sequential significance testing tuned on validation data, and (iii)
Benjamini--Hochberg edge selection. We proved that in each rolling window
the debiased coefficients are asymptotically normal around a
window-averaged target, with valid standard errors, even though the
coefficients drift; under causal sufficiency, a zero coefficient means
Granger non-causality. We validated it on synthetic data with known
ground truth and on nine U.S.\ sector ETFs.

\vspace{-2mm}
\section{Future Work}\label{sec:futurework}

Three directions stand out. First, a formal consistency result for
ACLE, covering the BH step and the validation-chosen level
(Section~\ref{sec:theory}). Second, weakening causal sufficiency: hidden
confounders turn lag coefficients into predictive rather than causal
quantities (Figure~\ref{fig:causal-suff}), so richer, high-dimensional
confounder sets, or tests that flag hidden confounding, would let more of
the discovered edges be read causally. Third, extending the model to a structural VAR (SVAR) with
contemporaneous effects~\citep{lutkepohl2005new} could distinguish direct
lagged effects from effects transmitted through same-time links, at the
cost of extra identifying assumptions.

\clearpage
\subsection*{AI use statement}

In this work, we used generative AI tools to state the main theorem and its assumptions precisely and to write its proof (Appendix~\ref{app:proofs}), to implement the method and the experiments, and to give feedback on the experimental design and its engineering, such as reusing repeated model fits.
We have not used generative AI tools to propose the research question or the method, to generate the synthetic data (the authors wrote the functions that simulate it), or to clean the financial data (done before this project used coding agents), and translation and qualitative data analysis are not applicable to this work.
We asked AI tools to comment on the results, but their comments listed details without picking out the main story, so the interpretations in the paper are our own.
Additionally, we used generative AI tools to draft parts of the paper, which the authors then heavily rewrote, to create Figures~\ref{fig:running-dag}, \ref{fig:acle-dag} and~\ref{fig:causal-suff}, to edit the text for readability, and to find related work and check every reference against its source.
We have reviewed all AI-assisted work.
One author, trained in mathematical analysis, checked every step and bound of the proofs by hand; Appendix~\ref{app:ai-use} describes this workflow with an example. The authors read the main mathematical and statistical functions of the code and checked that the outputs are correct or as expected.
We take responsibility for the final content of this work, including text, claims, or artifacts produced with the aid of generative AI.

\bibliography{references}

\clearpage
\appendix
\counterwithin*{equation}{section}\renewcommand{\theequation}{\thesection.\arabic{equation}}
\counterwithin*{theorem}{section}\renewcommand{\thetheorem}{\thesection.\arabic{theorem}}
\thispagestyle{empty}

\onecolumn
\aistatstitle{Causal Lag Structure Discovery in Confounded Time Series\\via Orthogonalized Adaptive Estimation:\\Supplementary Materials}

Our code, the synthetic and public data, and the per-run results are available at \url{https://github.com/HK-Tan/ORACLE-VARX}. Two licensed inputs, the CRSP returns and the ICE BofA spread, are not included; the repository's \texttt{README.md} explains how to obtain them.
\section{Full Proofs}\label{app:proofs}

\subsection{Setting, notation and assumptions}\label{app:assumptions}

\paragraph{Roadmap.}
This subsection fixes the notation and states the assumptions used in all later proofs. It then proves the basic facts that every later proof uses:
\begin{itemize}[nosep]
\item Lemma~\ref{lem:A-moments} bounds the moments of every random vector that appears. It comes first because the other lemmas need its integrability and its constant $K_J$.
\item Lemma~\ref{lem:A-identities} rewrites the cross-fitted residuals as ``true residual plus nuisance error''. This is the algebra behind every later decomposition.
\item Lemma~\ref{lem:A-target} collects the facts we need about the target $\bar\bTheta$.
\item \emph{Limits} fixes the notation $o_p$, $O_p$ and $\convd$; standard matrix and probability rules are named where they are used.
\item Lemma~\ref{lem:A-lln} is the one law of large numbers used throughout.
\end{itemize}

\paragraph{Data and filtration.}
We observe a time series $(\bY_t,\bW_t)$, $t=1,\ldots,T$, of length $T$. Here $\bY_t\in\R^n$ collects the $n$ endogenous series, whose lagged effects on each other we estimate, and $\bW_t\in\R^d$ collects the $d$ observed exogenous confounders (Section~\ref{sec:setup}). All limits are as $T\to\infty$. The distribution may change with $T$ (a \emph{triangular array}); we suppress $T$ in the notation. Write $X_t:=(\bY_t^\top,\bW_t^\top)^\top$. Let $(\mathcal F_t)$ be a filtration to which the data are adapted, i.e.\ $X_t$ is $\mathcal F_t$-measurable for every $t$. Think of $\mathcal F_t$ as the information available at time $t$. For example, $\mathcal F_t=\sigma(X_s:s\le t)$, the observed past, as in the main text.

\paragraph{The model.}
The maximal lag $p:=p_{\max}$ is fixed. For $t>p$, the endogenous series follow the partially linear VARX \eqref{eq:plvar} of the main text, written with all $p$ lags and with coefficients that may drift over time:
\begin{equation}\label{eq:A-model-sum}
  \bY_t=\sum_{k=1}^p\bA_k(t/T)\,\bY_{t-k}+g_t(\bW_{t-1},\ldots,\bW_{t-p})+\beps_t .
\end{equation}
The three terms are:
\begin{itemize}[nosep]
\item $\bA_k(t/T)\in\R^{n\times n}$, the \emph{lag-$k$ coefficient matrix}. Its entry $\bA_k[i,j]$ is the effect of $(\bY_{t-k})_j$ on $(\bY_t)_i$. These matrices are what we estimate. Theorem~\ref{thm:dml-var} does not need a true lag order: lags above it simply have $\bA_k=0$.
\item $g_t$, an unknown and possibly nonlinear effect of the lagged confounders.
\item $\beps_t\in\R^n$, the innovation. Assumption (M) below says that it is unpredictable from the past.
\end{itemize}
\emph{Why $\bA_k(t/T)$.} The coefficients may drift over time. We index the drift by \emph{rescaled time} $u=t/T\in[0,1]$: each $\bA_k$ is a function $\bA_k:[0,1]\to\R^{n\times n}$, and the coefficient at time $t$ is $\bA_k(t/T)$. As $T$ grows, neighbouring times $t$ and $t+1$ differ by only $1/T$ in $u$, so the coefficients change slowly from step to step.

\paragraph{Block form of the model.}
We now write \eqref{eq:A-model-sum} as one regression. Set the \emph{treatment} $\bT_t$ (the stacked lags of $\bY$, whose effects we estimate) and the \emph{control} $\bC_t$ (the stacked lags of the confounders $\bW$) as in \eqref{eq:stacked} of the main text:
\begin{equation}\label{eq:A-TC}
  \bT_t:=\begin{pmatrix}\bY_{t-1}\\ \vdots\\ \bY_{t-p}\end{pmatrix}\in\R^{np},
  \qquad
  \bC_t:=\begin{pmatrix}\bW_{t-1}\\ \vdots\\ \bW_{t-p}\end{pmatrix}\in\R^{dp}.
\end{equation}
Both are functions of $X_{t-1},\ldots,X_{t-p}$, hence $\mathcal F_{t-1}$-measurable. Stack the $p$ coefficient matrices into one block matrix,
\begin{equation}\label{eq:A-Theta}
  \bTheta(u):=\begin{pmatrix}\bA_1(u)^\top\\ \vdots\\ \bA_p(u)^\top\end{pmatrix}\in\R^{np\times n},
  \qquad \bTheta_t:=\bTheta(t/T).
\end{equation}
Block $k$ of $\bTheta$ (rows $(k-1)n+1,\ldots,kn$) is $\bA_k^\top$; we call it \emph{lag block $k$}. In entries, $\bTheta_{(k-1)n+j,\,i}=\bA_k[i,j]$. Block multiplication of \eqref{eq:A-Theta} and \eqref{eq:A-TC} gives the sum in \eqref{eq:A-model-sum} as one product,
\begin{equation}\label{eq:A-VARpart}
  \bTheta_t^\top\bT_t
  =\begin{pmatrix}\bA_1(t/T)&\cdots&\bA_p(t/T)\end{pmatrix}
   \begin{pmatrix}\bY_{t-1}\\ \vdots\\ \bY_{t-p}\end{pmatrix}
  =\sum_{k=1}^p\bA_k(t/T)\,\bY_{t-k}.
\end{equation}
Writing $g_t(\bC_t)$ for $g_t(\bW_{t-1},\ldots,\bW_{t-p})$, the model \eqref{eq:A-model-sum} becomes
\begin{equation}\label{eq:A-model}
  \bY_t=\bTheta_t^\top\bT_t+g_t(\bC_t)+\beps_t ,
\end{equation}
a linear regression on the treatment plus a nonlinear function of the control. This is \eqref{eq:dml-Y} of the main text, with drifting coefficients.

\emph{Equation $i$.} Since $\bY_t\in\R^n$, \eqref{eq:A-model} is $n$ scalar regressions, one per coordinate. \emph{Equation $i$} is coordinate $i$:
\begin{equation}\label{eq:A-eq-i}
  (\bY_t)_i=(\bTheta_t)_{\cdot i}^\top\bT_t+\bigl(g_t(\bC_t)\bigr)_i+(\beps_t)_i,
\end{equation}
where $\bTheta_{\cdot i}\in\R^{np}$ is \emph{column} $i$ of $\bTheta$, the coefficients of all $np$ regressors for the outcome $(\bY_t)_i$. So in an entry $\bTheta_{ri}$, the column $i$ names the equation (which series is the outcome), and the row $r=(k-1)n+j$ names the regressor (series $j$ at lag $k$). All $n$ equations share the same regressors $\bT_t$.

\paragraph{Nuisance functions and residuals.}
The nuisance functions $g_{T,t}:\R^{dp}\to\R^{np}$ and $g_{Y,t}:\R^{dp}\to\R^n$ are the best predictions of $\bT_t$ and $\bY_t$ from the controls alone, as $g_T$ and $g_Y$ in Section~\ref{sec:setup}:
\begin{equation}\label{eq:A-nuisance}
  g_{T,t}(\bC_t)=\E[\bT_t\mid\bC_t],\qquad g_{Y,t}(\bC_t)=\E[\bY_t\mid\bC_t]\qquad\text{almost surely}.
\end{equation}
Both conditional expectations exist because $\bT_t$ and $\bY_t$ are integrable under (D)(i) below. Each is a function of $\bC_t$ because it depends on the data only through $\bC_t$.

The \emph{oracle regressor residual} is the part of the treatment that the controls cannot predict; it is $\bV_t$ in \eqref{eq:dml-T} of the main text. We define it together with its second-moment matrix:
\begin{equation}\label{eq:A-VJ}
  \bV_t:=\bT_t-g_{T,t}(\bC_t),\qquad J_t:=\E[\bV_t\bV_t^\top]\in\R^{np\times np}.
\end{equation}
The entries of $J_t$ are finite by Lemma~\ref{lem:A-moments} below.

\paragraph{Window.}
Fix $u_0\in(0,1]$. Let $t_0:=\lfloor u_0T\rfloor$ and $\mathcal W:=\{t_0-T_w+1,\ldots,t_0\}$, a block of $T_w$ consecutive times. Unless said otherwise, every sum $\sum_t$ or $\sum_{\mathcal W}$ runs over $t\in\mathcal W$. Two facts about the window are used repeatedly.
\begin{enumerate}[label=(\alph*)]
\item \emph{The window eventually lies in $\{p+1,\ldots,T\}$.} Since $u_0T-1<t_0\le u_0T\le T$,
\begin{equation*}
  t_0-T_w+1>u_0T-T_w=T\Bigl(u_0-\frac{T_w}{T}\Bigr),
\end{equation*}
and $T_w/T\to0$ by (W) (see ``What (W) does'' after the assumptions). Since $u_0>0$, for all $T$ with $T_w/T\le u_0/2$ and $T>2p/u_0$ we get
\begin{equation*}
  t_0-T_w+1>T\Bigl(u_0-\frac{T_w}{T}\Bigr)\ge T\,\frac{u_0}{2}>p .
\end{equation*}
Both conditions hold for all large $T$. All statements below are for such $T$, so $\bT_t$ and $\bC_t$ are defined for every $t\in\mathcal W$.
\item \emph{Every window time is close to $u_0$.} For $t\in\mathcal W$ we have $t_0-T_w+1\le t\le t_0$, so
\begin{equation*}
  0\le u_0T-t_0\le u_0T-t\le u_0T-(t_0-T_w+1)<1+T_w-1=T_w .
\end{equation*}
Dividing by $T$,
\begin{equation}\label{eq:A-window}
  \Bigl\lvert\frac tT-u_0\Bigr\rvert<\frac{T_w}{T}\qquad\text{for every }t\in\mathcal W .
\end{equation}
Also, two window times $s,t$ satisfy $\lvert s-t\rvert\le T_w-1$.
\end{enumerate}

\paragraph{Cross-fitting.}
Each $t\in\mathcal W$ has fitted nuisance functions $\hat g_{T,(t)}:\R^{dp}\to\R^{np}$ and $\hat g_{Y,(t)}:\R^{dp}\to\R^n$, produced by a learner trained on data before the test block that contains $t$. The proofs never use how they were trained; they use only assumption (N) below, and that each $\hat g_{T,(t)}(\bC_t)$ and $\hat g_{Y,(t)}(\bC_t)$ is a random vector (measurable), so that $\hat J$, $\hbTheta$ and the events built from them are well defined. The cross-fitted residuals are the residuals \eqref{eq:residuals} of the main text, with the fold of $t$ made explicit:
\begin{equation}\label{eq:A-resid}
  \tilde\bT_t:=\bT_t-\hat g_{T,(t)}(\bC_t),\qquad \tilde\bY_t:=\bY_t-\hat g_{Y,(t)}(\bC_t),
\end{equation}
and the nuisance errors, evaluated at the data point, are
\begin{equation}\label{eq:A-Delta}
  \Delta_{T,t}:=\hat g_{T,(t)}(\bC_t)-g_{T,t}(\bC_t)\in\R^{np},\quad
  \Delta_{Y,t}:=\hat g_{Y,(t)}(\bC_t)-g_{Y,t}(\bC_t)\in\R^{n},\quad
  \Delta_{g,t}:=\bTheta_t^\top\Delta_{T,t}-\Delta_{Y,t}\in\R^{n}.
\end{equation}

\paragraph{Estimator.}
Recall the estimator \eqref{eq:theta-hat} of the main text, OLS of $\tilde\bY_t$ on $\tilde\bT_t$ over the window:
\[
  \hbTheta=\Bigl(\sum_{t\in\mathcal W}\tilde\bT_t\tilde\bT_t^\top\Bigr)^{-1}\Bigl(\sum_{t\in\mathcal W}\tilde\bT_t\tilde\bY_t^\top\Bigr).
\]
Dividing both sums by $T_w$ (the factors cancel) gives the form we use. Let
\begin{equation}\label{eq:A-thetahat}
  \hat J:=\frac1{T_w}\sum_{t\in\mathcal W}\tilde\bT_t\tilde\bT_t^\top,\qquad
  \hbTheta:=\hat J^{-1}\,\frac1{T_w}\sum_{t\in\mathcal W}\tilde\bT_t\tilde\bY_t^\top\quad\text{on the event }\{\hat J\text{ invertible}\}.
\end{equation}
If $\hat J$ is not invertible (for example, if the residualised lags happen to be exactly collinear in this sample), the formula has no value. To speak of probabilities and limits of $\hbTheta$, it must be a random variable, that is, have a value at every outcome, so there we set $\hbTheta:=0$. This is a placeholder, not an estimate, and any other value would do: $P(\hat J\text{ invertible})\to1$ (Lemma~\ref{lem:B-gram}), so the placeholder never affects a limit.

\paragraph{Standard errors and $z$-statistics.}
Recall the standard errors \eqref{eq:ols-se} of the main text, which are what the code computes: for equation $i$, $\widehat{\Var}(\hbTheta_{\cdot i})=\hat\sigma_i^2\bigl(\sum_t\tilde\bT_t\tilde\bT_t^\top\bigr)^{-1}$. In detail, let $\hat e_t:=\tilde\bY_t-\hbTheta^\top\tilde\bT_t\in\R^n$ be the residual and, for $T_w>np$ (true for all large $T$, since $T_w\to\infty$ by (W)),
\begin{equation}\label{eq:A-sigmahat}
  \hat\sigma_i^2:=\frac1{T_w-np}\sum_{t\in\mathcal W}(\hat e_t)_i^2 .
\end{equation}
On the event $\{\hat J\text{ invertible}\}$ we have $\sum_{t\in\mathcal W}\tilde\bT_t\tilde\bT_t^\top=T_w\hat J$, hence $\bigl(\sum_{t\in\mathcal W}\tilde\bT_t\tilde\bT_t^\top\bigr)^{-1}=\hat J^{-1}/T_w$, and
\begin{equation}\label{eq:A-se}
  \hat s_{ri}:=\Bigl(\hat\sigma_i^2\,\Bigl[\Bigl(\sum_{t\in\mathcal W}\tilde\bT_t\tilde\bT_t^\top\Bigr)^{-1}\Bigr]_{rr}\Bigr)^{1/2}
  =\Bigl(\frac{\hat\sigma_i^2\,[\hat J^{-1}]_{rr}}{T_w}\Bigr)^{1/2},
  \qquad
  z_{ri}:=\frac{\hbTheta_{ri}}{\hat s_{ri}} .
\end{equation}
This is the usual OLS standard error for column $i$ of $\hbTheta$, i.e.\ for equation $i$ in \eqref{eq:A-eq-i}: writing $X$ for the $T_w\times np$ matrix with rows $\tilde\bT_t^\top$, so that $X^\top X=\sum_{t\in\mathcal W}\tilde\bT_t\tilde\bT_t^\top$, it is the square root of entry $(r,r)$ of $\hat\sigma_i^2(X^\top X)^{-1}$. We do not assume that it is correct here; Lemma~\ref{lem:C-se} shows that it has the right scale.
(If $\hat J$ is singular or $\hat s_{ri}=0$, set $\hat s_{ri}:=1$ and $z_{ri}:=0$. Lemmas~\ref{lem:B-gram} and~\ref{lem:C-se} show that this event has probability tending to $0$, so the convention never affects a limit.)

\paragraph{Target.}
Recall the window target \eqref{eq:target} of the main text, $\bar\bTheta=\bigl(\sum_{t\in\mathcal W}J_t\bigr)^{-1}\sum_{t\in\mathcal W}J_t\bTheta_t$, a $J_t$-weighted average of the coefficients over the window. Dividing both sums by $T_w$ gives the form we use:
\begin{equation}\label{eq:A-target}
  \bar J:=\frac1{T_w}\sum_{t\in\mathcal W}J_t,\qquad
  \bar\bTheta:=\bar J^{-1}\,\frac1{T_w}\sum_{t\in\mathcal W}J_t\bTheta_t,\qquad
  D_t:=\bTheta_t-\bar\bTheta\in\R^{np\times n}.
\end{equation}
Lemma~\ref{lem:A-target}(i) shows that $\bar J$ is invertible, so $\bar\bTheta$ is well defined. Both $\bar\bTheta$ and $D_t$ are deterministic.

\paragraph{Norms.}
$\lvert\cdot\rvert$ is the Euclidean norm of a vector, $\lVert\cdot\rVert$ the operator norm of a matrix, and $\lvert\cdot\rvert_F$ the Frobenius norm. $e_r$ is the $r$-th unit vector, and $\lambda_{\min}(M)$ is the smallest eigenvalue of a symmetric matrix $M$. For a random vector $Z$ we write $M_8(Z):=(\E\lvert Z\rvert^8)^{1/8}$.

\paragraph{Assumptions.}
Assumption~\ref{asm:main} of the main text consists of (M), (S), (W) and the regularity conditions (D), (H), (N). They are stated precisely here.
\begin{enumerate}[leftmargin=3.2em]
\item[\textbf{(M)}] \emph{Model.} For $p<t\le T$, the model \eqref{eq:A-model} holds (equivalently \eqref{eq:A-model-sum}), where $g_t:\R^{dp}\to\R^n$ is Borel, $\beps_t$ is integrable, and $\E[\beps_t\mid\mathcal F_{t-1}]=\bzero$ a.s. In words: given the past, the innovation is unpredictable.
\item[\textbf{(S)}] \emph{Smoothness and identification.} There are constants $L<\infty$ and $c_J>0$, not depending on $t$ or $T$, such that (i) $\bTheta$ is $L$-Lipschitz: $\lVert\bTheta(u)-\bTheta(v)\rVert\le L\lvert u-v\rvert$ for all $u,v\in[0,1]$; and (ii) every $J_t$ is uniformly positive definite: $\lambda_{\min}(J_t)\ge c_J$ for all $t$ and $T$.
\item[\textbf{(W)}] \emph{Window size.} $T_w\to\infty$ and $T_w^{3/2}/T=O(1)$, i.e.\ $T_w=O(T^{2/3})$.
\item[\textbf{(D)}] \emph{Moments and short memory.} (i) $K_X:=\sup_T\sup_t\E\lvert X_t\rvert^8<\infty$. (ii) For some constant $C_D<\infty$,
\begin{equation}\label{eq:A-shortmem}
\begin{aligned}
  \sup_T\,\sup_t\,\max_{1\le a,b\le np}\ \sum_{h\ge0}\Bigl\lvert\Cov\bigl((\bV_t)_a(\bV_t)_b,\ (\bV_{t+h})_a(\bV_{t+h})_b\bigr)\Bigr\rvert&\le C_D,\\
  \sup_T\,\sup_t\,\max_{1\le a,b\le n}\ \sum_{h\ge0}\Bigl\lvert\Cov\bigl((\beps_t)_a(\beps_t)_b,\ (\beps_{t+h})_a(\beps_{t+h})_b\bigr)\Bigr\rvert&\le C_D.
\end{aligned}
\end{equation}
\item[\textbf{(H)}] \emph{Conditional variance.} $\E[\beps_t\beps_t^\top\mid\mathcal F_{t-1}]=\Sigma_\varepsilon(t/T)$ a.s. for all $p<t\le T$, where $\Sigma_\varepsilon:[0,1]\to\R^{n\times n}$ is deterministic, $\lVert\Sigma_\varepsilon(u)-\Sigma_\varepsilon(v)\rVert\le L_\varepsilon\lvert u-v\rvert$, and $\lambda_{\min}(\Sigma_\varepsilon(u))\ge c_\varepsilon>0$ for all $u,v\in[0,1]$.
\item[\textbf{(N)}] \emph{Nuisance accuracy.} (N1) The learners are consistent, with a product rate:
\begin{align}
  &\frac1{T_w}\sum_{\mathcal W}\lvert\Delta_{T,t}\rvert^2\convp0, \label{eq:N1a}\\
  &\frac1{T_w}\sum_{\mathcal W}\lvert\Delta_{Y,t}\rvert^2\convp0, \label{eq:N1b}\\
  &\sqrt{T_w}\,\Bigl(\frac1{T_w}\sum_{\mathcal W}\lvert\Delta_{T,t}\rvert^2\Bigr)^{1/2}\Bigl(\frac1{T_w}\sum_{\mathcal W}\lvert\Delta_{g,t}\rvert^2\Bigr)^{1/2}\convp0 . \label{eq:N1c}
\end{align}
(N2) The learner errors are asymptotically uncorrelated with the noise and with $\bV_t$:
\begin{align}
  &\frac1{\sqrt{T_w}}\sum_{\mathcal W}\Delta_{T,t}\beps_t^\top\convp0, \label{eq:N2a}\\
  &\frac1{\sqrt{T_w}}\sum_{\mathcal W}\bV_t\Delta_{g,t}^\top\convp0 . \label{eq:N2b}
\end{align}
\end{enumerate}

\emph{Why (S)(ii) is needed.} It says that no combination of the lags is a function of the confounders. Suppose instead that $x^\top\bT_t=h(\bC_t)$ for some $x\ne0$ and some function $h$. Then for any $b\in\R^n$, replacing $\bTheta_t$ by $\bTheta_t+xb^\top$ and $g_t$ by $g_t-b\,h$ leaves \eqref{eq:A-model} unchanged, since $(xb^\top)^\top\bT_t=b\,x^\top\bT_t=b\,h(\bC_t)$. So $\bTheta_t$ could not be recovered from the data. This is exactly the case $x^\top J_tx=\E(x^\top\bV_t)^2=0$. Requiring $\lambda_{\min}(J_t)\ge c_J>0$ also rules out the nearly collinear case, uniformly in $t$ and $T$.

\emph{What (W) does.} The window must be long enough to average out the noise, and short enough that the coefficients barely move inside it.
\begin{itemize}[nosep]
\item $T_w\to\infty$ gives enough terms for the laws of large numbers and the central limit theorem.
\item $T_w/T\to0$: the window is a vanishing fraction of the series, so it shrinks to the single point $u_0$ in rescaled time. This follows from (W) because $T_w/T=(T_w^{3/2}/T)\cdot T_w^{-1/2}$ is a bounded factor times a factor that tends to $0$.
\item $T_w^{3/2}/T=O(1)$ is where the exponent $2/3$ comes from. Across the window, $u=t/T$ ranges over an interval of length $T_w/T$, so by (S)(i) the coefficients drift by $O(T_w/T)$. The estimation noise is of order $T_w^{-1/2}$. The drift is at most of the order of the noise exactly when $(T_w/T)\big/T_w^{-1/2}=T_w^{3/2}/T$ stays bounded. With $T_w=o(T^{2/3})$ the drift is negligible, which is why Theorem~\ref{thm:dml-var} can then be centred at $\bTheta(u_0)$.
\end{itemize}

\emph{What (D)(ii) does.} It bounds the variance of window sums of the entries of $\bV_t\bV_t^\top$ and $\beps_t\beps_t^\top$. For independent terms the variance of a sum of $T_w$ terms grows like $T_w$ automatically; for dependent terms it need not, and (D)(ii) is what restores this. With it, Chebyshev's inequality gives a law of large numbers for the window averages (Lemma~\ref{lem:A-lln}), which in turn gives $\hat J\approx\bar J$ and $\hat\sigma_i^2\approx\bar\sigma_i^2$. Fix one entry, say $f_t:=(\bV_t)_1(\bV_t)_2$ for all $t$, so that \eqref{eq:A-shortmem} gives $\sum_{h\ge0}\lvert\Cov(f_t,f_{t+h})\rvert\le C_D$. The computation is one line: expanding the square and counting each pair $s\le t$ once,
\[
  \Var\Bigl(\sum_{t\in\mathcal W}f_t\Bigr)=\sum_{s,t\in\mathcal W}\Cov(f_s,f_t)\le2\sum_{s\in\mathcal W}\sum_{h\ge0}\lvert\Cov(f_s,f_{s+h})\rvert\le2C_DT_w ,
\]
so the window average has variance at most $2C_D/T_w$, and by Chebyshev $P\bigl(\lvert T_w^{-1}\sum_t(f_t-\E f_t)\rvert>\eta\bigr)\le2C_D/(\eta^2T_w)\to0$ for every $\eta>0$. It is called \emph{short memory} because the covariances are summable in the lag $h$: values far apart in time are nearly uncorrelated, so the series ``forgets'' its past quickly. (Under \emph{long memory} the sum is infinite, and averages converge more slowly or not at all.)

\emph{Notation in (H).} The subscript $\varepsilon$ in $\Sigma_\varepsilon$, $L_\varepsilon$, $c_\varepsilon$ is only a label, ``of the noise $\beps_t$''. The only argument of $\Sigma_\varepsilon$ is rescaled time $u$. (H) is the analogue of (S) for the noise: its covariance drifts slowly (Lipschitz in $u$) and stays away from singular, and it does not depend on the past.

\emph{What (N) does.} (N) is the only place where the machine-learning step enters. Recall that $\Delta_{T,t}$ and $\Delta_{Y,t}$ in \eqref{eq:A-Delta} are the learners' errors at the data point, and $\Delta_{g,t}=\bTheta_t^\top\Delta_{T,t}-\Delta_{Y,t}$ is the error they cause in the outcome equation. With perfect learners (all $\Delta=0$), the estimation error would be driven by the main term $S_0=T_w^{-1/2}\sum_t\bV_t\beps_t^\top$ alone, plus the drift of $\bTheta_t$. With learner errors, $\tilde\bT_t=\bV_t-\Delta_{T,t}$ and $\tilde\bY_t$ picks up $\Delta_{g,t}$ (Lemma~\ref{lem:A-identities}), and multiplying out creates three extra sums (Lemma~\ref{lem:B-decomp}):
\[
  \frac1{\sqrt{T_w}}\sum_t\Delta_{T,t}\beps_t^\top,\qquad
  \frac1{\sqrt{T_w}}\sum_t\bV_t\Delta_{g,t}^\top,\qquad
  \frac1{\sqrt{T_w}}\sum_t\Delta_{T,t}\Delta_{g,t}^\top .
\]
\begin{itemize}[nosep]
\item \eqref{eq:N2a} and \eqref{eq:N2b} in (N2) say that the first two tend to $0$. Each pairs a learner error with a mean-zero term ($\beps_t$ or $\bV_t$), so it is small when the errors are asymptotically uncorrelated with the noise and with $\bV_t$.
\item \eqref{eq:N1c} in (N1) handles the third sum, a product of two errors. By Cauchy--Schwarz its size is at most $\sqrt{T_w}$ times the root-mean-square of $\Delta_{T,\cdot}$ times that of $\Delta_{g,\cdot}$. This is the usual DML \emph{product rate}: it holds, for example, if both root-mean-square errors are $o(T_w^{-1/4})$, much slower than the parametric rate $T_w^{-1/2}$. Orthogonalization is what makes the learner errors enter only through such products.
\item \eqref{eq:N1a} and \eqref{eq:N1b} in (N1) say that each learner is consistent on average over the window: its mean squared error at the data points tends to $0$. It is used for $\hat J\approx\bar J$, for $\hat\sigma_i^2\approx\bar\sigma_i^2$, and for the drift terms.
\end{itemize}
(N) is a condition on the learners as they are actually run; we do not verify it for particular learners (Remark~\ref{rem:buffer}).

\emph{Constants.} $C,C',\ldots$ denote constants that depend only on the constants in the assumptions (namely $L$, $c_J$, $K_X$, $C_D$, $L_\varepsilon$, $c_\varepsilon$ and $\lVert\bTheta(0)\rVert$) and on $n,d,p$; never on $T$ or $t$.

\emph{Conditional expectations.} We use the standard rules (the tower property, taking out what is known, conditional Jensen) and name each one where it is used. For random vectors they apply entry by entry. Conditioning on a random vector $\bC$ means conditioning on $\sigma(\bC)$.

\begin{lemma}[Moments]\label{lem:A-moments}
Assume (M) and (D)(i). Recall from (D)(i) that $K_X=\sup_T\sup_t\E\lvert X_t\rvert^8<\infty$ bounds the eighth moment of the data $X_t=(\bY_t^\top,\bW_t^\top)^\top$. Put $K_8:=2^8p^8K_X$ and $K_J:=K_8^{1/4}=4p^2K_X^{1/4}$. For all $T$ and all $p<t\le T$,
\begin{equation*}
  \E\lvert\bT_t\rvert^8\le K_8,\qquad \E\lvert\bV_t\rvert^8\le K_8,\qquad \E\lvert\beps_t\rvert^8\le K_8,\qquad \lVert J_t\rVert\le K_J .
\end{equation*}
Consequently, for $Z\in\{\bT_t,\bV_t,\beps_t\}$ and $0<q\le8$, $\E\lvert Z\rvert^q\le K_8^{q/8}$; in particular $\E\lvert Z\rvert^2\le K_J$ and $\E\lvert Z\rvert^4\le K_8^{1/2}$. If (S)(ii) also holds, then every eigenvalue of $J_t$ lies in $[c_J,K_J]$: the lower bound is assumed and the upper bound comes from this lemma. In particular $c_J\le K_J$.
\end{lemma}

In words, every random vector in the model has a bounded eighth moment, with one constant for all $t$ and $T$.

\begin{proof}
Recall $M_8(Z):=(\E\lvert Z\rvert^8)^{1/8}$, the $L^8$ norm of $\lvert Z\rvert$. We first record two tools.
\begin{enumerate}[label=(\roman*),nosep]
\item \emph{Minkowski for random vectors.} For random vectors $Z_1,Z_2$, $M_8(Z_1+Z_2)\le M_8(Z_1)+M_8(Z_2)$.
\item \emph{Conditioning does not increase the eighth moment.} If $Z\in\R^m$ has $\E\lvert Z\rvert^8<\infty$ and $\mathcal G$ is a $\sigma$-field, then $\E\lvert\E[Z\mid\mathcal G]\rvert^8\le\E\lvert Z\rvert^8$.
\end{enumerate}

\emph{For (i).} The scalar Minkowski inequality says that for real random variables $U,W$, $(\E\lvert U+W\rvert^8)^{1/8}\le(\E\lvert U\rvert^8)^{1/8}+(\E\lvert W\rvert^8)^{1/8}$; this is the triangle inequality for the norm of $L^8$. For vectors,
\begin{equation*}
  M_8(Z_1+Z_2)\le M_8\bigl(\lvert Z_1\rvert+\lvert Z_2\rvert\bigr)\le M_8(Z_1)+M_8(Z_2).
\end{equation*}
The first step holds because $\lvert Z_1+Z_2\rvert\le\lvert Z_1\rvert+\lvert Z_2\rvert$ pointwise and $x\mapsto x^8$ is increasing on $[0,\infty)$; the second is the scalar inequality with $U=\lvert Z_1\rvert$ and $W=\lvert Z_2\rvert$.

\emph{For (ii).} For a scalar $Z$ this is conditional Jensen with $\phi(x)=x^8$, followed by the tower property. For a vector, $\lvert Z\rvert^8$ is not a function of a single entry, so we cannot apply the scalar rule directly. Instead we go through the squares, which we can handle entry by entry and then add up. Each entry $Z_j$ is square-integrable: $\E Z_j^2\le\E\lvert Z\rvert^2\le(\E\lvert Z\rvert^8)^{1/4}<\infty$, by Jensen's inequality for the concave map $x\mapsto x^{1/4}$. Conditional Jensen with $\phi(x)=x^2$, applied to the scalar $Z_j$, gives $(\E[Z_j\mid\mathcal G])^2\le\E[Z_j^2\mid\mathcal G]$, and summing over $j$,
\begin{equation*}
  \lvert\E[Z\mid\mathcal G]\rvert^2=\sum_{j=1}^m(\E[Z_j\mid\mathcal G])^2\le\sum_{j=1}^m\E[Z_j^2\mid\mathcal G]=\E[\lvert Z\rvert^2\mid\mathcal G].
\end{equation*}
Both sides are nonnegative (the conditional expectation of a nonnegative variable is a.s.\ nonnegative), and $x\mapsto x^4$ is increasing on $[0,\infty)$, so we may raise this to the fourth power. Then conditional Jensen with the convex map $\phi(x)=x^4$ on $[0,\infty)$, applied to the scalar $\lvert Z\rvert^2$, gives
\begin{equation*}
  \lvert\E[Z\mid\mathcal G]\rvert^8\le\bigl(\E[\lvert Z\rvert^2\mid\mathcal G]\bigr)^4\le\E[\lvert Z\rvert^8\mid\mathcal G].
\end{equation*}
Take expectations and use the tower property: $\E\lvert\E[Z\mid\mathcal G]\rvert^8\le\E\,\E[\lvert Z\rvert^8\mid\mathcal G]=\E\lvert Z\rvert^8$.

With these two tools, we now prove the lemma. Steps 1--3 bound the eighth moments of $\bT_t$, $\bV_t$ and $\beps_t$ in turn; Step 4 derives the lower moments; Step 5 bounds $J_t$.

\emph{Step 1: $\bT_t$.} Since $\lvert\bT_t\rvert^2=\sum_{k=1}^p\lvert\bY_{t-k}\rvert^2\le\bigl(\sum_{k=1}^p\lvert\bY_{t-k}\rvert\bigr)^2$ and $\lvert\bY_{t-k}\rvert\le\lvert X_{t-k}\rvert$ for each $k$ (since $\bY_{t-k}$ is part of $X_{t-k}=(\bY_{t-k}^\top,\bW_{t-k}^\top)^\top$),
\begin{equation*}
  \lvert\bT_t\rvert\le\sum_{k=1}^p\lvert X_{t-k}\rvert .
\end{equation*}
Here $t-k\ge1$ because $t>p$. Since $\lvert\bT_t\rvert$ is at most this sum, $M_8(\bT_t)\le M_8\bigl(\sum_k\lvert X_{t-k}\rvert\bigr)$. By (D)(i), $\E\lvert X_s\rvert^8\le K_X$ for every time $s$ (and every $T$), that is, $M_8(X_s)\le K_X^{1/8}$ for every $s$. Applying (i) $p-1$ times and then this bound with $s=t-1,\ldots,t-p$, we get $M_8(\bT_t)\le\sum_{k=1}^pM_8(X_{t-k})\le pK_X^{1/8}$. So $\E\lvert\bT_t\rvert^8\le p^8K_X\le K_8$.

\emph{Step 2: $\bV_t$.} By (ii) with $\mathcal G=\sigma(\bC_t)$, $\E\lvert g_{T,t}(\bC_t)\rvert^8=\E\lvert\E[\bT_t\mid\bC_t]\rvert^8\le\E\lvert\bT_t\rvert^8$. Recall from \eqref{eq:A-VJ} that $\bV_t=\bT_t-g_{T,t}(\bC_t)$. Apply (i) with $Z_1=\bT_t$ and $Z_2=-g_{T,t}(\bC_t)$ (note $M_8(-Z)=M_8(Z)$) for the first inequality below; the bound just shown for the second; and Step 1 for the third:
\begin{equation*}
  M_8(\bV_t)\le M_8(\bT_t)+M_8\bigl(g_{T,t}(\bC_t)\bigr)\le2M_8(\bT_t)\le2pK_X^{1/8},
\end{equation*}
so $\E\lvert\bV_t\rvert^8\le2^8p^8K_X=K_8$.

\emph{Step 3: $\beps_t$.} First, $g_t(\bC_t)=\bY_t-\bTheta_t^\top\bT_t-\beps_t$ by (M), and each term on the right is integrable ($\bY_t$ and $\bT_t$ by (D)(i) and Step 1, $\beps_t$ by (M)); so $g_t(\bC_t)$ is integrable. Now take $\E[\,\cdot\mid\mathcal F_{t-1}]$ of \eqref{eq:A-model}. The vector $\bTheta_t^\top\bT_t+g_t(\bC_t)$ is $\mathcal F_{t-1}$-measurable ($\bTheta_t$ is deterministic, $g_t$ is Borel and $\bT_t,\bC_t$ are $\mathcal F_{t-1}$-measurable), so by taking out what is known it is its own conditional expectation; and $\E[\beps_t\mid\mathcal F_{t-1}]=\bzero$ by (M). Hence
\begin{equation*}
  \E[\bY_t\mid\mathcal F_{t-1}]=\bTheta_t^\top\bT_t+g_t(\bC_t),
  \qquad\text{so}\qquad
  \beps_t=\bY_t-\E[\bY_t\mid\mathcal F_{t-1}]\quad\text{a.s.}
\end{equation*}
As in Step 2: apply (i) with $Z_1=\bY_t$ and $Z_2=-\E[\bY_t\mid\mathcal F_{t-1}]$ for the first inequality below; (ii) with $\mathcal G=\mathcal F_{t-1}$ for the second; and $\lvert\bY_t\rvert\le\lvert X_t\rvert$ with (D)(i) for the third:
\begin{equation*}
  M_8(\beps_t)\le M_8(\bY_t)+M_8\bigl(\E[\bY_t\mid\mathcal F_{t-1}]\bigr)\le2M_8(\bY_t)\le2K_X^{1/8},
\end{equation*}
so $\E\lvert\beps_t\rvert^8\le2^8K_X\le K_8$ (as $p\ge1$).

\emph{Step 4: lower moments.} For $0<q\le8$ the map $x\mapsto x^{q/8}$ is concave on $[0,\infty)$, so Jensen's inequality gives $\E\lvert Z\rvert^q=\E\bigl[(\lvert Z\rvert^8)^{q/8}\bigr]\le(\E\lvert Z\rvert^8)^{q/8}\le K_8^{q/8}$. With $q=2$ and $q=4$ this is $K_8^{1/4}=K_J$ and $K_8^{1/2}$.

\emph{Step 5: $J_t$.} Its entries are finite, since $\lvert\E[(\bV_t)_l(\bV_t)_b]\rvert\le\E\lvert\bV_t\rvert^2\le K_J$ by Step 4. $J_t$ is symmetric and positive semidefinite ($x^\top J_tx=\E(x^\top\bV_t)^2\ge0$). For such a matrix the operator norm is the largest eigenvalue, $\lVert J_t\rVert=\max_{\lvert x\rvert=1}x^\top J_tx$. For $\lvert x\rvert=1$, Cauchy--Schwarz gives $(x^\top\bV_t)^2\le\lvert x\rvert^2\lvert\bV_t\rvert^2=\lvert\bV_t\rvert^2$, so
\begin{equation*}
  x^\top J_tx=\E(x^\top\bV_t)^2\le\E\lvert\bV_t\rvert^2\le K_J .
\end{equation*}
Hence $\lVert J_t\rVert\le K_J$. Finally, by (S)(ii), $c_J\le\lambda_{\min}(J_t)\le\lVert J_t\rVert\le K_J$.
\end{proof}

\begin{lemma}[Residual identities and a martingale difference]\label{lem:A-identities}
Assume (M), (D)(i), and (W) (the last only so that the window lies in $\{p+1,\ldots,T\}$). For all large $T$, on one event of probability one, the following hold for every $t\in\mathcal W$:
\begin{align}
  \bY_t-g_{Y,t}(\bC_t)&=\bTheta_t^\top\bV_t+\beps_t, \label{eq:A-oracle}\\
  \tilde\bT_t&=\bV_t-\Delta_{T,t}, \label{eq:A-tildeT}\\
  \tilde\bY_t&=\bTheta_t^\top\tilde\bT_t+\Delta_{g,t}+\beps_t. \label{eq:A-tildeY}
\end{align}
Moreover:
\begin{enumerate}[label=(\roman*)]
\item $\bV_t$ is $\mathcal F_{t-1}$-measurable and $\beps_t$ is $\mathcal F_t$-measurable.
\item For any deterministic $a\in\R^{np}$ and any $i\in\{1,\ldots,n\}$, the scalars $\xi_t:=a^\top\bV_t\,(\beps_t)_i$ form a \emph{martingale difference sequence} with respect to $(\mathcal F_t)$: each $\xi_t$ is $\mathcal F_t$-measurable and integrable, and $\E[\xi_t\mid\mathcal F_{t-1}]=0$ a.s. (In words: given the past, the next term has mean zero.)
\item $\E\xi_t^2\le\lvert a\rvert^2K_8^{1/2}$, and $\E[\xi_s\xi_t]=0$ for $s\ne t$. So the $\xi_t$ are uncorrelated and $\E\bigl(\sum_{t\in\mathcal W}\xi_t\bigr)^2=\sum_{t\in\mathcal W}\E\xi_t^2$.
\end{enumerate}
\end{lemma}

In words, after partialling out the confounders, the model becomes a plain linear regression of residual on residual, and cross-fitting adds only the explicit error terms $\Delta_{T,t}$ and $\Delta_{g,t}$.

\begin{proof}
\emph{Step 1: identity \eqref{eq:A-oracle}.} All terms of \eqref{eq:A-model} are integrable (Lemma~\ref{lem:A-moments} and Step 3 of its proof). Take $\E[\,\cdot\mid\bC_t]$ of \eqref{eq:A-model} term by term.
\begin{itemize}[leftmargin=1.5em]
\item $\E[\bY_t\mid\bC_t]=g_{Y,t}(\bC_t)$ by definition. Since $\bTheta_t$ is a deterministic matrix, linearity gives $\E[\bTheta_t^\top\bT_t\mid\bC_t]=\bTheta_t^\top\E[\bT_t\mid\bC_t]=\bTheta_t^\top g_{T,t}(\bC_t)$.
\item $g_t(\bC_t)$ is a Borel function of $\bC_t$, so $\E[g_t(\bC_t)\mid\bC_t]=g_t(\bC_t)$ by taking out what is known.
\item $\sigma(\bC_t)\subseteq\mathcal F_{t-1}$ because $\bC_t$ is $\mathcal F_{t-1}$-measurable. By the tower property and (M), $\E[\beps_t\mid\bC_t]=\E\bigl[\E[\beps_t\mid\mathcal F_{t-1}]\bigm|\bC_t\bigr]=\E[\bzero\mid\bC_t]=\bzero$.
\end{itemize}
Therefore $g_{Y,t}(\bC_t)=\bTheta_t^\top g_{T,t}(\bC_t)+g_t(\bC_t)$ a.s. Subtract this from \eqref{eq:A-model}:
\begin{align*}
  \bY_t-g_{Y,t}(\bC_t)
  &=\bTheta_t^\top\bT_t+g_t(\bC_t)+\beps_t-\bTheta_t^\top g_{T,t}(\bC_t)-g_t(\bC_t)\\
  &=\bTheta_t^\top\bigl(\bT_t-g_{T,t}(\bC_t)\bigr)+\beps_t
  =\bTheta_t^\top\bV_t+\beps_t .
\end{align*}
Each $t$ gives one null set on which this may fail. $\mathcal W$ is finite, and a finite union of null sets is a null set, so \eqref{eq:A-oracle} holds for all $t\in\mathcal W$ on one event of probability one. We work on this event from now on.

\emph{Step 2: identity \eqref{eq:A-tildeT}.} Add and subtract $g_{T,t}(\bC_t)$:
\begin{equation*}
  \tilde\bT_t=\bT_t-\hat g_{T,(t)}(\bC_t)
  =\bigl(\bT_t-g_{T,t}(\bC_t)\bigr)-\bigl(\hat g_{T,(t)}(\bC_t)-g_{T,t}(\bC_t)\bigr)
  =\bV_t-\Delta_{T,t}.
\end{equation*}
This is pure algebra and holds everywhere.

\emph{Step 3: identity \eqref{eq:A-tildeY}.} Add and subtract $g_{Y,t}(\bC_t)$, then use \eqref{eq:A-oracle}, then $\bV_t=\tilde\bT_t+\Delta_{T,t}$ from \eqref{eq:A-tildeT}:
\begin{align*}
  \tilde\bY_t
  &=\bigl(\bY_t-g_{Y,t}(\bC_t)\bigr)-\bigl(\hat g_{Y,(t)}(\bC_t)-g_{Y,t}(\bC_t)\bigr)
   =\bTheta_t^\top\bV_t+\beps_t-\Delta_{Y,t}\\
  &=\bTheta_t^\top\bigl(\tilde\bT_t+\Delta_{T,t}\bigr)+\beps_t-\Delta_{Y,t}
   =\bTheta_t^\top\tilde\bT_t+\bigl(\bTheta_t^\top\Delta_{T,t}-\Delta_{Y,t}\bigr)+\beps_t\\
  &=\bTheta_t^\top\tilde\bT_t+\Delta_{g,t}+\beps_t .
\end{align*}

\emph{Step 4: part (i).} $\bT_t$ and $\bC_t$ are functions of $X_{t-1},\ldots,X_{t-p}$, and $g_{T,t}$ is Borel, so $\bV_t=\bT_t-g_{T,t}(\bC_t)$ is $\mathcal F_{t-1}$-measurable. By (M), read as an identity between random vectors, $\beps_t=\bY_t-\bTheta_t^\top\bT_t-g_t(\bC_t)$ is a Borel function of $X_t,\ldots,X_{t-p}$, hence $\mathcal F_t$-measurable.

\emph{Step 5: part (ii).} By part (i), $\xi_t$ is $\mathcal F_t$-measurable. Next we bound $\E\lvert\xi_t\rvert$. By Cauchy--Schwarz for vectors, $\lvert a^\top\bV_t\rvert\le\lvert a\rvert\lvert\bV_t\rvert$; also $\lvert(\beps_t)_i\rvert\le\lvert\beps_t\rvert$. Then Cauchy--Schwarz for expectations and Lemma~\ref{lem:A-moments} ($\E\lvert\bV_t\rvert^2\le K_J$, $\E\lvert\beps_t\rvert^2\le K_J$) give
\begin{equation*}
  \E\lvert\xi_t\rvert\le\lvert a\rvert\,\E\bigl[\lvert\bV_t\rvert\,\lvert\beps_t\rvert\bigr]
  \le\lvert a\rvert\bigl(\E\lvert\bV_t\rvert^2\bigr)^{1/2}\bigl(\E\lvert\beps_t\rvert^2\bigr)^{1/2}\le\lvert a\rvert K_J^{1/2}K_J^{1/2}=\lvert a\rvert K_J<\infty .
\end{equation*}
Hence $\xi_t$ is integrable. Since $a^\top\bV_t$ is $\mathcal F_{t-1}$-measurable (part (i)), taking out what is known and (M) give
\begin{equation*}
  \E[\xi_t\mid\mathcal F_{t-1}]=a^\top\bV_t\,\E[(\beps_t)_i\mid\mathcal F_{t-1}]=a^\top\bV_t\cdot0=0 .
\end{equation*}

\emph{Step 6: part (iii).} Squaring $\lvert\xi_t\rvert\le\lvert a\rvert\lvert\bV_t\rvert\lvert\beps_t\rvert$ from Step 5, then Cauchy--Schwarz for expectations and Lemma~\ref{lem:A-moments} ($\E\lvert Z\rvert^4\le K_8^{1/2}$),
\begin{equation*}
  \E\xi_t^2\le\lvert a\rvert^2\,\E\bigl[\lvert\bV_t\rvert^2\lvert\beps_t\rvert^2\bigr]
  \le\lvert a\rvert^2\bigl(\E\lvert\bV_t\rvert^4\bigr)^{1/2}\bigl(\E\lvert\beps_t\rvert^4\bigr)^{1/2}
  \le\lvert a\rvert^2\,K_8^{1/4}K_8^{1/4}=\lvert a\rvert^2K_8^{1/2}.
\end{equation*}
Now let $s<t$. The product $\xi_s\xi_t$ is integrable (Cauchy--Schwarz and the bound just shown), and $\xi_s$ is $\mathcal F_s$-measurable, hence $\mathcal F_{t-1}$-measurable. By the tower property, taking out what is known and part (ii),
\begin{equation*}
  \E[\xi_s\xi_t]=\E\bigl[\E[\xi_s\xi_t\mid\mathcal F_{t-1}]\bigr]=\E\bigl[\xi_s\,\E[\xi_t\mid\mathcal F_{t-1}]\bigr]=\E[\xi_s\cdot0]=0 .
\end{equation*}
Expanding the square, $\E\bigl(\sum_t\xi_t\bigr)^2=\sum_t\E\xi_t^2+\sum_{s\ne t}\E[\xi_s\xi_t]=\sum_t\E\xi_t^2$.
\end{proof}

\begin{lemma}[The target]\label{lem:A-target}
Assume (M), (D)(i), (S), and (W) (the last only so that the window lies in $\{p+1,\ldots,T\}$). For all large $T$:
\begin{enumerate}[label=(\roman*)]
\item $\bar J$ is symmetric, $\lambda_{\min}(\bar J)\ge c_J$ and $\lVert\bar J\rVert\le K_J$. Hence $\bar J$ is invertible and $\lVert\bar J^{-1}\rVert\le1/c_J$.
\item $\sum_{t\in\mathcal W}J_tD_t=0$.
\item $\max_{t\in\mathcal W}\lVert D_t\rVert\le w_T:=L\,\dfrac{K_J}{c_J}\,\dfrac{T_w}{T}$.
\item $\lVert\bar\bTheta-\bTheta(u_0)\rVert\le w_T$.
\end{enumerate}
\end{lemma}

In words, $\bar\bTheta$ is a weighted average of the true coefficient matrices $\bTheta_t$ over the window, so it lies within $O(T_w/T)$ of each of them and of $\bTheta(u_0)$.

\noindent\emph{Example.} For $n=p=1$ all matrices are numbers, $J_t=\E\bV_t^2>0$ and $\bar\bTheta=\sum_tJ_t\bTheta_t/\sum_tJ_t$: a weighted average of $\bTheta_t$ with positive weights. A weighted average of numbers that all lie within $L(T_w-1)/T$ of each other is itself within that distance of each of them. The lemma is the matrix version of this, where the weights are matrices and the ratio $K_J/c_J$ pays for their not being scalars.

\begin{proof}
\emph{Part (i).} Each $J_t$ is symmetric, so $\bar J$ is. For any $x$, (S)(ii) gives $x^\top J_tx\ge c_J\lvert x\rvert^2$, so
\begin{equation*}
  x^\top\bar Jx=\frac1{T_w}\sum_{t\in\mathcal W}x^\top J_tx\ge\frac1{T_w}\sum_{t\in\mathcal W}c_J\lvert x\rvert^2=c_J\lvert x\rvert^2 ,
\end{equation*}
and hence $\lambda_{\min}(\bar J)=\min_{\lvert x\rvert=1}x^\top\bar Jx\ge c_J$. By the triangle inequality and Lemma~\ref{lem:A-moments}, $\lVert\bar J\rVert\le T_w^{-1}\sum_t\lVert J_t\rVert\le K_J$. A symmetric matrix with all eigenvalues at least $c_J>0$ is invertible, and the eigenvalues of its inverse are the reciprocals, all at most $1/c_J$; so $\lVert\bar J^{-1}\rVert\le1/c_J$.

\emph{A key identity.} Let $M\in\R^{np\times n}$ be any fixed matrix. Since $\bar J^{-1}\bar J=\bI$ and $\bar J=T_w^{-1}\sum_tJ_t$,
\begin{equation*}
  M=\bar J^{-1}\bar JM=\bar J^{-1}\frac1{T_w}\sum_{t\in\mathcal W}J_tM .
\end{equation*}
Subtracting this from the definition of $\bar\bTheta$,
\begin{equation}\label{eq:A-key}
  \bar\bTheta-M=\bar J^{-1}\frac1{T_w}\sum_{t\in\mathcal W}J_t\bTheta_t-\bar J^{-1}\frac1{T_w}\sum_{t\in\mathcal W}J_tM
  =\bar J^{-1}\frac1{T_w}\sum_{t\in\mathcal W}J_t(\bTheta_t-M).
\end{equation}

\emph{Part (ii).} Using $\sum_tJ_t=T_w\bar J$ and the definition \eqref{eq:A-target} of $\bar\bTheta$,
\begin{equation*}
  \sum_{t\in\mathcal W}J_tD_t=\sum_{t\in\mathcal W}J_t\bTheta_t-\Bigl(\sum_{t\in\mathcal W}J_t\Bigr)\bar\bTheta
  =\sum_{t\in\mathcal W}J_t\bTheta_t-T_w\bar J\,\bar J^{-1}\frac1{T_w}\sum_{t\in\mathcal W}J_t\bTheta_t=0 .
\end{equation*}

\emph{Part (iii).} Fix $s\in\mathcal W$ and take $M=\bTheta_s$ in \eqref{eq:A-key}. For $t\in\mathcal W$, (S)(i) and $\lvert t-s\rvert\le T_w-1$ give $\lVert\bTheta_t-\bTheta_s\rVert\le L\lvert t-s\rvert/T\le L(T_w-1)/T$. By $\lVert AB\rVert\le\lVert A\rVert\lVert B\rVert$, the triangle inequality, part (i) and Lemma~\ref{lem:A-moments},
\begin{align*}
  \lVert D_s\rVert=\lVert\bar\bTheta-\bTheta_s\rVert
  &\le\lVert\bar J^{-1}\rVert\,\frac1{T_w}\sum_{t\in\mathcal W}\lVert J_t\rVert\,\lVert\bTheta_t-\bTheta_s\rVert\\
  &\le\frac1{c_J}\cdot\frac1{T_w}\sum_{t\in\mathcal W}K_J\cdot L\frac{T_w-1}{T}
  =L\,\frac{K_J}{c_J}\,\frac{T_w-1}{T}\le w_T .
\end{align*}
Since $s\in\mathcal W$ was arbitrary, this proves (iii).

\emph{Part (iv).} Take $M=\bTheta(u_0)$ in \eqref{eq:A-key}. For $t\in\mathcal W$, both $t/T$ and $u_0$ lie in $[0,1]$, so (S)(i) and \eqref{eq:A-window} give $\lVert\bTheta_t-\bTheta(u_0)\rVert\le L\lvert t/T-u_0\rvert\le L\,T_w/T$. Exactly as in part (iii),
\begin{equation*}
  \lVert\bar\bTheta-\bTheta(u_0)\rVert
  \le\frac1{c_J}\cdot\frac1{T_w}\sum_{t\in\mathcal W}K_J\cdot L\frac{T_w}{T}
  =L\,\frac{K_J}{c_J}\,\frac{T_w}{T}=w_T. \qedhere
\end{equation*}
\end{proof}

\paragraph{Limits.}
All limits are as $T\to\infty$. Let $X_T$ be random (a scalar, or a vector or matrix measured in a norm).
\begin{itemize}[leftmargin=1.5em]
\item \emph{Convergence in probability}, $X_T\convp0$, also written $X_T=o_p(1)$: $P(\lvert X_T\rvert>\epsilon)\to0$ for every $\epsilon>0$. In words: large values become rare. For vectors and matrices this is the same as entry by entry.
\item \emph{Bounded in probability}, $X_T=O_p(1)$: for every $\eta>0$ there are $M$ and $T_0$ with $P(\lvert X_T\rvert>M)\le\eta$ for all $T\ge T_0$. In words: $X_T$ does not escape to infinity. For deterministic $a_T>0$, $X_T=o_p(a_T)$ means $X_T/a_T=o_p(1)$, and similarly for $O_p$.
\item \emph{Convergence in distribution}, $X_T\convd\N(0,1)$: $P(X_T\le x)\to\Phi(x)$ for every $x\in\R$, where $\Phi$ is the standard normal distribution function.
\end{itemize}
Standard matrix inequalities and the usual rules for these limits (sums and products, Markov's and Chebyshev's inequalities, Slutsky's lemma) are named where they are used.

\begin{lemma}[Law of large numbers on the window]\label{lem:A-lln}
Let $f_t$, $t\in\mathcal W$, be scalar random variables with $\E f_t^2<\infty$ and $\sum_{h\ge0}\lvert\Cov(f_t,f_{t+h})\rvert\le C_D$ for every $t$ (the sum over those $h$ with $t+h\in\mathcal W$ suffices). Let $b_t$ be deterministic with $\lvert b_t\rvert\le B$, where $B$ does not depend on $T$ or $t$. Then
\begin{equation}\label{eq:A-lln-sum}
  \E\Bigl(\sum_{t\in\mathcal W}b_t\,(f_t-\E f_t)\Bigr)^2\le2B^2C_D\,T_w ,
  \qquad\text{equivalently}\qquad
  \E\Bigl(\frac1{T_w}\sum_{t\in\mathcal W}b_t\,(f_t-\E f_t)\Bigr)^2\le\frac{2B^2C_D}{T_w}.
\end{equation}
Consequently, under (W), $T_w^{-1}\sum_{t\in\mathcal W}b_t(f_t-\E f_t)\convp0$. By (D)(ii) and Lemma~\ref{lem:A-moments}, every entry of $\bV_t\bV_t^\top$ and every entry of $\beps_t\beps_t^\top$ satisfies the hypotheses on $f_t$.
\end{lemma}

In words, a weighted window average of a short-memory sequence stays within $O(T_w^{-1/2})$ of its mean.

\begin{proof}
Write $\bar f_t:=f_t-\E f_t$, so that $\E[\bar f_s\bar f_t]=\Cov(f_s,f_t)$, which is finite because $\E f_t^2<\infty$. Expanding the square,
\begin{align*}
  \E\Bigl(\sum_{t\in\mathcal W}b_t\bar f_t\Bigr)^2
  =\sum_{s\in\mathcal W}\sum_{t\in\mathcal W}b_sb_t\,\Cov(f_s,f_t)
  \le\sum_{s\in\mathcal W}\sum_{t\in\mathcal W}\lvert b_s\rvert\lvert b_t\rvert\,\lvert\Cov(f_s,f_t)\rvert
  \le B^2\sum_{s\in\mathcal W}\sum_{t\in\mathcal W}\lvert\Cov(f_s,f_t)\rvert .
\end{align*}
Split the double sum into the pairs with $t\ge s$ and the pairs with $t<s$. In the first group write $t=s+h$ with $h\ge0$; in the second, swap the names and write $s=t+h$ with $h\ge1$. Using $\Cov(f_s,f_t)=\Cov(f_t,f_s)$,
\begin{align*}
  \sum_{s\in\mathcal W}\sum_{t\in\mathcal W}\lvert\Cov(f_s,f_t)\rvert
  &=\sum_{s\in\mathcal W}\sum_{\substack{h\ge0\\ s+h\in\mathcal W}}\lvert\Cov(f_s,f_{s+h})\rvert
   +\sum_{t\in\mathcal W}\sum_{\substack{h\ge1\\ t+h\in\mathcal W}}\lvert\Cov(f_t,f_{t+h})\rvert\\
  &\le\sum_{s\in\mathcal W}C_D+\sum_{t\in\mathcal W}C_D=2C_DT_w .
\end{align*}
Combining the two displays gives the first inequality in \eqref{eq:A-lln-sum}; dividing by $T_w^2$ gives the second.

For the consequence, fix $\eta>0$. By Chebyshev's inequality and \eqref{eq:A-lln-sum},
\begin{equation*}
  P\Bigl(\Bigl\lvert\frac1{T_w}\sum_{t\in\mathcal W}b_t\bar f_t\Bigr\rvert>\eta\Bigr)
  \le\frac1{\eta^2}\,\E\Bigl(\frac1{T_w}\sum_{t\in\mathcal W}b_t\bar f_t\Bigr)^2
  \le\frac{2B^2C_D}{\eta^2\,T_w}\longrightarrow0,
\end{equation*}
because $T_w\to\infty$ by (W).

Finally, for $f_t=(\bV_t)_l(\bV_t)_b$ we have $\E f_t^2\le\E\lvert\bV_t\rvert^4\le K_8^{1/2}$, and for $f_t=(\beps_t)_i(\beps_t)_j$ we have $\E f_t^2\le\E\lvert\beps_t\rvert^4\le K_8^{1/2}$ (Lemma~\ref{lem:A-moments}, using $\lvert (\bV_t)_l(\bV_t)_b\rvert\le\lvert\bV_t\rvert^2$ and likewise for $\beps_t$); the covariance bound is (D)(ii).
\end{proof}

\subsection{Proof of Theorem~\ref{thm:dml-var}: asymptotic normality}\label{app:thm1-proof}

We restate the theorem, with the quantities of Appendix~\ref{app:assumptions}.

\medskip\noindent\textbf{Theorem~\ref{thm:dml-var} (restated).} \emph{Assume (M), (S), (W), (D), (H), (N) of Appendix~\ref{app:assumptions}. Let $\hbTheta$ be the estimator \eqref{eq:A-thetahat}, $\bar\bTheta$ the window target \eqref{eq:A-target}, and $\hat s_{ri}$ the standard error \eqref{eq:A-se}. Then for every entry $(r,i)$,}
\begin{equation}\label{eq:A-thm1}
  \frac{\hbTheta_{ri}-\bar\bTheta_{ri}}{\hat s_{ri}}\convd\N(0,1).
\end{equation}
\emph{If moreover $T_w=o(T^{2/3})$, then $\bar\bTheta_{ri}$ may be replaced by $\bTheta_{ri}(u_0)$.}
\medskip

\paragraph{Roadmap.}
This section proves that each entry of $\sqrt{T_w}(\hbTheta-\bar\bTheta)$, divided by a deterministic scale $\sqrt{v_{T,ri}}$, is asymptotically $\N(0,1)$ (Proposition~\ref{prop:B-main}). The steps, in order:
\begin{itemize}[nosep]
\item Lemma~\ref{lem:B-decomp} (pure algebra): write $\sqrt{T_w}(\hbTheta-\bar\bTheta)$ exactly as $\hat J^{-1}$ times the sum of one main term $S_0$ and six remainders.
\item Lemma~\ref{lem:B-gram}: $\hat J$ is close to the deterministic matrix $\bar J$.
\item Lemma~\ref{lem:B-remainders}: the six remainders vanish.
\item Theorem~\ref{thm:B-mclt} and Lemma~\ref{lem:B-clt}: $S_0$ satisfies a central limit theorem. The first is the cited martingale CLT; the second checks its conditions.
\item Proposition~\ref{prop:B-main}: combine the above.
\end{itemize}
Each lemma uses only the ones before it, so the final proof is a short chain of citations. The standard error, which replaces $v_{T,ri}$ by an estimate, is treated in Appendix~\ref{app:se}.

\paragraph{Conventions for this section.}
Throughout, $\sum_t$ means $\sum_{t\in\mathcal W}$, a sum of exactly $T_w$ terms. Both $\bT_t$ and $\bV_t$ have length $np$.
All statements are for $T$ large enough that the window lies in $\{p+1,\ldots,T\}$; this is the range in which the lemmas of Appendix~\ref{app:assumptions} hold, and it is all that a limit as $T\to\infty$ needs.
Recall from Lemma~\ref{lem:A-moments} that, with $K_8:=2^8p^8K_X$ and $K_J:=K_8^{1/4}$, for every $t\in\mathcal W$ and $0<s\le8$,
\[
\E\lvert\bV_t\rvert^s\le K_8^{s/8},\qquad\E\lvert\beps_t\rvert^s\le K_8^{s/8},\qquad\lVert J_t\rVert\le K_J ;
\]
we use $s=2$ ($K_8^{1/4}=K_J$), $s=4$ ($K_8^{1/2}$) and $s=8$.
The matrices $D_t=\bTheta_t-\bar\bTheta$ are deterministic, because $\bTheta_t=\bTheta(t/T)$ is a fixed matrix and $\bar\bTheta$ is built from the deterministic $J_t$ and $\bTheta_t$.
By Lemma~\ref{lem:A-target}(iii), $\max_t\lVert D_t\rVert\le w_T=L(K_J/c_J)\,T_w/T$. Two consequences of (W) are used repeatedly:
\begin{equation}\label{eq:B-wT}
w_T=L\frac{K_J}{c_J}\cdot\frac{T_w^{3/2}}{T}\cdot\frac{1}{\sqrt{T_w}}\longrightarrow0,
\qquad
\sqrt{T_w}\,w_T=L\frac{K_J}{c_J}\cdot\frac{T_w^{3/2}}{T}\le C_w\quad\text{for all large }T,
\end{equation}
where $C_w<\infty$ is a constant: the first because $T_w^{3/2}/T=O(1)$ and $T_w\to\infty$, the second because $T_w^{3/2}/T=O(1)$.

\subsubsection*{Step 1: an exact decomposition}

\begin{lemma}[Error decomposition]\label{lem:B-decomp}
Assume (M), (S), (W) and (D)(i), so that the identities of Lemma~\ref{lem:A-identities} hold and $\bar J$ is invertible (Lemma~\ref{lem:A-target}(i)). For all large $T$, almost surely on the event $\{\hat J\text{ invertible}\}$,
\begin{equation}\label{eq:B-decomp}
\sqrt{T_w}\,(\hbTheta-\bar\bTheta)=\hat J^{-1}\bigl(S_0+R_1+R_2+R_3+R_{4a}+R_{4b}+R_{4c}\bigr),
\end{equation}
where all terms are $np\times n$ matrices:
\begin{align*}
S_0&=\frac1{\sqrt{T_w}}\sum_t\bV_t\beps_t^\top, &
R_1&=-\frac1{\sqrt{T_w}}\sum_t\Delta_{T,t}\beps_t^\top, &
R_2&=\frac1{\sqrt{T_w}}\sum_t\bV_t\Delta_{g,t}^\top,\\
R_3&=-\frac1{\sqrt{T_w}}\sum_t\Delta_{T,t}\Delta_{g,t}^\top, &
R_{4a}&=\frac1{\sqrt{T_w}}\sum_t(\bV_t\bV_t^\top-J_t)D_t, &
R_{4b}&=-\frac1{\sqrt{T_w}}\sum_t(\bV_t\Delta_{T,t}^\top+\Delta_{T,t}\bV_t^\top)D_t,\\
&& R_{4c}&=\frac1{\sqrt{T_w}}\sum_t\Delta_{T,t}\Delta_{T,t}^\top D_t. &&
\end{align*}
\end{lemma}

In words, the estimation error splits exactly, with no approximation, into one main term and six remainders, all multiplied by the same matrix $\hat J^{-1}$.

\begin{proof}
\emph{A scalar example first.} Take $n=p=1$, so every matrix is a number; write $\tilde v_t=\tilde\bT_t$, $v_t=\bV_t$, $\theta_t=\bTheta_t$. Then $\hat\theta=\sum_t\tilde v_t\tilde y_t/\sum_t\tilde v_t^2$ and $\tilde y_t=\theta_t\tilde v_t+\Delta_{g,t}+\varepsilon_t$, so
\[
\hat\theta-\bar\theta=\frac{\sum_t\tilde v_t^2(\theta_t-\bar\theta)+\sum_t\tilde v_t\Delta_{g,t}+\sum_t\tilde v_t\varepsilon_t}{\sum_t\tilde v_t^2}.
\]
Substituting $\tilde v_t=v_t-\Delta_{T,t}$, the last sum gives $S_0$ and $R_1$, the middle sum gives $R_2$ and $R_3$, and the first sum (the drift of the coefficient inside the window) gives $R_{4a},R_{4b},R_{4c}$ after its mean part cancels. The matrix computation below is the same.

\emph{Where the identities hold.} By Lemma~\ref{lem:A-identities}, for all large $T$ there is one event of probability one on which the identities
\begin{equation}\label{eq:B-ident}
\tilde\bT_t=\bV_t-\Delta_{T,t},\qquad \tilde\bY_t=\bTheta_t^\top\tilde\bT_t+\Delta_{g,t}+\beps_t
\end{equation}
hold for all $t\in\mathcal W$ at once. We fix an outcome in this event at which $\hat J$ is invertible. Everything below is algebra at this outcome.

\emph{Step (a): the cross product.} Transposing the second identity in \eqref{eq:B-ident} gives $\tilde\bY_t^\top=\tilde\bT_t^\top\bTheta_t+\Delta_{g,t}^\top+\beps_t^\top$. Multiplying on the left by the column $\tilde\bT_t$,
\[
\tilde\bT_t\tilde\bY_t^\top=\tilde\bT_t\tilde\bT_t^\top\bTheta_t+\tilde\bT_t\Delta_{g,t}^\top+\tilde\bT_t\beps_t^\top .
\]

\emph{Step (b): the last two products.} Insert $\tilde\bT_t=\bV_t-\Delta_{T,t}$:
\[
\tilde\bT_t\Delta_{g,t}^\top=\bV_t\Delta_{g,t}^\top-\Delta_{T,t}\Delta_{g,t}^\top,
\qquad
\tilde\bT_t\beps_t^\top=\bV_t\beps_t^\top-\Delta_{T,t}\beps_t^\top .
\]

\emph{Step (c): the first product.} Write $\bTheta_t=\bar\bTheta+D_t$, so $\tilde\bT_t\tilde\bT_t^\top\bTheta_t=\tilde\bT_t\tilde\bT_t^\top\bar\bTheta+\tilde\bT_t\tilde\bT_t^\top D_t$. Expanding the Gram product with $\tilde\bT_t=\bV_t-\Delta_{T,t}$,
\[
\tilde\bT_t\tilde\bT_t^\top=(\bV_t-\Delta_{T,t})(\bV_t-\Delta_{T,t})^\top
=\bV_t\bV_t^\top-\bV_t\Delta_{T,t}^\top-\Delta_{T,t}\bV_t^\top+\Delta_{T,t}\Delta_{T,t}^\top,
\]
and adding and subtracting $J_t$ in the first term,
\[
\tilde\bT_t\tilde\bT_t^\top D_t
=(\bV_t\bV_t^\top-J_t)D_t+J_tD_t-(\bV_t\Delta_{T,t}^\top+\Delta_{T,t}\bV_t^\top)D_t+\Delta_{T,t}\Delta_{T,t}^\top D_t .
\]

\emph{Step (d): the exact cancellation.} $\sum_tJ_tD_t=0$ by Lemma~\ref{lem:A-target}(ii).
The term that appears in Step (c) is $J_tD_t$, with $J_t=\E[\bV_t\bV_t^\top]$; it cancels because $\bar\bTheta$ is built with these same weights. Any other weights would leave a deterministic bias.

\emph{Step (e): collect.} Average Step (a) over $t\in\mathcal W$ and insert Steps (b) to (d). Since $\bar\bTheta$ does not depend on $t$, $T_w^{-1}\sum_t\tilde\bT_t\tilde\bT_t^\top\bar\bTheta=\hat J\bar\bTheta$. Writing $\frac1{T_w}=\frac1{\sqrt{T_w}}\cdot\frac1{\sqrt{T_w}}$ and matching each sum with its definition,
\begin{align*}
\frac1{T_w}\sum_t\tilde\bT_t\tilde\bY_t^\top
&=\hat J\bar\bTheta
+\frac1{T_w}\sum_t(\bV_t\bV_t^\top-J_t)D_t
+\frac1{T_w}\underbrace{\sum_tJ_tD_t}_{=0}\\
&\qquad-\frac1{T_w}\sum_t(\bV_t\Delta_{T,t}^\top+\Delta_{T,t}\bV_t^\top)D_t
+\frac1{T_w}\sum_t\Delta_{T,t}\Delta_{T,t}^\top D_t\\
&\qquad+\frac1{T_w}\sum_t\bV_t\Delta_{g,t}^\top-\frac1{T_w}\sum_t\Delta_{T,t}\Delta_{g,t}^\top
+\frac1{T_w}\sum_t\bV_t\beps_t^\top-\frac1{T_w}\sum_t\Delta_{T,t}\beps_t^\top\\
&=\hat J\bar\bTheta+\frac1{\sqrt{T_w}}\bigl(R_{4a}+R_{4b}+R_{4c}+R_2+R_3+S_0+R_1\bigr).
\end{align*}

\emph{Step (f): solve.} Since $\hat J$ is invertible, $\hbTheta=\hat J^{-1}T_w^{-1}\sum_t\tilde\bT_t\tilde\bY_t^\top$. Multiplying Step (e) on the left by $\hat J^{-1}$ and using $\hat J^{-1}\hat J\bar\bTheta=\bar\bTheta$,
\[
\hbTheta=\bar\bTheta+\frac1{\sqrt{T_w}}\hat J^{-1}\bigl(S_0+R_1+R_2+R_3+R_{4a}+R_{4b}+R_{4c}\bigr).
\]
Subtracting $\bar\bTheta$ and multiplying by $\sqrt{T_w}$ gives \eqref{eq:B-decomp}.
\end{proof}

\subsubsection*{Step 2: the Gram matrix}

\begin{lemma}[Gram matrix]\label{lem:B-gram}
Assume (M), (S), (W), (D) and \eqref{eq:N1a} in (N1). Let $E_T:=\{\lambda_{\min}(\hat J)\ge c_J/2\}$ ($\hat J$ is symmetric, so $\lambda_{\min}(\hat J)$ is defined). Then:
\begin{enumerate}[label=(\roman*)]
\item $\lVert\hat J-\bar J\rVert\convp0$;
\item $P(E_T)\to1$, and on $E_T$ the matrix $\hat J$ is invertible with $\lVert\hat J^{-1}\rVert\le2/c_J$;
\item on $E_T$, $\lVert\hat J^{-1}-\bar J^{-1}\rVert\le(2/c_J^2)\lVert\hat J-\bar J\rVert$. Hence for every $\epsilon>0$, $P\bigl(E_T\cap\{\lVert\hat J^{-1}-\bar J^{-1}\rVert\le\epsilon\}\bigr)\to1$. This is what ``$\hat J^{-1}-\bar J^{-1}\convp0$'' means here; the value of $\hat J^{-1}$ off $E_T$ is irrelevant because $P(E_T^c)\to0$.
\end{enumerate}
\end{lemma}

In words, the empirical Gram matrix $\hat J$ of the cross-fitted residuals is close to the average true Gram matrix $\bar J$, so it is invertible with a bounded inverse, and its inverse is close to $\bar J^{-1}$.

\begin{proof}
\emph{Step (a): split $\hat J$.} By \eqref{eq:A-tildeT}, $\tilde\bT_t=\bV_t-\Delta_{T,t}$, and by the expansion in Step (c) of Lemma~\ref{lem:B-decomp} (pure algebra from \eqref{eq:A-tildeT}, valid at every outcome),
\[
\hat J=\frac1{T_w}\sum_t\tilde\bT_t\tilde\bT_t^\top
=\frac1{T_w}\sum_t\bV_t\bV_t^\top-G-G^\top+Q,
\qquad G:=\frac1{T_w}\sum_t\bV_t\Delta_{T,t}^\top,\quad Q:=\frac1{T_w}\sum_t\Delta_{T,t}\Delta_{T,t}^\top,
\]
using $\Delta_{T,t}\bV_t^\top=(\bV_t\Delta_{T,t}^\top)^\top$. Subtracting $\bar J$,
\[
\hat J-\bar J=\Bigl(\frac1{T_w}\sum_t\bV_t\bV_t^\top-\bar J\Bigr)-G-G^\top+Q .
\]
We show that each of the four pieces is $o_p(1)$.

\emph{Step (b): the oracle part.} Fix $l,b\in\{1,\dots,np\}$ and let $f_t:=(\bV_t)_l(\bV_t)_b$, an entry of $\bV_t\bV_t^\top$. Since $J_t=\E[\bV_t\bV_t^\top]$, $(J_t)_{lb}=\E f_t$, and since $\bar J=T_w^{-1}\sum_tJ_t$,
\[
\Bigl(\frac1{T_w}\sum_t\bV_t\bV_t^\top-\bar J\Bigr)_{lb}=\frac1{T_w}\sum_t\bigl(f_t-\E f_t\bigr).
\]
By (D)(ii) and Lemma~\ref{lem:A-moments}, $f_t$ satisfies the hypotheses of Lemma~\ref{lem:A-lln} (as noted at the end of that lemma), so Lemma~\ref{lem:A-lln} with all weights equal to $1$ (so $B=1$) gives
$\E\bigl(T_w^{-1}\sum_t(f_t-\E f_t)\bigr)^2\le2C_D/T_w$. Summing over the $(np)^2$ entries,
\[
\E\Bigl\lvert\frac1{T_w}\sum_t\bV_t\bV_t^\top-\bar J\Bigr\rvert_F^2\le\frac{2(np)^2C_D}{T_w},
\qquad\text{so}\qquad
P\Bigl(\Bigl\lvert\frac1{T_w}\sum_t\bV_t\bV_t^\top-\bar J\Bigr\rvert_F>\epsilon\Bigr)\le\frac{2(np)^2C_D}{\epsilon^2T_w}\longrightarrow0
\]
by Chebyshev's inequality and $T_w\to\infty$ (W).

\emph{Step (c): the cross term $G$.} By the triangle inequality for $\lvert\cdot\rvert_F$, then $\lvert ab^\top\rvert_F=\lvert a\rvert\,\lvert b\rvert$, then Cauchy--Schwarz for averages,
\[
\lvert G\rvert_F\le\frac1{T_w}\sum_t\lvert\bV_t\Delta_{T,t}^\top\rvert_F=\frac1{T_w}\sum_t\lvert\bV_t\rvert\,\lvert\Delta_{T,t}\rvert
\le\Bigl(\frac1{T_w}\sum_t\lvert\bV_t\rvert^2\Bigr)^{1/2}\Bigl(\frac1{T_w}\sum_t\lvert\Delta_{T,t}\rvert^2\Bigr)^{1/2}.
\]
First factor: $\E\,T_w^{-1}\sum_t\lvert\bV_t\rvert^2\le\max_t\E\lvert\bV_t\rvert^2\le K_8^{1/4}$ (conventions above), so $T_w^{-1}\sum_t\lvert\bV_t\rvert^2=O_p(1)$ by Markov's inequality, and so is its square root. Second factor: $T_w^{-1}\sum_t\lvert\Delta_{T,t}\rvert^2\convp0$ by \eqref{eq:N1a} in (N1), so is its square root. Since $O_p(1)\cdot o_p(1)=o_p(1)$, $\lvert G\rvert_F=o_p(1)$, and $\lvert G^\top\rvert_F=\lvert G\rvert_F$.

\emph{Step (d): the quadratic term $Q$.} By the triangle inequality and $\lvert ab^\top\rvert_F=\lvert a\rvert\,\lvert b\rvert$, $\lvert Q\rvert_F\le T_w^{-1}\sum_t\lvert\Delta_{T,t}\rvert^2\convp0$ by \eqref{eq:N1a} in (N1).

\emph{Step (e): proof of (i).} Since $\lVert M\rVert\le\lvert M\rvert_F$, $\lVert\hat J-\bar J\rVert\le\lvert\hat J-\bar J\rvert_F$, and by the triangle inequality this is at most the sum of the Frobenius norms of the four pieces. Each is $o_p(1)$ by Steps (b) to (d), so the sum is $o_p(1)$ too.

\emph{Step (f): proof of (ii).} Recall $E_T=\{\lambda_{\min}(\hat J)\ge c_J/2\}$ from the statement. We use Weyl's inequality, in the form for symmetric $A,B$, $\lambda_{\min}(A)\ge\lambda_{\min}(B)-\lVert A-B\rVert$: moving a symmetric matrix by $\lVert A-B\rVert$ in operator norm moves its smallest eigenvalue by at most that much.
Both $\hat J$ and $\bar J$ are symmetric (averages of symmetric matrices), and $\lambda_{\min}(\bar J)\ge c_J$ by Lemma~\ref{lem:A-target}(i). So
\[
\lambda_{\min}(\hat J)\ge c_J-\lVert\hat J-\bar J\rVert.
\]
In particular, at every outcome with $\lVert\hat J-\bar J\rVert\le c_J/2$ we get $\lambda_{\min}(\hat J)\ge c_J/2$, that is, $\{\lVert\hat J-\bar J\rVert\le c_J/2\}\subseteq E_T$. Taking complements, $P(E_T^c)\le P(\lVert\hat J-\bar J\rVert>c_J/2)\to0$ by (i). On $E_T$, since $\hat J$ is symmetric with $\lambda_{\min}(\hat J)>0$, $\hat J$ is invertible with $\lVert\hat J^{-1}\rVert=1/\lambda_{\min}(\hat J)\le2/c_J$.

\emph{Step (g): proof of (iii).} On $E_T$ both inverses exist: $\hat J$ is invertible on $E_T$ by (ii), and $\bar J$ is invertible at every outcome because $\lambda_{\min}(\bar J)\ge c_J>0$ (Lemma~\ref{lem:A-target}(i)), so all its eigenvalues are positive. Then
\[
\hat J^{-1}-\bar J^{-1}=\hat J^{-1}\bar J\bar J^{-1}-\hat J^{-1}\hat J\bar J^{-1}=\hat J^{-1}(\bar J-\hat J)\bar J^{-1}.
\]
By $\lVert AB\rVert\le\lVert A\rVert\,\lVert B\rVert$, $\lVert A^{-1}\rVert=1/\lambda_{\min}(A)$ for a symmetric positive definite $A$, and $\lambda_{\min}(\bar J)\ge c_J$,
\[
\lVert\hat J^{-1}-\bar J^{-1}\rVert\le\lVert\hat J^{-1}\rVert\,\lVert\bar J-\hat J\rVert\,\lVert\bar J^{-1}\rVert\le\frac2{c_J}\cdot\lVert\hat J-\bar J\rVert\cdot\frac1{c_J}.
\]
Therefore $E_T\cap\{\lVert\hat J-\bar J\rVert\le c_J^2\epsilon/2\}\subseteq E_T\cap\{\lVert\hat J^{-1}-\bar J^{-1}\rVert\le\epsilon\}$, and the probability of the left event is at least $1-P(E_T^c)-P(\lVert\hat J-\bar J\rVert>c_J^2\epsilon/2)$, which tends to $1$ by (ii) and (i).
\end{proof}

\subsubsection*{Step 3: the remainders}

\begin{lemma}[Remainders vanish]\label{lem:B-remainders}
Assume (M), (S), (W), (D), (N). Then $\lvert R_1\rvert_F$, $\lvert R_2\rvert_F$, $\lvert R_3\rvert_F$, $\lvert R_{4a}\rvert_F$, $\lvert R_{4b}\rvert_F$ and $\lvert R_{4c}\rvert_F$ (the remainders defined in \eqref{eq:B-decomp} of Lemma~\ref{lem:B-decomp}) all tend to $0$ in probability.
\end{lemma}

In words, the nuisance-error terms and the coefficient-drift terms are all negligible after the $\sqrt{T_w}$ scaling.

\begin{proof}
\emph{$R_1$ and $R_2$.} $R_1=-T_w^{-1/2}\sum_t\Delta_{T,t}\beps_t^\top$ and $R_2=T_w^{-1/2}\sum_t\bV_t\Delta_{g,t}^\top$ tend to $0$ in probability directly by \eqref{eq:N2a} and \eqref{eq:N2b} in (N2).

\emph{$R_3$.} By the triangle inequality, $\lvert ab^\top\rvert_F=\lvert a\rvert\,\lvert b\rvert$, and Cauchy--Schwarz for sums,
\begin{align*}
\lvert R_3\rvert_F
&\le\frac1{\sqrt{T_w}}\sum_t\lvert\Delta_{T,t}\Delta_{g,t}^\top\rvert_F
=\frac1{\sqrt{T_w}}\sum_t\lvert\Delta_{T,t}\rvert\,\lvert\Delta_{g,t}\rvert
\le\frac1{\sqrt{T_w}}\Bigl(\sum_t\lvert\Delta_{T,t}\rvert^2\Bigr)^{1/2}\Bigl(\sum_t\lvert\Delta_{g,t}\rvert^2\Bigr)^{1/2}\\
&=\frac{T_w}{\sqrt{T_w}}\Bigl(\frac1{T_w}\sum_t\lvert\Delta_{T,t}\rvert^2\Bigr)^{1/2}\Bigl(\frac1{T_w}\sum_t\lvert\Delta_{g,t}\rvert^2\Bigr)^{1/2}
=\sqrt{T_w}\Bigl(\frac1{T_w}\sum_t\lvert\Delta_{T,t}\rvert^2\Bigr)^{1/2}\Bigl(\frac1{T_w}\sum_t\lvert\Delta_{g,t}\rvert^2\Bigr)^{1/2}.
\end{align*}
The last expression tends to $0$ in probability by the product rate \eqref{eq:N1c} in (N1).

\emph{$R_{4a}$.} Fix an entry $(l,i)$, $l\in\{1,\dots,np\}$, $i\in\{1,\dots,n\}$. Writing out the matrix product,
\[
(R_{4a})_{li}=\frac1{\sqrt{T_w}}\sum_t\sum_{b=1}^{np}\bigl((\bV_t)_l(\bV_t)_b-(J_t)_{lb}\bigr)(D_t)_{bi}
=\sqrt{T_w}\sum_{b=1}^{np}\chi_{lbi},
\qquad
\chi_{lbi}:=\frac1{T_w}\sum_t(D_t)_{bi}\bigl(f^{lb}_t-\E f^{lb}_t\bigr),
\]
with $f^{lb}_t:=(\bV_t)_l(\bV_t)_b$ and $\E f^{lb}_t=(J_t)_{lb}$. The weights $(D_t)_{bi}$ are deterministic, and $\lvert(D_t)_{bi}\rvert\le\lVert D_t\rVert\le w_T$. The first inequality holds because an entry is at most the operator norm: by Cauchy--Schwarz, $\lvert(D_t)_{bi}\rvert=\lvert e_b^\top(D_te_i)\rvert\le\lvert e_b\rvert\,\lvert D_te_i\rvert\le\lVert D_t\rVert$. The second is Lemma~\ref{lem:A-target}(iii). The variable $f^{lb}_t$ is an entry of $\bV_t\bV_t^\top$, so it satisfies the hypotheses of Lemma~\ref{lem:A-lln}. That lemma asks for a weight bound $B$ that does not depend on $T$, while $w_T$ does, so we rescale the weights.
If $L=0$, then $w_T=0$, so every $D_t=0$ and $R_{4a}=0$; there is nothing to prove. If $L>0$, then $w_T>0$, and the rescaled weights $\beta_t:=(D_t)_{bi}/w_T$ (for the fixed $b,i$) are deterministic with $\lvert\beta_t\rvert\le1$. Since $\chi_{lbi}=w_T\cdot T_w^{-1}\sum_t\beta_t(f^{lb}_t-\E f^{lb}_t)$, Lemma~\ref{lem:A-lln} with weights $\beta_t$ and $B=1$ gives
\[
\E \chi_{lbi}^2=w_T^2\,\E\Bigl(\frac1{T_w}\sum_t\beta_t\bigl(f^{lb}_t-\E f^{lb}_t\bigr)\Bigr)^2\le w_T^2\cdot\frac{2C_D}{T_w}=\frac{2w_T^2C_D}{T_w}.
\]
By Cauchy--Schwarz with the all-ones vector, $\bigl(\sum_{b=1}^{np}\chi_{lbi}\bigr)^2\le np\sum_{b=1}^{np}\chi_{lbi}^2$, so
\[
\E(R_{4a})_{li}^2=T_w\,\E\Bigl(\sum_{b=1}^{np}\chi_{lbi}\Bigr)^2\le T_w\cdot np\cdot np\cdot\frac{2w_T^2C_D}{T_w}=2(np)^2C_Dw_T^2 .
\]
Summing over the $np\cdot n$ entries, $\E\lvert R_{4a}\rvert_F^2\le np\cdot n\cdot2(np)^2C_Dw_T^2=2n^4p^3C_Dw_T^2$. Since $w_T\to0$ by \eqref{eq:B-wT}, Chebyshev's inequality gives $P(\lvert R_{4a}\rvert_F>\epsilon)\le2n^4p^3C_Dw_T^2/\epsilon^2\to0$.

\emph{$R_{4b}$.} For each $t$, by $\lvert AB\rvert_F\le\lvert A\rvert_F\,\lVert B\rVert$ and $\lvert ab^\top\rvert_F=\lvert a\rvert\,\lvert b\rvert$,
\[
\lvert\bV_t\Delta_{T,t}^\top D_t\rvert_F\le\lvert\bV_t\Delta_{T,t}^\top\rvert_F\,\lVert D_t\rVert=\lvert\bV_t\rvert\,\lvert\Delta_{T,t}\rvert\,\lVert D_t\rVert\le w_T\lvert\bV_t\rvert\,\lvert\Delta_{T,t}\rvert,
\]
and in the same way $\lvert\Delta_{T,t}\bV_t^\top D_t\rvert_F\le w_T\lvert\bV_t\rvert\,\lvert\Delta_{T,t}\rvert$. By the triangle inequality and then Cauchy--Schwarz for averages,
\begin{align*}
\lvert R_{4b}\rvert_F
&\le\frac1{\sqrt{T_w}}\sum_t2w_T\lvert\bV_t\rvert\,\lvert\Delta_{T,t}\rvert
=2\sqrt{T_w}\,w_T\cdot\frac1{T_w}\sum_t\lvert\bV_t\rvert\,\lvert\Delta_{T,t}\rvert\\
&\le2\,\underbrace{\sqrt{T_w}\,w_T}_{\le C_w}\cdot\underbrace{\Bigl(\frac1{T_w}\sum_t\lvert\bV_t\rvert^2\Bigr)^{1/2}}_{O_p(1)}\cdot\underbrace{\Bigl(\frac1{T_w}\sum_t\lvert\Delta_{T,t}\rvert^2\Bigr)^{1/2}}_{o_p(1)}.
\end{align*}
The first factor is deterministic and bounded for large $T$ by \eqref{eq:B-wT}, which uses (W), so it is $O_p(1)$. The second is $O_p(1)$ and the third is $o_p(1)$, exactly as in Step (c) of Lemma~\ref{lem:B-gram} (Markov and (N1)). Applying $O_p(1)\cdot o_p(1)=o_p(1)$ twice, first to the product of the second and third factors and then to the first factor times that product, $\lvert R_{4b}\rvert_F=o_p(1)$.

\emph{$R_{4c}$.} For each $t$, by $\lvert AB\rvert_F\le\lvert A\rvert_F\,\lVert B\rVert$ and $\lvert ab^\top\rvert_F=\lvert a\rvert\,\lvert b\rvert$, $\lvert\Delta_{T,t}\Delta_{T,t}^\top D_t\rvert_F\le\lvert\Delta_{T,t}\rvert^2\lVert D_t\rVert\le w_T\lvert\Delta_{T,t}\rvert^2$. So
\[
\lvert R_{4c}\rvert_F\le\frac1{\sqrt{T_w}}\sum_tw_T\lvert\Delta_{T,t}\rvert^2
=\sqrt{T_w}\,w_T\cdot\frac1{T_w}\sum_t\lvert\Delta_{T,t}\rvert^2\le C_w\cdot\frac1{T_w}\sum_t\lvert\Delta_{T,t}\rvert^2
\]
for large $T$, by \eqref{eq:B-wT}. The average tends to $0$ in probability by \eqref{eq:N1a} in (N1), so $\lvert R_{4c}\rvert_F=o_p(1)$ (a bounded number times an $o_p(1)$ term).
\end{proof}

\subsubsection*{Step 4: the central limit theorem for the main term}

We use the following martingale central limit theorem. A \emph{martingale difference sequence} with respect to increasing $\sigma$-fields $\mathcal G_0\subseteq\mathcal G_1\subseteq\cdots$ is a sequence $x_1,x_2,\ldots$ with $x_k$ $\mathcal G_k$-measurable, integrable, and $\E[x_k\mid\mathcal G_{k-1}]=0$: a fair-game increment, whose expected value given everything seen so far is zero.

\begin{theorem}[Martingale central limit theorem]\label{thm:B-mclt}
This is \citet{billingsley1995probability}, Theorem~35.12 (p.~476), with $\sigma=1$.
For each $T$ let $x_{T,1},\ldots,x_{T,k_T}$ be real random variables and $\mathcal G_{T,0}\subseteq\mathcal G_{T,1}\subseteq\cdots\subseteq\mathcal G_{T,k_T}$ be $\sigma$-fields, with $\mathcal G_{T,0}=\{\emptyset,\Omega\}$ trivial. No relation between the $\sigma$-fields of different $T$ is required. Assume:
\begin{enumerate}[label=(\alph*)]
\item each $x_{T,k}$ is $\mathcal G_{T,k}$-measurable, $\E x_{T,k}^2<\infty$, and $\E[x_{T,k}\mid\mathcal G_{T,k-1}]=0$ almost surely;
\item $\sum_{k=1}^{k_T}\E[x_{T,k}^2\mid\mathcal G_{T,k-1}]\convp1$;
\item for every $\eta>0$, $\sum_{k=1}^{k_T}\E\bigl[x_{T,k}^2\mathbf 1\{\lvert x_{T,k}\rvert\ge\eta\}\bigr]\to0$ (Lindeberg's condition: no single term carries a non-negligible share of the variance).
\end{enumerate}
Then $\sum_{k=1}^{k_T}x_{T,k}\convd\N(0,1)$.
\end{theorem}

Two remarks match this statement to the book's. The book uses infinite rows and pads a finite row itself with $x_{T,k}:=0$, $\mathcal G_{T,k}:=\mathcal G_{T,k_T}$ for $k>k_T$; its side condition that the row sums converge almost surely then holds trivially. The book also allows any limit $\sigma^2>0$ in (b), with conclusion $\sum_kx_{T,k}\convd\sigma\,\N(0,1)$; we take $\sigma=1$. Its Lindeberg condition is exactly our (c).

\begin{lemma}[CLT for the main term]\label{lem:B-clt}
Assume (M), (S), (W), (D), (H). Fix an entry $(r,i)$, $r\in\{1,\dots,np\}$, $i\in\{1,\dots,n\}$. Let $a:=\bar J^{-1}e_r\in\R^{np}$ and
\[
v_{T,ri}:=\frac1{T_w}\sum_t\Sigma_{\varepsilon,ii}(t/T)\,a^\top J_ta .
\]
\begin{enumerate}[label=(\alph*)]
\item $c_v\le v_{T,ri}\le C_v$, with $c_v:=c_\varepsilon c_J/K_J^2$ and $C_v:=K_8^{1/4}K_J/c_J^2$.
\item $(\bar J^{-1}S_0)_{ri}/\sqrt{v_{T,ri}}\convd\N(0,1)$.
\item $\E\lvert S_0\rvert_F^2\le n^2p\,K_8^{1/2}$, so $S_0=O_p(1)$.
\end{enumerate}
\end{lemma}

In words, after multiplication by the fixed matrix $\bar J^{-1}$, each entry of the main term $S_0$ is asymptotically normal with an explicit variance, and $S_0$ as a whole stays bounded in probability. Part (b) is the result; parts (a) and (c) are bounds that later proofs also use (Proposition~\ref{prop:B-main} and Lemma~\ref{lem:C-se}).

\begin{proof}
\emph{Plan.} Part (b) is Theorem~\ref{thm:B-mclt} applied to the normalised summands $x_{T,k}=\xi_{t_k}/\sqrt{T_wv_{T,ri}}$ of Step (b1), where $\xi_t=a^\top\bV_t(\beps_t)_i$. We check the theorem's three conditions: condition (a), each summand is a fair-game increment, follows from (M) (Step (b2)); condition (b), the conditional variances add up to $1$, follows from the law of large numbers, Lemma~\ref{lem:A-lln} (Step (b3)); condition (c), Lindeberg's condition, follows from fourth moments (Step (b4)). Part (a) of the lemma comes first, so that we may divide by $v_{T,ri}$. Part (c) is a direct variance count (Step (c)).

\emph{Step (a): the bounds on $v_{T,ri}$.}
First, the size of $a$. Since $\lambda_{\min}(\bar J)\ge c_J$ (Lemma~\ref{lem:A-target}(i)), we have $\lVert\bar J^{-1}\rVert\le1/c_J$, so $\lvert a\rvert\le\lVert\bar J^{-1}\rVert\,\lvert e_r\rvert\le1/c_J$. Since $e_r=\bar Ja$ and $\lVert\bar J\rVert\le K_J$ (Lemma~\ref{lem:A-target}(i)), $1=\lvert e_r\rvert\le\lVert\bar J\rVert\,\lvert a\rvert\le K_J\lvert a\rvert$, so $\lvert a\rvert\ge1/K_J$.
Second, the quadratic form. By (S)(ii) and the spectral theorem, $a^\top J_ta\ge c_J\lvert a\rvert^2$, and by Lemma~\ref{lem:A-moments}, $a^\top J_ta\le\lVert J_t\rVert\,\lvert a\rvert^2\le K_J\lvert a\rvert^2$.
Third, the noise variance. By (H), $\Sigma_{\varepsilon,ii}(t/T)=e_i^\top\Sigma_\varepsilon(t/T)e_i\ge\lambda_{\min}(\Sigma_\varepsilon(t/T))\ge c_\varepsilon$. For the upper bound, take the $(i,i)$ entry in (H): $\E[(\beps_t)_i^2\mid\mathcal F_{t-1}]=\Sigma_{\varepsilon,ii}(t/T)$ a.s., and the right side is a deterministic number. Taking expectations (the expectation of a constant is the constant) gives the first equality below, the tower property $\E\bigl[\E[X\mid\mathcal F_{t-1}]\bigr]=\E X$ gives the second, and $(\beps_t)_i^2\le\lvert\beps_t\rvert^2$ and Lemma~\ref{lem:A-moments} give the two inequalities:
\begin{equation}\label{eq:B-sigma-bound}
\Sigma_{\varepsilon,ii}(t/T)=\E\bigl[\E[(\beps_t)_i^2\mid\mathcal F_{t-1}]\bigr]=\E(\beps_t)_i^2\le\E\lvert\beps_t\rvert^2\le K_8^{1/4}.
\end{equation}
Combining the three,
\[
v_{T,ri}\ge\frac1{T_w}\sum_tc_\varepsilon\,c_J\lvert a\rvert^2\ge c_\varepsilon c_J\frac1{K_J^2}=c_v,
\qquad
v_{T,ri}\le\frac1{T_w}\sum_tK_8^{1/4}K_J\lvert a\rvert^2\le K_8^{1/4}K_J\frac1{c_J^2}=C_v .
\]

\emph{Step (b1): the sum as a martingale.} Put $\xi_t:=a^\top\bV_t\,(\beps_t)_i$. The $(r,i)$ entry of a matrix $M$ is $M_{ri}=e_r^\top Me_i$: $e_r^\top$ picks row $r$ and $e_i$ picks column $i$. We split the product this way because each factor simplifies. Since $\bar J$ is symmetric, so is $\bar J^{-1}$, and $e_r^\top\bar J^{-1}=(\bar J^{-1}e_r)^\top=a^\top$. Also $S_0e_i=T_w^{-1/2}\sum_t\bV_t(\beps_t)_i$ (column $i$ of $S_0$). Hence
\[
\frac{(\bar J^{-1}S_0)_{ri}}{\sqrt{v_{T,ri}}}=\frac{e_r^\top\bar J^{-1}S_0e_i}{\sqrt{v_{T,ri}}}=\frac1{\sqrt{T_wv_{T,ri}}}\sum_t\xi_t .
\]
The vector $a$ depends on $T$ but is deterministic for each $T$, so by Lemma~\ref{lem:A-identities}(ii) $\{\xi_t\}$ is a martingale difference sequence with respect to $\{\mathcal F_t\}$: $\xi_t$ is $\mathcal F_t$-measurable, integrable, and $\E[\xi_t\mid\mathcal F_{t-1}]=0$. List the window as $t_k:=t_0-T_w+k$, $k=1,\ldots,T_w$, and set
\[
x_{T,k}:=\frac{\xi_{t_k}}{\sqrt{T_wv_{T,ri}}},\qquad \mathcal G_{T,0}:=\{\emptyset,\Omega\},\qquad\mathcal G_{T,k}:=\mathcal F_{t_k}\ (k\ge1).
\]
Then $\sum_{k=1}^{T_w}x_{T,k}=(\bar J^{-1}S_0)_{ri}/\sqrt{v_{T,ri}}$, and the $\sigma$-fields increase: the trivial $\sigma$-field is contained in every $\sigma$-field, and $\mathcal F_{t_{k-1}}\subseteq\mathcal F_{t_k}$. The index $k$ and the $\sigma$-fields $\mathcal G_{T,k}$ only relabel the window to match the notation of Theorem~\ref{thm:B-mclt}. We check conditions (a), (b), (c) of Theorem~\ref{thm:B-mclt}, with $k_T=T_w$.

Two moment bounds are needed. By Cauchy--Schwarz in $\R^{np}$, $\lvert a^\top\bV_t\rvert\le\lvert a\rvert\,\lvert\bV_t\rvert$, and $\lvert(\beps_t)_i\rvert\le\lvert\beps_t\rvert$. So, by Cauchy--Schwarz for expectations ($\E[UW]\le(\E U^2)^{1/2}(\E W^2)^{1/2}$ for $U,W\ge0$), the lower moments above, and $\lvert a\rvert\le1/c_J$,
\begin{align}
\E\xi_t^2&\le\lvert a\rvert^2\,\E\bigl[\lvert\bV_t\rvert^2\lvert\beps_t\rvert^2\bigr]\le\lvert a\rvert^2\bigl(\E\lvert\bV_t\rvert^4\bigr)^{1/2}\bigl(\E\lvert\beps_t\rvert^4\bigr)^{1/2}\le\frac{K_8^{1/4}\,K_8^{1/4}}{c_J^2}=\frac{K_8^{1/2}}{c_J^2},\label{eq:B-xi2}\\
\E\xi_t^4&\le\lvert a\rvert^4\,\E\bigl[\lvert\bV_t\rvert^4\lvert\beps_t\rvert^4\bigr]\le\lvert a\rvert^4\bigl(\E\lvert\bV_t\rvert^8\bigr)^{1/2}\bigl(\E\lvert\beps_t\rvert^8\bigr)^{1/2}\le\frac{K_8}{c_J^4}.\label{eq:B-xi4}
\end{align}

\emph{Step (b2): condition (a).} $x_{T,k}$ is a constant times $\xi_{t_k}$, hence $\mathcal F_{t_k}=\mathcal G_{T,k}$-measurable, and $\E x_{T,k}^2<\infty$ by \eqref{eq:B-xi2}. For $k\ge2$, $\mathcal G_{T,k-1}=\mathcal F_{t_k-1}$ (because $t_{k-1}=t_k-1$), so $\E[x_{T,k}\mid\mathcal G_{T,k-1}]=\E[\xi_{t_k}\mid\mathcal F_{t_k-1}]/\sqrt{T_wv_{T,ri}}=0$. For $k=1$, conditioning on the trivial $\sigma$-field is plain expectation, and by the tower property $\E x_{T,1}=\E\bigl[\E[\xi_{t_1}\mid\mathcal F_{t_1-1}]\bigr]/\sqrt{T_wv_{T,ri}}=0$.

\emph{Step (b3): condition (b), the conditional variance.} Let $q_t:=(a^\top\bV_t)^2\,\Sigma_{\varepsilon,ii}(t/T)$. Since $\bV_t$ is $\mathcal F_{t-1}$-measurable (Lemma~\ref{lem:A-identities}(i)), so is $(a^\top\bV_t)^2$; taking it out of the conditional expectation (the product $\xi_t^2$ is integrable by \eqref{eq:B-xi2}) and using (H),
\[
\E[\xi_t^2\mid\mathcal F_{t-1}]=(a^\top\bV_t)^2\,\E[(\beps_t)_i^2\mid\mathcal F_{t-1}]=(a^\top\bV_t)^2\,\Sigma_{\varepsilon,ii}(t/T)=q_t .
\]
So for $k\ge2$, $\E[x_{T,k}^2\mid\mathcal G_{T,k-1}]=q_{t_k}/(T_wv_{T,ri})$. For $k=1$, $\E[x_{T,1}^2\mid\mathcal G_{T,0}]=\E\xi_{t_1}^2/(T_wv_{T,ri})=\E q_{t_1}/(T_wv_{T,ri})$ by the tower property. (Theorem~\ref{thm:B-mclt}, like the book, fixes $\mathcal G_{T,0}$ trivial, so the first conditional variance is the number $\E q_{t_1}$, while the others are the random $q_{t_k}$.) Summing, and then adding and subtracting $q_{t_1}$ so that every $t\in\mathcal W$ appears in one sum,
\begin{equation}\label{eq:B-condvar}
\begin{aligned}
\sum_{k=1}^{T_w}\E[x_{T,k}^2\mid\mathcal G_{T,k-1}]
&=\frac{1}{T_wv_{T,ri}}\Bigl(\E q_{t_1}+\sum_{k=2}^{T_w}q_{t_k}\Bigr)\\
&=\underbrace{\frac{1}{T_wv_{T,ri}}\sum_tq_t}_{\text{first term}}+\underbrace{\frac{\E q_{t_1}-q_{t_1}}{T_wv_{T,ri}}}_{\text{second term}} .
\end{aligned}
\end{equation}
We show that the first term tends to $1$ (this is where the law of large numbers enters) and the second, a single leftover term, tends to $0$.
\emph{The first term.} Let
\[
\Gamma_T:=\frac1{T_w}\sum_t\Sigma_{\varepsilon,ii}(t/T)\bigl(\bV_t\bV_t^\top-J_t\bigr)\in\R^{np\times np}.
\]
Since $q_t=\Sigma_{\varepsilon,ii}(t/T)\,a^\top\bV_t\bV_t^\top a$ and $v_{T,ri}=T_w^{-1}\sum_t\Sigma_{\varepsilon,ii}(t/T)\,a^\top J_ta$,
\[
\frac1{T_w}\sum_tq_t-v_{T,ri}=a^\top\Gamma_Ta .
\]
Since $\E[\bV_t\bV_t^\top]=J_t$, each entry $(\Gamma_T)_{bc}=T_w^{-1}\sum_t\Sigma_{\varepsilon,ii}(t/T)\bigl((\bV_t)_b(\bV_t)_c-\E[(\bV_t)_b(\bV_t)_c]\bigr)$ is a weighted window average of a centred entry of $\bV_t\bV_t^\top$. The weights $\Sigma_{\varepsilon,ii}(t/T)$ are deterministic and lie in $(0,K_8^{1/4}]$ by \eqref{eq:B-sigma-bound}, so (D)(ii) and Lemma~\ref{lem:A-lln} with $B=K_8^{1/4}$ give $\E(\Gamma_T)_{bc}^2\le2K_8^{1/2}C_D/T_w$. Summing over the $(np)^2$ entries,
\[
\E\lvert\Gamma_T\rvert_F^2\le(np)^2\,\frac{2K_8^{1/2}C_D}{T_w}\longrightarrow0
\]
by (W), so $\lvert\Gamma_T\rvert_F\convp0$ by Chebyshev's inequality. By Cauchy--Schwarz, $\lVert M\rVert\le\lvert M\rvert_F$ and $\lvert a\rvert\le1/c_J$ (Step (a)), $\lvert a^\top\Gamma_Ta\rvert\le\lvert a\rvert\,\lvert\Gamma_Ta\rvert\le\lvert a\rvert^2\lvert\Gamma_T\rvert_F\le\lvert\Gamma_T\rvert_F/c_J^2$. Hence the first term of \eqref{eq:B-condvar} is
\[
\frac{v_{T,ri}+a^\top\Gamma_Ta}{v_{T,ri}}=1+\frac{a^\top\Gamma_Ta}{v_{T,ri}},\qquad
\Bigl\lvert\frac{a^\top\Gamma_Ta}{v_{T,ri}}\Bigr\rvert\le\frac{\lvert\Gamma_T\rvert_F}{c_J^2\,c_v}\convp0
\]
by part (a).
\emph{The second term.} $q_{t_1}\ge0$ and $\E q_{t_1}=\Sigma_{\varepsilon,ii}(t_1/T)\,a^\top J_{t_1}a\le C_v$ by the bounds of Step (a). So
\[
\E\Bigl\lvert\frac{\E q_{t_1}-q_{t_1}}{T_wv_{T,ri}}\Bigr\rvert\le\frac{\E q_{t_1}+\E q_{t_1}}{T_wc_v}\le\frac{2C_v}{c_vT_w}\longrightarrow0,
\]
and the second term is $o_p(1)$ by Markov's inequality. Adding, the sum \eqref{eq:B-condvar} tends to $1$ in probability.

\emph{Step (b4): condition (c), Lindeberg via fourth moments.} For real $x$ and $\eta>0$, $x^2\mathbf 1\{\lvert x\rvert\ge\eta\}\le x^4/\eta^2$: if $\lvert x\rvert<\eta$ the left side is $0$, and if $\lvert x\rvert\ge\eta$ then $1\le x^2/\eta^2$. Taking expectations and using \eqref{eq:B-xi4} and $v_{T,ri}\ge c_v$,
\[
\sum_{k=1}^{T_w}\E\bigl[x_{T,k}^2\mathbf 1\{\lvert x_{T,k}\rvert\ge\eta\}\bigr]
\le\frac1{\eta^2}\sum_{k=1}^{T_w}\frac{\E\xi_{t_k}^4}{T_w^2v_{T,ri}^2}
\le\frac1{\eta^2}\cdot T_w\cdot\frac{K_8}{c_J^4\,T_w^2\,c_v^2}
=\frac{K_8}{\eta^2c_J^4c_v^2}\cdot\frac1{T_w}\longrightarrow0
\]
by (W). Theorem~\ref{thm:B-mclt} now gives $\sum_kx_{T,k}\convd\N(0,1)$, which is (b).

\emph{Step (c): the second moment of $S_0$.} Fix an entry $(b,i)$. Then $(S_0)_{bi}=T_w^{-1/2}\sum_t(\bV_t)_b(\beps_t)_i$, and the summands $(\bV_t)_b(\beps_t)_i$ are the $\xi_t$ of Lemma~\ref{lem:A-identities}(ii) with $a=e_b$ (so $\lvert a\rvert=1$). By Lemma~\ref{lem:A-identities}(iii) they are uncorrelated (fair-game increments at different times are uncorrelated) with second moments at most $K_8^{1/2}$. Hence
\[
\E(S_0)_{bi}^2=\frac1{T_w}\,\E\Bigl(\sum_t(\bV_t)_b(\beps_t)_i\Bigr)^2=\frac1{T_w}\sum_t\E\bigl[(\bV_t)_b^2(\beps_t)_i^2\bigr]\le\frac1{T_w}\cdot T_w\cdot K_8^{1/2}=K_8^{1/2}.
\]
Summing over the $np\cdot n=n^2p$ entries, $\E\lvert S_0\rvert_F^2\le n^2p\,K_8^{1/2}$. By Chebyshev, $P(\lvert S_0\rvert_F>M)\le n^2pK_8^{1/2}/M^2$ for every $M>0$ and every $T$, so $S_0=O_p(1)$.
\end{proof}

\subsubsection*{Step 5: assembly}

The last tool is \emph{Slutsky's lemma} in its simplest form: if $N_T\convd Z$ and $U_T\convp0$, then $N_T+U_T\convd Z$. In words, adding a term that vanishes in probability does not change a limit in distribution.

\begin{proposition}[Asymptotic normality at the window target]\label{prop:B-main}
Assume (M), (S), (W), (D), (H), (N). For every entry $(r,i)$,
\[
\frac{\sqrt{T_w}\,(\hbTheta_{ri}-\bar\bTheta_{ri})}{\sqrt{v_{T,ri}}}\convd\N(0,1),
\qquad\text{and}\qquad
\sqrt{T_w}\,\lvert\hbTheta-\bar\bTheta\rvert_F=O_p(1).
\]
\end{proposition}

In words, each entry of the estimation error, scaled by $\sqrt{T_w}$ and by the deterministic standard deviation $\sqrt{v_{T,ri}}$, is asymptotically standard normal.

\begin{proof}
Let $E_T=\{\lambda_{\min}(\hat J)\ge c_J/2\}$ as in Lemma~\ref{lem:B-gram}, and $R:=R_1+R_2+R_3+R_{4a}+R_{4b}+R_{4c}$. By the triangle inequality, $\lvert R\rvert_F\le\lvert R_1\rvert_F+\lvert R_2\rvert_F+\lvert R_3\rvert_F+\lvert R_{4a}\rvert_F+\lvert R_{4b}\rvert_F+\lvert R_{4c}\rvert_F$; each term is $o_p(1)$ by Lemma~\ref{lem:B-remainders}, so $\lvert R\rvert_F\convp0$.

\emph{Where we work.} Recall from \eqref{eq:A-thetahat} that the formula for $\hbTheta$ needs $\hat J^{-1}$; where $\hat J$ is singular we set $\hbTheta:=0$ only so that $\hbTheta$ is defined at every outcome (any value would do). On $E_T$, $\hat J$ is invertible (Lemma~\ref{lem:B-gram}(ii)), so $\hbTheta$ is given by the formula and Lemma~\ref{lem:B-decomp} applies. Off $E_T$ we use no bound at all: $P(E_T^c)\to0$, so what happens there cannot affect a limit. The argument below therefore needs only bounds on $E_T$, plus the vanishing term $P(E_T^c)$.

\emph{Split the error.} On $E_T$, by Lemma~\ref{lem:B-decomp}, adding and subtracting $\bar J^{-1}S_0$,
\[
\sqrt{T_w}(\hbTheta-\bar\bTheta)=\hat J^{-1}(S_0+R)=\bar J^{-1}S_0+(\hat J^{-1}-\bar J^{-1})S_0+\hat J^{-1}R .
\]
Take entry $(r,i)$ and divide by $\sqrt{v_{T,ri}}$. Define, on the whole sample space,
\[
Z_T:=\frac{\sqrt{T_w}(\hbTheta_{ri}-\bar\bTheta_{ri})}{\sqrt{v_{T,ri}}},\qquad
 N_T:=\frac{(\bar J^{-1}S_0)_{ri}}{\sqrt{v_{T,ri}}},\qquad
 U_T:=Z_T-N_T .
\]
On $E_T$, $U_T=\bigl[(\hat J^{-1}-\bar J^{-1})S_0+\hat J^{-1}R\bigr]_{ri}/\sqrt{v_{T,ri}}$.

\emph{Bound $U_T$ on $E_T$.} By $\lvert M_{ri}\rvert\le\lVert M\rVert$ (an entry is at most the operator norm), then $\lVert AB\rVert\le\lVert A\rVert\,\lVert B\rVert$ and $\lVert M\rVert\le\lvert M\rvert_F$, then Lemma~\ref{lem:B-gram}(ii)--(iii) and $v_{T,ri}\ge c_v$ (Lemma~\ref{lem:B-clt}(a)),
\begin{align*}
\lvert U_T\rvert
&\le\frac{\lVert(\hat J^{-1}-\bar J^{-1})S_0\rVert+\lVert\hat J^{-1}R\rVert}{\sqrt{c_v}}
\le\frac{\lVert\hat J^{-1}-\bar J^{-1}\rVert\,\lvert S_0\rvert_F+\lVert\hat J^{-1}\rVert\,\lvert R\rvert_F}{\sqrt{c_v}}\\
&\le\frac1{\sqrt{c_v}}\Bigl(\frac2{c_J^2}\lVert\hat J-\bar J\rVert\,\lvert S_0\rvert_F+\frac2{c_J}\lvert R\rvert_F\Bigr)=:\bar\rho_T .
\end{align*}
The bound $\bar\rho_T$ is defined on the whole sample space and $\bar\rho_T\convp0$: $\lVert\hat J-\bar J\rVert=o_p(1)$ (Lemma~\ref{lem:B-gram}(i)) times $\lvert S_0\rvert_F=O_p(1)$ (Lemma~\ref{lem:B-clt}(c)) is $o_p(1)$ (as $O_p(1)\cdot o_p(1)=o_p(1)$); $\lvert R\rvert_F=o_p(1)$; and so the sum is $o_p(1)$. Therefore, for every $\epsilon>0$,
\[
P(\lvert U_T\rvert>\epsilon)\le P(E_T^c)+P(E_T\cap\{\lvert U_T\rvert>\epsilon\})\le P(E_T^c)+P(\bar\rho_T>\epsilon)\longrightarrow0,
\]
so $U_T\convp0$.

\emph{Conclude the first claim.} $N_T\convd\N(0,1)$ by Lemma~\ref{lem:B-clt}(b), and $U_T\convp0$. By Slutsky's lemma, $Z_T=N_T+U_T\convd\N(0,1)$.

\emph{The second claim.} On $E_T$, by Lemma~\ref{lem:B-decomp}, $\lvert AB\rvert_F\le\lVert A\rVert\,\lvert B\rvert_F$, the triangle inequality and Lemma~\ref{lem:B-gram}(ii),
\[
\sqrt{T_w}\,\lvert\hbTheta-\bar\bTheta\rvert_F=\lvert\hat J^{-1}(S_0+R)\rvert_F\le\lVert\hat J^{-1}\rVert\bigl(\lvert S_0\rvert_F+\lvert R\rvert_F\bigr)\le\frac2{c_J}\bigl(\lvert S_0\rvert_F+\lvert R\rvert_F\bigr).
\]
Given $\delta>0$, choose $M$ with $16\,n^2pK_8^{1/2}/(c_J^2M^2)\le\delta$. If the right side exceeds $M$, then $\lvert S_0\rvert_F>c_JM/4$ or $\lvert R\rvert_F>c_JM/4$. So, using Chebyshev with Lemma~\ref{lem:B-clt}(c) for $S_0$,
\[
P\bigl(\sqrt{T_w}\lvert\hbTheta-\bar\bTheta\rvert_F>M\bigr)
\le P(E_T^c)+\frac{16\,n^2pK_8^{1/2}}{c_J^2M^2}+P\bigl(\lvert R\rvert_F>c_JM/4\bigr)
\le\delta+o(1),
\]
since $P(E_T^c)\to0$ and $\lvert R\rvert_F\convp0$. Hence the probability is at most $2\delta$ for all large $T$, which is the definition of $O_p(1)$.
\end{proof}

Proposition~\ref{prop:B-main} divides by the deterministic scale $\sqrt{v_{T,ri}}$, which is unknown in practice. Appendix~\ref{app:se} shows that $T_w\,\hat s_{ri}^2-v_{T,ri}\convp0$, which completes the proof of Theorem~\ref{thm:dml-var}.

\subsection{Standard errors, and completion of the proof of Theorem~\ref{thm:dml-var}}\label{app:se}

\paragraph{Roadmap.}
Appendix~\ref{app:thm1-proof} showed that $\sqrt{T_w}(\hbTheta_{ri}-\bar\bTheta_{ri})$ is approximately $\N(0,v_{T,ri})$ (Proposition~\ref{prop:B-main}). The code does not know $v_{T,ri}$; it divides by its own standard error $\hat s_{ri}$ from \eqref{eq:A-se}. This section shows that $T_w\,\hat s_{ri}^2$ is close to $v_{T,ri}$, so the division gives the right scale. The steps, in order:
\begin{itemize}[nosep]
\item Lemma~\ref{lem:C-sigma}: the residual variance $\hat\sigma_i^2$ is close to the average true noise variance $\bar\sigma_i^2$.
\item Lemma~\ref{lem:C-se}: $\hat\sigma_i^2[\hat J^{-1}]_{rr}$ is then close to $v_{T,ri}$, because under (H) the noise variance changes only slowly over the window.
\item Completion of the proof of Theorem~\ref{thm:dml-var}: Slutsky's lemma turns Proposition~\ref{prop:B-main} into the theorem.
\end{itemize}

\paragraph{Conventions.}
All statements are for $T$ large enough that $T_w>np$, which holds eventually because $T_w\to\infty$ by (W). The formula \eqref{eq:A-se} for $\hat s_{ri}$ (the main-text formula \eqref{eq:ols-se}) needs $\hat J$ to be invertible, and it gives a positive number only if $\hat\sigma_i^2>0$. We use the convention of Appendix~\ref{app:assumptions}: if $\hat J$ is singular or the formula gives $\hat s_{ri}=0$, then $\hat s_{ri}:=1$ and $z_{ri}:=0$. So $\hat s_{ri}>0$ always, and every ratio with $\hat s_{ri}$ in the denominator is defined. The convention never matters for a limit: it can only act off the event $\tilde E_T$ of Lemma~\ref{lem:C-se}, whose probability tends to $1$ and on which the formula holds and is positive, and what happens off an event whose probability tends to one never affects a limit in probability. For a random matrix, ``$=O_p(1)$'' and ``$\convp0$'' are meant entry by entry, as in \emph{Limits} in Appendix~\ref{app:assumptions}; this is the same as saying that the Frobenius norm is $O_p(1)$, or $\convp0$.

Standard matrix and probability rules are named where they are used.

\begin{lemma}[The residual variance is consistent]\label{lem:C-sigma}

Let $\bar\sigma_i^2:=T_w^{-1}\sum_{t\in\mathcal W}\Sigma_{\varepsilon,ii}(t/T)$ and $K_J$ the constant of Lemma~\ref{lem:A-moments}. Under Assumption~\ref{asm:main}, for every $i\in\{1,\ldots,n\}$,
\[
c_\varepsilon\le\bar\sigma_i^2\le K_J\qquad\text{and}\qquad\hat\sigma_i^2-\bar\sigma_i^2\convp0 .
\]
\end{lemma}

In words, the code's residual variance for equation $i$ settles down to the average true noise variance of equation $i$ over the window.

\begin{proof}
Fix $i$. We write $\langle x\rangle:=T_w^{-1}\sum_{t\in\mathcal W}x_t$ for the window average of a sequence $x_t$. The plan: the residual is the true noise plus a remainder (Step 1); the remainder is small because $\hbTheta$ is close to every $\bTheta_t$ in the window and the nuisance errors are small (Steps 2--3); the squared true noise averages to $\bar\sigma_i^2$ (Step 4); and the cross term and the degrees-of-freedom factor are then handled (Steps 5--6). Step 0 records bounds on $\bar\sigma_i^2$ used later.

\emph{Step 0 (bounds on $\bar\sigma_i^2$).} Fix $t\in\mathcal W$ and put $u:=t/T\in[0,1]$. By (H), $\Sigma_\varepsilon(u)$ equals the symmetric matrix $\E[\beps_t\beps_t^\top\mid\mathcal F_{t-1}]$ almost surely, so $\Sigma_\varepsilon(u)$ is symmetric. By (H) and $e_i^\top Ae_i\ge\lambda_{\min}(A)$ for a symmetric $A$, $\Sigma_{\varepsilon,ii}(u)=e_i^\top\Sigma_\varepsilon(u)e_i\ge\lambda_{\min}(\Sigma_\varepsilon(u))\ge c_\varepsilon$. Taking expectations in (H) (tower property), then Lemma~\ref{lem:A-moments},
\[
\Sigma_{\varepsilon,ii}(u)=\E(\beps_t)_i^2\le\E\lvert\beps_t\rvert^2\le K_J .
\]
Since $\bar\sigma_i^2$ is the average of these numbers over $t\in\mathcal W$, $c_\varepsilon\le\bar\sigma_i^2\le K_J$.

\emph{Step 1 (residual identity).} Recall \eqref{eq:A-tildeY}: $\tilde\bY_t=\bTheta_t^\top\tilde\bT_t+\Delta_{g,t}+\beps_t$. Substituting this into the residual $\hat e_t=\tilde\bY_t-\hbTheta^\top\tilde\bT_t$ of \eqref{eq:A-sigmahat} gives
\begin{align*}
\hat e_t&=\bTheta_t^\top\tilde\bT_t+\Delta_{g,t}+\beps_t-\hbTheta^\top\tilde\bT_t
=\beps_t+\rho_t,\qquad\text{where}\quad\rho_t:=\Delta_{g,t}+(\bTheta_t-\hbTheta)^\top\tilde\bT_t .
\end{align*}
This is pure algebra on the probability-one event of Lemma~\ref{lem:A-identities} and holds for every value of $\hbTheta$, including the convention $\hbTheta=0$ when $\hat J$ is singular. Taking coordinate $i$: $(\hat e_t)_i=(\beps_t)_i+(\rho_t)_i$.

\emph{Step 2 (the coefficient gap is uniformly small).} Let $\delta_T:=\max_{t\in\mathcal W}\lVert\bTheta_t-\hbTheta\rVert$. For each $t\in\mathcal W$, by the triangle inequality, the definition $D_t=\bTheta_t-\bar\bTheta$, and Lemma~\ref{lem:A-target}(iii),
\[
\lVert\bTheta_t-\hbTheta\rVert\le\lVert\bTheta_t-\bar\bTheta\rVert+\lVert\bar\bTheta-\hbTheta\rVert=\lVert D_t\rVert+\lVert\hbTheta-\bar\bTheta\rVert\le w_T+\lVert\hbTheta-\bar\bTheta\rVert .
\]
The right side does not depend on $t$, so it also bounds $\delta_T$. Since $\lVert M\rVert\le\lvert M\rvert_F$,
\[
\lVert\hbTheta-\bar\bTheta\rVert\le\lvert\hbTheta-\bar\bTheta\rvert_F=T_w^{-1/2}\,Z_T,\qquad Z_T:=\bigl\lvert\sqrt{T_w}\,(\hbTheta-\bar\bTheta)\bigr\rvert_F .
\]
Hence $0\le\delta_T\le w_T+T_w^{-1/2}Z_T$. Here $Z_T=O_p(1)$ by Proposition~\ref{prop:B-main}. Both $w_T=L(K_J/c_J)\,T_w/T$ and $T_w^{-1/2}$ are deterministic and tend to $0$ by (W) (which gives $T_w/T\to0$ and $T_w\to\infty$). So $T_w^{-1/2}Z_T\convp0$ (a vanishing number times an $O_p(1)$ term), and hence $\delta_T\convp0$.

\emph{Step 3 (the remainder is small on average).} Since $(\rho_t)_i$ is a coordinate of $\rho_t$, $(\rho_t)_i^2\le\lvert\rho_t\rvert^2$. Expanding $\rho_t$ and using $\lvert x+y\rvert^2\le2\lvert x\rvert^2+2\lvert y\rvert^2$, then $\lvert M^\top x\rvert\le\lVert M\rVert\,\lvert x\rvert$ with the definition of $\delta_T$, then $\tilde\bT_t=\bV_t-\Delta_{T,t}$ from \eqref{eq:A-tildeT} and the same inequality again:
\begin{align*}
\lvert\rho_t\rvert^2&\le2\lvert\Delta_{g,t}\rvert^2+2\bigl\lvert(\bTheta_t-\hbTheta)^\top\tilde\bT_t\bigr\rvert^2\\
&\le2\lvert\Delta_{g,t}\rvert^2+2\delta_T^2\lvert\tilde\bT_t\rvert^2\\
&=2\lvert\Delta_{g,t}\rvert^2+2\delta_T^2\lvert\bV_t-\Delta_{T,t}\rvert^2\\
&\le2\lvert\Delta_{g,t}\rvert^2+4\delta_T^2\lvert\bV_t\rvert^2+4\delta_T^2\lvert\Delta_{T,t}\rvert^2 .
\end{align*}
Averaging over $t\in\mathcal W$ ($\delta_T$ does not depend on $t$):
\begin{equation}\label{eq:C-rho}
\langle\rho_{\cdot,i}^2\rangle\le2\langle\lvert\Delta_{g,\cdot}\rvert^2\rangle+4\delta_T^2\langle\lvert\bV_\cdot\rvert^2\rangle+4\delta_T^2\langle\lvert\Delta_{T,\cdot}\rvert^2\rangle .
\end{equation}
We treat the three averages in turn.
\begin{enumerate}[label=(\roman*)]
\item $\langle\lvert\Delta_{T,\cdot}\rvert^2\rangle\convp0$ by \eqref{eq:N1a} in (N1).
\item $\langle\lvert\Delta_{g,\cdot}\rvert^2\rangle\convp0$. In words: the outcome-equation error $\Delta_{g,t}=\bTheta_t^\top\Delta_{T,t}-\Delta_{Y,t}$ is built from the two learner errors with bounded coefficients, so it is small when they are. In detail, let $\lVert\bTheta\rVert_\infty:=\sup_{u\in[0,1]}\lVert\bTheta(u)\rVert$; by (S)(i), $\lVert\bTheta(u)\rVert\le\lVert\bTheta(0)\rVert+Lu$, so $\lVert\bTheta\rVert_\infty\le\lVert\bTheta(0)\rVert+L<\infty$. By the triangle inequality, $\lvert M^\top x\rvert\le\lVert M\rVert\,\lvert x\rvert$ and $(x+y)^2\le2x^2+2y^2$,
\[
  \lvert\Delta_{g,t}\rvert^2\le\bigl(\lVert\bTheta\rVert_\infty\lvert\Delta_{T,t}\rvert+\lvert\Delta_{Y,t}\rvert\bigr)^2\le C_\Theta\bigl(\lvert\Delta_{T,t}\rvert^2+\lvert\Delta_{Y,t}\rvert^2\bigr),\qquad C_\Theta:=2\max(\lVert\bTheta\rVert_\infty^2,1).
\]
Averaging, $0\le\langle\lvert\Delta_{g,\cdot}\rvert^2\rangle\le C_\Theta\bigl(\langle\lvert\Delta_{T,\cdot}\rvert^2\rangle+\langle\lvert\Delta_{Y,\cdot}\rvert^2\rangle\bigr)$. Both averages on the right tend to $0$ in probability, by \eqref{eq:N1a} and \eqref{eq:N1b} in (N1), so the right side does too, and hence so does the left side.
\item $\langle\lvert\bV_\cdot\rvert^2\rangle=O_p(1)$. By Lemma~\ref{lem:A-moments}, $\E\lvert\bV_t\rvert^2\le K_J$ for every $t$, so by linearity $\E\langle\lvert\bV_\cdot\rvert^2\rangle=\langle\E\lvert\bV_\cdot\rvert^2\rangle\le K_J$ for every $T$. So it is $O_p(1)$ by Markov's inequality.
\end{enumerate}
Since $\delta_T\convp0$ (Step 2), also $\delta_T^2\convp0$. Since $o_p(1)\cdot O_p(1)=o_p(1)$, $\delta_T^2\langle\lvert\bV_\cdot\rvert^2\rangle\convp0$ (using (iii)) and $\delta_T^2\langle\lvert\Delta_{T,\cdot}\rvert^2\rangle\convp0$ (using (i); a term tending to $0$ in probability is $O_p(1)$). With (ii), the right side of \eqref{eq:C-rho} tends to $0$ in probability, hence so does $\langle\rho_{\cdot,i}^2\rangle$.

\emph{Step 4 (the true noise averages to $\bar\sigma_i^2$).} First, the mean of each term. The variable $(\beps_t)_i^2\le\lvert\beps_t\rvert^2$ is integrable, since $\E\lvert\beps_t\rvert^2\le K_J$ by Lemma~\ref{lem:A-moments}. By the tower property and (H) (the conditional expectation of the $(i,i)$ entry is the $(i,i)$ entry of the conditional expectation),
\[
\E\,(\beps_t)_i^2=\E\bigl[\E[(\beps_t)_i^2\mid\mathcal F_{t-1}]\bigr]=\E\bigl[\Sigma_{\varepsilon,ii}(t/T)\bigr]=\Sigma_{\varepsilon,ii}(t/T),
\]
the last step because $\Sigma_\varepsilon$ is deterministic. Averaging, $\langle\E\,\varepsilon_{\cdot,i}^2\rangle=\bar\sigma_i^2$. Next, $f_t:=(\beps_t)_i^2$ is the $(i,i)$ entry of $\beps_t\beps_t^\top$, so it satisfies the hypotheses of Lemma~\ref{lem:A-lln} (finite second moment by Lemma~\ref{lem:A-moments}, covariance bound by (D)(ii); see the last sentence of Lemma~\ref{lem:A-lln}). Lemma~\ref{lem:A-lln} with weights $b_t\equiv1$ (so $B=1$) gives
\[
\E\bigl(\langle\varepsilon_{\cdot,i}^2\rangle-\bar\sigma_i^2\bigr)^2=\E\Bigl(\frac1{T_w}\sum_{t\in\mathcal W}\bigl((\beps_t)_i^2-\E\,(\beps_t)_i^2\bigr)\Bigr)^2\le\frac{2C_D}{T_w}.
\]
By Chebyshev's inequality, for every $\epsilon>0$,
\[
P\bigl(\lvert\langle\varepsilon_{\cdot,i}^2\rangle-\bar\sigma_i^2\rvert>\epsilon\bigr)\le\frac{\E\bigl(\langle\varepsilon_{\cdot,i}^2\rangle-\bar\sigma_i^2\bigr)^2}{\epsilon^2}\le\frac{2C_D}{T_w\,\epsilon^2}\longrightarrow0
\]
by (W). So $\langle\varepsilon_{\cdot,i}^2\rangle-\bar\sigma_i^2\convp0$. Since $0\le\bar\sigma_i^2\le K_J$ (Step 0), also $\langle\varepsilon_{\cdot,i}^2\rangle=O_p(1)$ (a bounded number plus an $o_p(1)$ term).

\emph{Step 5 (the residual average).} Square the identity of Step 1, $(\hat e_t)_i^2=((\beps_t)_i+(\rho_t)_i)^2=(\beps_t)_i^2+2(\beps_t)_i(\rho_t)_i+(\rho_t)_i^2$, and average:
\[
\langle\hat e_{\cdot,i}^2\rangle-\bar\sigma_i^2=\bigl(\langle\varepsilon_{\cdot,i}^2\rangle-\bar\sigma_i^2\bigr)+2\langle\varepsilon_{\cdot,i}\rho_{\cdot,i}\rangle+\langle\rho_{\cdot,i}^2\rangle .
\]
The first term $\convp0$ by Step 4 and the last $\convp0$ by Step 3. For the middle term, Cauchy--Schwarz for averages gives
\[
\bigl\lvert\langle\varepsilon_{\cdot,i}\rho_{\cdot,i}\rangle\bigr\rvert\le\langle\varepsilon_{\cdot,i}^2\rangle^{1/2}\,\langle\rho_{\cdot,i}^2\rangle^{1/2},
\]
where the first factor is $O_p(1)$ (Step 4, taking square roots) and the second $\convp0$ (Step 3, taking square roots). So the product $\convp0$ ($O_p(1)\cdot o_p(1)=o_p(1)$), and so does the middle term. Adding up: $\langle\hat e_{\cdot,i}^2\rangle-\bar\sigma_i^2\convp0$. As in Step 4, it follows that $\langle\hat e_{\cdot,i}^2\rangle=O_p(1)$.

\emph{Step 6 (degrees of freedom).} By definition $\hat\sigma_i^2=(T_w-np)^{-1}\sum_{t\in\mathcal W}(\hat e_t)_i^2=\frac{T_w}{T_w-np}\langle\hat e_{\cdot,i}^2\rangle$. The $-np$ is the usual OLS degrees-of-freedom correction ($np$ regressors per equation), and it is what the code uses; for the limit any fixed shift $c$ would do, since $T_w/(T_w-c)\to1$. Hence
\begin{align*}
\hat\sigma_i^2-\bar\sigma_i^2&=\Bigl(\frac{T_w}{T_w-np}-1\Bigr)\langle\hat e_{\cdot,i}^2\rangle+\bigl(\langle\hat e_{\cdot,i}^2\rangle-\bar\sigma_i^2\bigr)
=\frac{np}{T_w-np}\,\langle\hat e_{\cdot,i}^2\rangle+\bigl(\langle\hat e_{\cdot,i}^2\rangle-\bar\sigma_i^2\bigr).
\end{align*}
The deterministic factor $np/(T_w-np)\to0$ by (W), and $\langle\hat e_{\cdot,i}^2\rangle=O_p(1)$, so the first term $\convp0$. The second $\convp0$ by Step 5. Adding, $\hat\sigma_i^2-\bar\sigma_i^2\convp0$.
\end{proof}

\begin{lemma}[The standard error has the right scale]\label{lem:C-se}

Assume Assumption~\ref{asm:main}. Fix an entry $(r,i)$, let $a=\bar J^{-1}e_r$ and $v_{T,ri}=T_w^{-1}\sum_{t\in\mathcal W}\Sigma_{\varepsilon,ii}(t/T)\,a^\top J_ta$ as in Lemma~\ref{lem:B-clt}, and let $c_v\le v_{T,ri}\le C_v$ be the bounds of Lemma~\ref{lem:B-clt}(a). With $E_T=\{\lambda_{\min}(\hat J)\ge c_J/2\}$ the event of Lemma~\ref{lem:B-gram}, let
\[
\tilde E_T:=E_T\cap\bigl\{\hat\sigma_i^2\ge c_\varepsilon/2\bigr\},\qquad
G_T:=\tilde E_T\cap\bigl\{\lvert T_w\,\hat s_{ri}^2-v_{T,ri}\rvert\le c_v/2\bigr\}.
\]
\begin{enumerate}[label=(\alph*)]
\item $T_w\,\hat s_{ri}^2-v_{T,ri}\convp0$.
\item $P(\tilde E_T)\to1$ and $P(G_T)\to1$. On $\tilde E_T$ (so also on $G_T\subseteq \tilde E_T$), $\hat J$ is invertible and $T_w\,\hat s_{ri}^2=\hat\sigma_i^2[\hat J^{-1}]_{rr}>0$, so the convention for $\hat s_{ri}$ and $z_{ri}$ is not used. On $G_T$, moreover, $\sqrt{T_w}\,\hat s_{ri}\ge\sqrt{c_v/2}$.
\item $\mathbf 1_{G_T}\Bigl\lvert\dfrac{\sqrt{v_{T,ri}}}{\sqrt{T_w}\,\hat s_{ri}}-1\Bigr\rvert\convp0$.
\end{enumerate}
\end{lemma}

In words, $\sqrt{T_w}$ times the code's standard error is close to $\sqrt{v_{T,ri}}$, the true standard deviation from Lemma~\ref{lem:B-clt}, and it stays bounded away from zero.

\begin{proof}
Throughout the proof we write
\begin{equation}\label{eq:C-notation}
v:=v_{T,ri},\qquad
h:=[\hat J^{-1}]_{rr}=e_r^\top\hat J^{-1}e_r\ \ (\text{on }\tilde E_T),\qquad
\bar h:=[\bar J^{-1}]_{rr}=e_r^\top\bar J^{-1}e_r ,
\end{equation}
so $h$ and $\bar h$ are the $(r,r)$ diagonal entries of the two inverses, written also as quadratic forms in the unit vector $e_r$ (as in Lemma~\ref{lem:B-clt}); we use whichever form is shorter. The plan: $T_w\,\hat s_{ri}^2=\hat\sigma_i^2h$ on the good event $\tilde E_T$ (Step 0). We compare $\hat\sigma_i^2h$ with $v$ through the middle term $\bar\sigma_i^2\bar h$: the gap $v-\bar\sigma_i^2\bar h$ is deterministic and small because the noise variance changes slowly (Steps 1--2), and the gap $\hat\sigma_i^2h-\bar\sigma_i^2\bar h$ is small in probability by Lemmas~\ref{lem:C-sigma} and~\ref{lem:B-gram} (Step 3). Steps 4--6 then read off (a), (b) and (c).

\emph{Step 0 (the event $\tilde E_T$).} First we show that $P(\tilde E_T)\to1$, by showing that the complement $\tilde E_T^c=E_T^c\cup\{\hat\sigma_i^2<c_\varepsilon/2\}$ has probability tending to $0$. Lemma~\ref{lem:C-sigma} gives $\bar\sigma_i^2\ge c_\varepsilon$, so if $\hat\sigma_i^2<c_\varepsilon/2$ then $\bar\sigma_i^2-\hat\sigma_i^2>c_\varepsilon-c_\varepsilon/2=c_\varepsilon/2$. Thus $\{\hat\sigma_i^2<c_\varepsilon/2\}\subseteq\{\lvert\hat\sigma_i^2-\bar\sigma_i^2\rvert>c_\varepsilon/2\}$, and by a union bound
\[
P(\tilde E_T^c)\le P(E_T^c)+P\bigl(\lvert\hat\sigma_i^2-\bar\sigma_i^2\rvert>c_\varepsilon/2\bigr)\longrightarrow0,
\]
by Lemma~\ref{lem:B-gram}(ii) for the first term and Lemma~\ref{lem:C-sigma} for the second. Second, what holds on $\tilde E_T$. The matrix $\hat J=T_w^{-1}\sum_{t\in\mathcal W}\tilde\bT_t\tilde\bT_t^\top$ is symmetric with $\lambda_{\min}(\hat J)\ge c_J/2$, so (by its spectral decomposition) it is invertible, $\lVert\hat J^{-1}\rVert\le2/c_J$, and $h=e_r^\top\hat J^{-1}e_r\ge\lambda_{\min}(\hat J^{-1})>0$ (by \eqref{eq:C-notation}, and since a quadratic form in a unit vector is at least the smallest eigenvalue). Hence the formula \eqref{eq:A-se} applies and gives $\hat\sigma_i^2h/T_w\ge(c_\varepsilon/2)h/T_w>0$, so the convention is not used and
\begin{equation}\label{eq:C-seformula}
T_w\,\hat s_{ri}^2=\hat\sigma_i^2\,h\qquad\text{on }\tilde E_T .
\end{equation}

\emph{Step 1 (rewrite $\bar\sigma_i^2\bar h$ as a window average).} Since $\bar J$ is symmetric, so is $\bar J^{-1}$, and $a^\top=e_r^\top\bar J^{-1}$. Hence
\[
a^\top\bar Ja=e_r^\top\bar J^{-1}\bar J\bar J^{-1}e_r=e_r^\top\bar J^{-1}e_r=\bar h \quad\text{by \eqref{eq:C-notation}}.
\]
Since $\bar J=T_w^{-1}\sum_{t\in\mathcal W}J_t$, linearity gives $a^\top\bar Ja=T_w^{-1}\sum_{t\in\mathcal W}a^\top J_ta$. Therefore
\[
\bar\sigma_i^2\,\bar h=\frac1{T_w}\sum_{t\in\mathcal W}\bar\sigma_i^2\,a^\top J_ta .
\]

\emph{Step 2 (the deterministic gap $v-\bar\sigma_i^2\bar h$).} Subtracting Step 1 from the definition of $v$,
\[
v-\bar\sigma_i^2\bar h=\frac1{T_w}\sum_{t\in\mathcal W}\bigl(\Sigma_{\varepsilon,ii}(t/T)-\bar\sigma_i^2\bigr)\,a^\top J_ta .
\]
We bound the two factors. For $t\in\mathcal W$, since $\bar\sigma_i^2$ is the average of $\Sigma_{\varepsilon,ii}(t'/T)$ over $t'\in\mathcal W$,
\[
\Sigma_{\varepsilon,ii}(t/T)-\bar\sigma_i^2=\frac1{T_w}\sum_{t'\in\mathcal W}\bigl(\Sigma_{\varepsilon,ii}(t/T)-\Sigma_{\varepsilon,ii}(t'/T)\bigr).
\]
By $\lvert M_{ii}\rvert\le\lVert M\rVert$ and the Lipschitz bound in (H), each summand satisfies
\[
\lvert\Sigma_{\varepsilon,ii}(t/T)-\Sigma_{\varepsilon,ii}(t'/T)\rvert\le\lVert\Sigma_\varepsilon(t/T)-\Sigma_\varepsilon(t'/T)\rVert\le L_\varepsilon\frac{\lvert t-t'\rvert}{T}\le L_\varepsilon\frac{T_w-1}{T},
\]
since two indices in $\mathcal W$ differ by at most $T_w-1$. An average of $T_w$ numbers of absolute value at most $L_\varepsilon(T_w-1)/T$ has absolute value at most $L_\varepsilon(T_w-1)/T$, so
\[
\lvert\Sigma_{\varepsilon,ii}(t/T)-\bar\sigma_i^2\rvert\le L_\varepsilon\frac{T_w-1}{T}\le L_\varepsilon\frac{T_w}{T}\qquad(t\in\mathcal W).
\]
The second factor is nonnegative: $a$ is deterministic and $J_t=\E[\bV_t\bV_t^\top]$, so $a^\top J_ta=\E[(a^\top\bV_t)^2]\ge0$. Hence
\begin{align*}
\lvert v-\bar\sigma_i^2\bar h\rvert
&\le\frac1{T_w}\sum_{t\in\mathcal W}\lvert\Sigma_{\varepsilon,ii}(t/T)-\bar\sigma_i^2\rvert\,a^\top J_ta
\le L_\varepsilon\frac{T_w}{T}\cdot\frac1{T_w}\sum_{t\in\mathcal W}a^\top J_ta
=L_\varepsilon\frac{T_w}{T}\,\bar h,
\end{align*}
using Step 1 in the last equality. Finally $\bar h\le\lVert\bar J^{-1}\rVert\le1/c_J$ by $\lvert M_{rr}\rvert\le\lVert M\rVert$ and Lemma~\ref{lem:A-target}(i). So
\begin{equation}\label{eq:C-drift}
\lvert v-\bar\sigma_i^2\bar h\rvert\le\frac{L_\varepsilon}{c_J}\cdot\frac{T_w}{T}\longrightarrow0,
\end{equation}
because $T_w/T\to0$ by (W).

\emph{Step 3 (the random gap, on $\tilde E_T$).} On $\tilde E_T$, by \eqref{eq:C-seformula}, $T_w\,\hat s_{ri}^2=\hat\sigma_i^2h$. Add and subtract $\bar\sigma_i^2h$:
\[
\hat\sigma_i^2h-\bar\sigma_i^2\bar h=\hat\sigma_i^2h-\bar\sigma_i^2h+\bar\sigma_i^2h-\bar\sigma_i^2\bar h=(\hat\sigma_i^2-\bar\sigma_i^2)\,h+\bar\sigma_i^2\,(h-\bar h).
\]
On $\tilde E_T$, $\lvert h\rvert\le\lVert\hat J^{-1}\rVert\le2/c_J$ (an entry is at most the operator norm; Step 0). Also $\tilde E_T\subseteq E_T$, so Lemma~\ref{lem:B-gram}(iii) applies on $\tilde E_T$ and, with the same fact,
\[
\lvert h-\bar h\rvert=\bigl\lvert[\hat J^{-1}-\bar J^{-1}]_{rr}\bigr\rvert\le\lVert\hat J^{-1}-\bar J^{-1}\rVert\le\frac2{c_J^2}\,\lVert\hat J-\bar J\rVert .
\]
With $0<\bar\sigma_i^2\le K_J$ (Lemma~\ref{lem:C-sigma}) and the triangle inequality:
\[
\mathbf 1_{\tilde E_T}\bigl\lvert T_w\,\hat s_{ri}^2-\bar\sigma_i^2\bar h\bigr\rvert\le\frac2{c_J}\,\lvert\hat\sigma_i^2-\bar\sigma_i^2\rvert+\frac{2K_J}{c_J^2}\,\lVert\hat J-\bar J\rVert .
\]
The first term $\convp0$ by Lemma~\ref{lem:C-sigma} (times the constant $2/c_J$). The second term $\convp0$ by Lemma~\ref{lem:B-gram}(i) (times the constant $2K_J/c_J^2$). Adding, $\mathbf 1_{\tilde E_T}\lvert T_w\,\hat s_{ri}^2-\bar\sigma_i^2\bar h\rvert\convp0$.

\emph{Step 4 (proof of (a)).} By the triangle inequality,
\[
\mathbf 1_{\tilde E_T}\lvert T_w\,\hat s_{ri}^2-v\rvert\le\mathbf 1_{\tilde E_T}\lvert T_w\,\hat s_{ri}^2-\bar\sigma_i^2\bar h\rvert+\lvert\bar\sigma_i^2\bar h-v\rvert .
\]
The first term $\convp0$ by Step 3; the second is deterministic and $\to0$ by \eqref{eq:C-drift}. So the left side $\convp0$. Since $P(\tilde E_T)\to1$ by Step 0, and what happens off $\tilde E_T$ does not affect the limit, we get $T_w\,\hat s_{ri}^2-v\convp0$.

\emph{Step 5 (proof of (b)).} $P(\tilde E_T)\to1$ and the statement about $\tilde E_T$ were shown in Step 0. By a union bound,
\[
P(G_T^c)\le P(\tilde E_T^c)+P\bigl(\lvert T_w\,\hat s_{ri}^2-v\rvert>c_v/2\bigr)\longrightarrow0,
\]
by Step 0 and part (a). On $G_T$, using $v\ge c_v$, $T_w\,\hat s_{ri}^2\ge v-c_v/2\ge c_v-c_v/2=c_v/2$; taking square roots gives $\sqrt{T_w}\,\hat s_{ri}\ge\sqrt{c_v/2}$.

\emph{Step 6 (proof of (c)).} On $G_T$ put $x:=\sqrt{T_w}\,\hat s_{ri}$, so $x\ge\sqrt{c_v/2}>0$ by (b), and $\sqrt v\ge\sqrt{c_v}$. Using $\sqrt v-x=(v-x^2)/(\sqrt v+x)$ and $\sqrt v+x\ge\sqrt{c_v}$,
\[
\Bigl\lvert\frac{\sqrt v}{x}-1\Bigr\rvert=\frac{\lvert\sqrt v-x\rvert}{x}=\frac{\lvert v-x^2\rvert}{(\sqrt v+x)\,x}\le\frac{\lvert v-x^2\rvert}{\sqrt{c_v}\,\sqrt{c_v/2}}=\frac{\sqrt2}{c_v}\,\bigl\lvert T_w\,\hat s_{ri}^2-v\bigr\rvert .
\]
So $\mathbf 1_{G_T}\lvert\sqrt v/x-1\rvert\le(\sqrt2/c_v)\lvert T_w\,\hat s_{ri}^2-v\rvert\convp0$ by part (a).
\end{proof}

\paragraph{Completion of the proof of Theorem~\ref{thm:dml-var}.}
The proof is now a short chain: Proposition~\ref{prop:B-main} gives the normal limit on the scale $\sqrt{v_{T,ri}}$, Lemma~\ref{lem:C-se} swaps that scale for the code's standard error, and Lemma~\ref{lem:A-target}(iv) controls the recentring.

\begin{proof}[Proof of Theorem~\ref{thm:dml-var}]
Fix $(r,i)$ and let $G_T$ be the event of Lemma~\ref{lem:C-se}, so $P(G_T)\to1$ and $\hat s_{ri}>0$ on $G_T$.

\emph{Step 1 (the $t$-ratio).} Let $Q_T:=(\hbTheta_{ri}-\bar\bTheta_{ri})/\hat s_{ri}$ and
\[
X_T:=\frac{\sqrt{T_w}\,(\hbTheta_{ri}-\bar\bTheta_{ri})}{\sqrt{v_{T,ri}}},\qquad
Y_T:=1+\mathbf 1_{G_T}\Bigl(\frac{\sqrt{v_{T,ri}}}{\sqrt{T_w}\,\hat s_{ri}}-1\Bigr).
\]
On $G_T$, multiplying and dividing by $\sqrt{T_w}\sqrt{v_{T,ri}}$,
\[
Q_T=\frac{\sqrt{T_w}\,(\hbTheta_{ri}-\bar\bTheta_{ri})}{\sqrt{v_{T,ri}}}\cdot\frac{\sqrt{v_{T,ri}}}{\sqrt{T_w}\,\hat s_{ri}}=X_TY_T \qquad\text{on }G_T .
\]
Now $X_T\convd\N(0,1)$ by Proposition~\ref{prop:B-main}, $Y_T\convp1$ by Lemma~\ref{lem:C-se}(c), and $U_T:=Q_T-X_TY_T\convp0$, since $P(G_T)\to1$ and $U_T=0$ on $G_T$. ($Q_T=X_TY_T$ holds only on $G_T$; the term $U_T$ makes $Q_T=X_TY_T+U_T$ hold at every outcome.) To use Slutsky's lemma in the additive form stated before Proposition~\ref{prop:B-main}, write
\[
Q_T=X_T+\bigl[X_T(Y_T-1)+U_T\bigr].
\]
A sequence that converges in distribution is $O_p(1)$, so $X_T=O_p(1)$; with $Y_T-1\convp0$ this gives $X_T(Y_T-1)=o_p(1)$, since $O_p(1)\cdot o_p(1)=o_p(1)$. So the bracket is $o_p(1)$, and by Slutsky's lemma, $Q_T\convd\N(0,1)$. This is the first claim of the theorem.

\emph{Step 2 (the bias bound).} Lemma~\ref{lem:A-target}(iv) gives $\lVert\bar\bTheta-\bTheta(u_0)\rVert\le w_T=C\,T_w/T$ with $C:=L\,K_J/c_J$. This bound is used in Step 3.

\emph{Step 3 (recentring at $\bTheta(u_0)$).} Assume $T_w=o(T^{2/3})$, i.e.\ $T_w/T^{2/3}\to0$. Raising to the power $3/2$ (a continuous map with $0\mapsto0$) gives $T_w^{3/2}/T=(T_w/T^{2/3})^{3/2}\to0$. Write
\[
\frac{\hbTheta_{ri}-\bTheta_{ri}(u_0)}{\hat s_{ri}}=Q_T+U_T',\qquad U_T':=\frac{b_T}{\hat s_{ri}},\quad b_T:=\bar\bTheta_{ri}-\bTheta_{ri}(u_0),
\]
where $b_T$ is deterministic. By $\lvert M_{ri}\rvert\le\lVert M\rVert$ and Step 2, $\lvert b_T\rvert\le\lVert\bar\bTheta-\bTheta(u_0)\rVert\le C\,T_w/T$. On $G_T$, by Lemma~\ref{lem:C-se}(b), $\sqrt{T_w}\,\hat s_{ri}\ge\sqrt{c_v/2}$, so
\begin{align*}
\mathbf 1_{G_T}\lvert U_T'\rvert
=\mathbf 1_{G_T}\frac{\sqrt{T_w}\,\lvert b_T\rvert}{\sqrt{T_w}\,\hat s_{ri}}
\le\sqrt{\frac2{c_v}}\;C\,\frac{T_w^{3/2}}{T}\longrightarrow0.
\end{align*}
Since $P(G_T)\to1$, $U_T'\convp0$. By Step 1 and Slutsky's lemma, $(\hbTheta_{ri}-\bTheta_{ri}(u_0))/\hat s_{ri}\convd\N(0,1)$. So $\bar\bTheta_{ri}$ may be replaced by $\bTheta_{ri}(u_0)$.
\end{proof}

\subsection{Definitions and Proof of Proposition~\ref{prop:granger}}\label{app:granger-proof}
\begin{proposition}[Connection to Granger Non-Causality]\label{prop:granger}
Assume the stationary partially linear model~\eqref{eq:dml-Y}--\eqref{eq:dml-T}
with $\bT_t$ square-integrable and $J := \E[\bV_t \bV_t^\top]$ positive definite.
Under causal sufficiency (i.e., $\bW_t$ contains all common causes of
$\bY_t$ and its lags), the entry $\bA_k[i,j] = 0$ as estimated by DML-VAR if
and only if $Y^{(j)}$ does not Granger-cause $Y^{(i)}$ at lag $k$, conditional on
the full information set $\mathcal{F}_{t-1} = \sigma(\bY_{t-1}, \bY_{t-2},
\ldots, \bW_{t-1}, \bW_{t-2}, \ldots)$.
For the direction from Granger non-causality to $\bA_k[i,j]=0$, assume also the \emph{non-degeneracy condition}: $Y^{(j)}_{t-k}$ is not almost surely a function of the other variables in $\mathcal{F}_{t-1}$.
\end{proposition}

Proposition~\ref{prop:granger} connects the algebraic condition $\bA_k[i,j] = 0$ (a statement about model coefficients) to the statistical concept of Granger non-causality (a statement about predictive information).
Before proving the result, we state precise definitions of the key concepts.
These definitions formalize the central insight of \citet{granger1969investigating}: \emph{causes precede and help predict their effects}.

\paragraph{Key Definitions.}

\begin{definition}[Full Information Set]\label{def:info-set}
The \emph{full information set} available at time $t-1$ is the $\sigma$-algebra generated by the entire observable history:
\[
    \mathcal{F}_{t-1}
    \;=\; \sigma\!\bigl(\bY_{t-1}, \bY_{t-2}, \ldots,\;
                        \bW_{t-1}, \bW_{t-2}, \ldots\bigr).
\]
This includes every endogenous variable $Y^{(s)}_{t-\ell}$ and every exogenous variable $W_{r,t-\ell}$ for all $s,r$ and all lags $\ell \geq 1$.
\end{definition}

\begin{definition}[Granger Causality; \citet{granger1969investigating}]\label{def:granger-cause}
Variable $Y^{(j)}$ \emph{Granger-causes} variable $Y^{(i)}$ with respect to the information set $\mathcal{F}_{t-1}$ if there exists some lag $k \geq 1$ such that
\[
    \E\bigl[Y^{(i)}_t \mid \mathcal{F}_{t-1}\bigr]
    \;\neq\;
    \E\bigl[Y^{(i)}_t \mid \mathcal{F}_{t-1} \setminus \{Y^{(j)}_{t-k}\}\bigr].
\]
In words: $Y^{(j)}$ Granger-causes $Y^{(i)}$ if at least one past value $Y^{(j)}_{t-k}$ carries predictive information about $Y^{(i)}_t$ that is \emph{not already contained} in the remaining variables.
Removing that one lagged value from the information set strictly worsens the optimal (mean-squared-error) forecast of $Y^{(i)}_t$.
\end{definition}

\begin{definition}[Granger Non-Causality at Lag $k$]\label{def:granger-noncause}
Variable $Y^{(j)}$ \emph{does not Granger-cause} variable $Y^{(i)}$ \emph{at lag $k$} with respect to $\mathcal{F}_{t-1}$ if
\[
    \E\bigl[Y^{(i)}_t \mid \mathcal{F}_{t-1}\bigr]
    \;=\;
    \E\bigl[Y^{(i)}_t \mid \mathcal{F}_{t-1} \setminus \{Y^{(j)}_{t-k}\}\bigr].
\]
That is, the specific value $Y^{(j)}_{t-k}$ adds \emph{no marginal predictive information} about $Y^{(i)}_t$ beyond what is already known from all other variables.
Note that this is a \emph{lag-specific} concept: $Y^{(j)}$ may Granger-cause $Y^{(i)}$ at lag~1 but not at lag~2, for example.
Equivalently, $Y^{(j)}$ Granger-causes $Y^{(i)}$ at lag $k$ if removing $Y^{(j)}_{t-k}$ from the conditioning set \emph{does} change the optimal forecast of $Y^{(i)}_t$.
\end{definition}

\begin{definition}[Causal Sufficiency]\label{def:causal-suff}
The exogenous variable set $\bW_t$ satisfies \emph{causal sufficiency} for the endogenous system $\bY_t$ if every common cause of any pair of variables in $\bY_t$ (at any lag) is either
\emph{(i)}~already included in $\bY_t$ itself, or
\emph{(ii)}~contained in $\bW_t$.
Formally, there exists no unobserved stochastic process $\bU_t$ such that $\bU_t$ simultaneously causes both $Y^{(j)}_{t-k}$ and $Y^{(i)}_t$ for any $i, j \in \{1,\ldots,n\}$ and any $k \geq 1$.
\end{definition}

\begin{remark}[Why Causal Sufficiency Matters]\label{rem:causal-suff}
Causal sufficiency is the critical bridge between \emph{statistical} Granger causality (a statement about prediction) and \emph{structural} causality (a statement about direct causal mechanisms).
Without it, the equivalence in Proposition~\ref{prop:granger} can fail. Figure~\ref{fig:causal-suff} shows the failure (panel~b) and a hidden variable that does \emph{not} cause one (panel~c):
\begin{itemize}[leftmargin=1.5em]
    \item \textbf{Spurious edges (false positives):} An unobserved common cause $U$ of both $Y^{(j)}$ and $Y^{(i)}$ induces a statistical association between them.
    Since $U$ is not in $\bW_t$, DML cannot control for it.
    The estimator attributes $U$'s confounding effect to a direct $Y^{(j)} \to Y^{(i)}$ link, yielding $\hat{\bA}_k[i,j] \neq 0$ even when no direct causal relationship exists.
    \item \textbf{Hidden mediators (indirect edges, not a failure):} An unobserved mediator $U$ lies on the sole causal pathway from $Y^{(j)}$ to $Y^{(i)}$ (i.e., $Y^{(j)} \to U \to Y^{(i)}$), and no direct $Y^{(j)} \to Y^{(i)}$ edge exists.
    A mediator is not a common cause, so Definition~\ref{def:causal-suff} still holds, and the edge is not lost: with $U$ unobserved, the past of $Y^{(j)}$ carries $U$'s signal. The edge is real and $Y^{(j)}$ does Granger-cause $Y^{(i)}$; only reading it as a \emph{direct} mechanism would be wrong.
\end{itemize}
\end{remark}

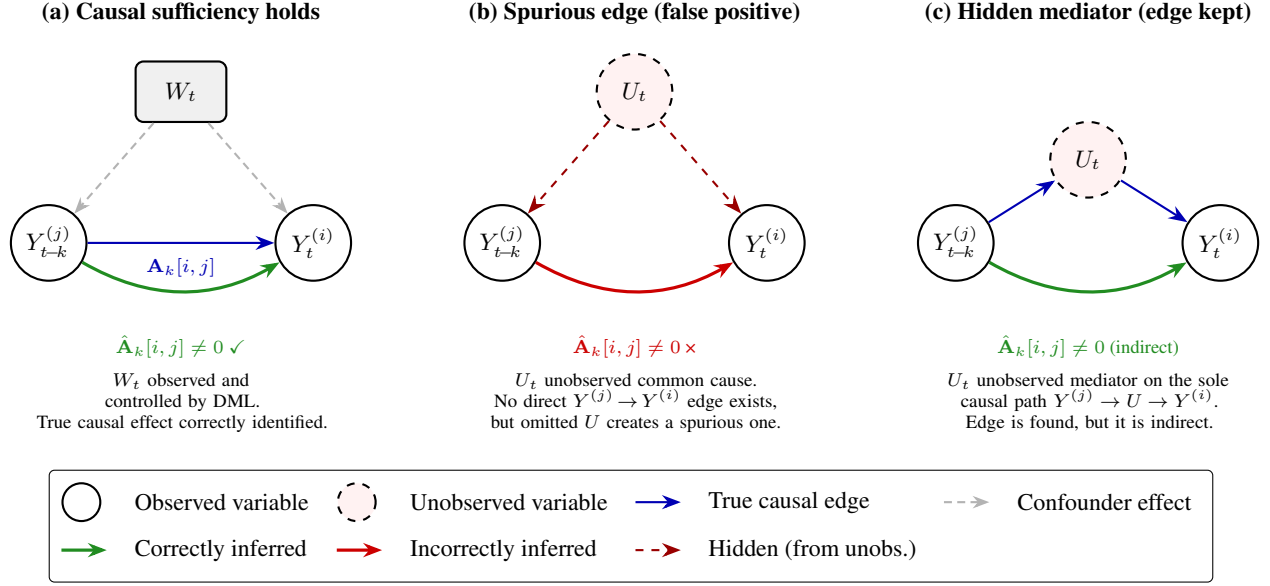
\begin{figure}[ht]
\centering
\begin{tikzpicture}[
    scale=1.0,
    every node/.style={font=\small},
    obs/.style={circle, draw, thick, minimum size=10mm, inner sep=0pt},
    unobs/.style={circle, draw, thick, dashed, minimum size=10mm,
                  inner sep=0pt, fill=red!5},
    confobs/.style={rectangle, draw, thick, rounded corners=3pt,
                    fill=gray!12, minimum width=12mm, minimum height=8mm,
                    inner sep=2pt},
    causal/.style={-{Stealth[length=2.5mm]}, thick, blue!70!black},
    confedge/.style={-{Stealth[length=2.2mm]}, thick, gray!60, densely dashed},
    spurious/.style={-{Stealth[length=2.5mm]}, thick, red!70!black, densely dotted},
    hidden/.style={-{Stealth[length=2.5mm]}, thick, red!60!black, dashed},
    inferred/.style={-{Stealth[length=2.5mm]}, very thick, ForestGreen},
    wrong/.style={-{Stealth[length=2.5mm]}, very thick, red!80!black},
    panellabel/.style={font=\bfseries\small, anchor=north},
]

\begin{scope}[shift={(0,0)}]
    \node[obs] (Yja) at (0, 0) {$Y^{(j)}_{t\!-\!k}$};
    \node[obs] (Yia) at (3.5, 0) {$Y^{(i)}_t$};
    \node[confobs] (Wa) at (1.75, 2) {$W_t$};

    \draw[causal] (Yja) -- (Yia)
        node[midway, below, font=\scriptsize, yshift=-2pt] {$\bA_k[i,j]$};
    \draw[confedge] (Wa) -- (Yja);
    \draw[confedge] (Wa) -- (Yia);

    \draw[inferred] (Yja) to[bend right=30] (Yia);
    \node[font=\scriptsize, ForestGreen, anchor=north] at (1.75, -1.1)
        {$\hat{\bA}_k[i,j] \neq 0$ \checkmark};

    \node[panellabel] at (1.75, 3.3)
        {(a) Causal sufficiency holds};
    \node[font=\scriptsize, text width=3.8cm, align=center, anchor=north]
        at (1.75, -1.6)
        {$W_t$ observed and controlled by DML.\\True causal effect correctly identified.};
\end{scope}

\begin{scope}[shift={(6,0)}]
    \node[obs] (Yjb) at (0, 0) {$Y^{(j)}_{t\!-\!k}$};
    \node[obs] (Yib) at (3.5, 0) {$Y^{(i)}_t$};
    \node[unobs] (Ub) at (1.75, 2) {$U_t$};

    \draw[hidden] (Ub) -- (Yjb);
    \draw[hidden] (Ub) -- (Yib);

    \draw[wrong] (Yjb) to[bend right=30] (Yib);
    \node[font=\scriptsize, red!80!black, anchor=north] at (1.75, -1.1)
        {$\hat{\bA}_k[i,j] \neq 0$ {\texttimes}};

    \node[panellabel] at (1.75, 3.3)
        {(b) Spurious edge (false positive)};
    \node[font=\scriptsize, text width=3.8cm, align=center, anchor=north]
        at (1.75, -1.6)
        {$U_t$ unobserved common cause.\\No direct $Y^{(j)}\!\to\!Y^{(i)}$ edge exists,\\but omitted $U$ creates a spurious one.};
\end{scope}

\begin{scope}[shift={(12,0)}]
    \node[obs] (Yjc) at (0, 0) {$Y^{(j)}_{t\!-\!k}$};
    \node[obs] (Yic) at (3.5, 0) {$Y^{(i)}_t$};
    \node[unobs] (Uc) at (1.75, 1.1) {$U_t$};

    \draw[causal] (Yjc) -- (Uc);
    \draw[causal] (Uc) -- (Yic);

    \draw[inferred] (Yjc) to[bend right=30] (Yic);
    \node[font=\scriptsize, ForestGreen, anchor=north] at (1.75, -1.1)
        {$\hat{\bA}_k[i,j] \neq 0$ (indirect)};

    \node[panellabel] at (1.75, 3.3)
        {(c) Hidden mediator (edge kept)};
    \node[font=\scriptsize, text width=3.8cm, align=center, anchor=north]
        at (1.75, -1.6)
        {$U_t$ unobserved mediator on the sole\\causal path $Y^{(j)}\!\to\!U\!\to\!Y^{(i)}$.\\Edge is found, but it is indirect.};
\end{scope}

\matrix[anchor=north west, draw, rounded corners=2pt, inner sep=4pt,
        column sep=6pt, row sep=2pt, font=\scriptsize,
        nodes={anchor=west}]
    at (0, -3.0) {
    \node[obs, minimum size=5mm, font=\tiny] {}; &
        \node{Observed variable}; &
    \node[unobs, minimum size=5mm, font=\tiny] {}; &
        \node{Unobserved variable}; &
    \draw[causal] (0,0) -- (0.6,0); &
        \node{True causal edge}; &
    \draw[confedge] (0,0) -- (0.6,0); &
        \node{Confounder effect}; \\
    \draw[inferred] (0,0) -- (0.6,0); &
        \node{Correctly inferred}; &
    \draw[wrong] (0,0) -- (0.6,0); &
        \node{Incorrectly inferred}; &
    \draw[hidden] (0,0) -- (0.6,0); &
        \node{Hidden (from unobs.)}; &
    &  \\
};

\end{tikzpicture}
\caption{Three scenarios illustrating the role of causal sufficiency in Proposition~\ref{prop:granger}.
\textbf{(a)}~When all confounders are observed ($W_t$), DML correctly controls for them and the estimated coefficient $\hat{\bA}_k[i,j]$ reflects the true causal relationship.
\textbf{(b)}~An unobserved common cause $U_t$ (dashed circle) of $Y^{(j)}$ and $Y^{(i)}$ creates a spurious association: DML estimates $\hat{\bA}_k[i,j] \neq 0$ even though no direct causal link exists.
\textbf{(c)}~An unobserved mediator $U_t$ on the sole causal pathway from $Y^{(j)}$ to $Y^{(i)}$: the indirect effect $Y^{(j)} \to U \to Y^{(i)}$ is not lost. It appears as a nonzero $\hat{\bA}_k[i,j]$ at the lag where the two steps add up, so the edge is real but indirect.
Causal sufficiency (Definition~\ref{def:causal-suff}) rules out scenario (b) by requiring that all common causes are observed; scenario (c) does not violate it, because a mediator is not a common cause.}
\label{fig:causal-suff}
\end{figure}

\paragraph{Proof.}

We now prove the equivalence stated in Proposition~\ref{prop:granger}.
The proof has two directions: showing that a zero coefficient implies Granger non-causality (the ``easier'' direction), and that Granger non-causality implies a zero coefficient (which requires causal sufficiency and the non-degeneracy condition).

\begin{proof}[Proof of Proposition~\ref{prop:granger}]

\medskip\noindent\textbf{Setup and notation.}\;
Recall the partially linear VAR model~\eqref{eq:plvar}. For a specific outcome variable $Y^{(i)}$, the model reads:
\begin{equation}\label{eq:plvar-scalar}
    Y^{(i)}_t
    = \underbrace{\sum_{\ell=1}^{p} \sum_{s=1}^{n} \bA_\ell[i,s]\, Y^{(s)}_{t-\ell}}_{\text{linear in lagged } \bY}
    + \underbrace{g^{(i)}(\bC_t)}_{\text{nonlinear in confounders}}
    + \underbrace{\varepsilon^{(i)}_t}_{\text{innovation}},
\end{equation}
where $\bC_t = (\bW_{t-1}^\top, \ldots, \bW_{t-p}^\top)^\top$ stacks the lagged confounders and $\E[\varepsilon^{(i)}_t \mid \mathcal{F}_{t-1}] = 0$.

The \emph{full information set} (Definition~\ref{def:info-set}) is
$\mathcal{F}_{t-1} = \sigma(\bY_{t-1}, \bY_{t-2}, \ldots, \bW_{t-1}, \bW_{t-2}, \ldots)$.
Taking the conditional expectation of~\eqref{eq:plvar-scalar} with respect to $\mathcal{F}_{t-1}$ and using the zero-mean innovation property, we obtain
\begin{equation}\label{eq:cond-exp}
    \E\bigl[Y^{(i)}_t \mid \mathcal{F}_{t-1}\bigr]
    = \sum_{\ell=1}^{p} \sum_{s=1}^{n} \bA_\ell[i,s]\, Y^{(s)}_{t-\ell}
      + g^{(i)}(\bC_t).
\end{equation}
This is the \emph{best prediction} (in the $L^2$ sense) of $Y^{(i)}_t$ given all available past information.
Notice that this expression has two parts: a \emph{linear} part that depends on lagged $\bY$ through the coefficients $\bA_\ell[i,s]$, and a \emph{nonparametric} part $g^{(i)}(\bC_t)$ that depends only on the confounders.
The coefficient $\bA_k[i,j]$ captures the \emph{direct, linear} contribution of $Y^{(j)}_{t-k}$ to the forecast of $Y^{(i)}_t$.

\medskip\noindent\textbf{Direction (a): Zero coefficient $\Rightarrow$ Granger non-causality.}\;

\emph{Goal:} Show that if $\bA_k[i,j] = 0$, then $Y^{(j)}$ does not Granger-cause $Y^{(i)}$ at lag $k$ (Definition~\ref{def:granger-noncause}).

\emph{Proof of direction (a):}
Suppose $\bA_k[i,j] = 0$.
We need to show that removing $Y^{(j)}_{t-k}$ from the conditioning set does not change the conditional expectation of $Y^{(i)}_t$.
Expanding~\eqref{eq:cond-exp}, we can isolate the contribution of $Y^{(j)}_{t-k}$:
\begin{align}
    \E\bigl[Y^{(i)}_t \mid \mathcal{F}_{t-1}\bigr]
    &= \bA_k[i,j]\, Y^{(j)}_{t-k}
       + \underbrace{\sum_{\substack{(\ell,s) \neq (k,j)}} \bA_\ell[i,s]\, Y^{(s)}_{t-\ell}
       + g^{(i)}(\bC_t)}_{\text{does not depend on } Y^{(j)}_{t-k}} \notag \\[4pt]
    &= 0 \cdot Y^{(j)}_{t-k}
       + \sum_{\substack{(\ell,s) \neq (k,j)}} \bA_\ell[i,s]\, Y^{(s)}_{t-\ell}
       + g^{(i)}(\bC_t) \notag \\[4pt]
    &= \sum_{\substack{(\ell,s) \neq (k,j)}} \bA_\ell[i,s]\, Y^{(s)}_{t-\ell}
       + g^{(i)}(\bC_t). \label{eq:no-yjk}
\end{align}
Since $\bA_k[i,j] = 0$, the term involving $Y^{(j)}_{t-k}$ vanishes.
The remaining expression in~\eqref{eq:no-yjk} depends only on variables in $\mathcal{F}_{t-1} \setminus \{Y^{(j)}_{t-k}\}$.
In particular, the confounder term $g^{(i)}(\bC_t)$ depends only on lagged $\bW$ and is completely unaffected by $Y^{(j)}_{t-k}$.
Therefore:
\[
    \E\bigl[Y^{(i)}_t \mid \mathcal{F}_{t-1}\bigr]
    = \E\bigl[Y^{(i)}_t \mid \mathcal{F}_{t-1} \setminus \{Y^{(j)}_{t-k}\}\bigr],
\]
which is exactly the definition of Granger non-causality at lag $k$ (Definition~\ref{def:granger-noncause}). \hfill $\square$ (Direction a)

\medskip
\emph{Intuition:}
This direction is straightforward---if the coefficient is zero, then $Y^{(j)}_{t-k}$ simply does not appear in the prediction formula.
Removing it from the information set has no effect.
Note that this direction does \emph{not} require causal sufficiency: it holds regardless of whether there are unobserved variables.

\medskip\noindent\textbf{Direction (b): Granger non-causality $\Rightarrow$ zero coefficient.}\;

\emph{Goal:} Show that if $Y^{(j)}$ does not Granger-cause $Y^{(i)}$ at lag $k$, then $\bA_k[i,j] = 0$.

\emph{This is the harder direction.}
It requires three conditions: (i)~the model~\eqref{eq:plvar} is correctly specified, (ii)~causal sufficiency holds, and (iii)~the non-degeneracy condition of Proposition~\ref{prop:granger} holds.
The proof proceeds in three steps.

\medskip
\noindent\textbf{Step 1: Correct model specification under causal sufficiency.}

Under causal sufficiency (Definition~\ref{def:causal-suff}), all common causes of $\bY_t$ are either in $\bY_t$ itself or in $\bW_t$.
This means that the partially linear model~\eqref{eq:plvar} captures \emph{all} relevant sources of variation in $Y^{(i)}_t$: the lagged endogenous variables $\bY_{t-1}, \ldots, \bY_{t-p}$ account for direct causal effects, and the confounder function $g^{(i)}(\bC_t)$ absorbs the effect of every exogenous common cause.
There is no omitted variable, so equation~\eqref{eq:cond-exp} is the \emph{true} conditional expectation:
\[
    \E\bigl[Y^{(i)}_t \mid \mathcal{F}_{t-1}\bigr]
    = \sum_{\ell=1}^{p} \sum_{s=1}^{n} \bA_\ell[i,s]\, Y^{(s)}_{t-\ell}
      + g^{(i)}(\bC_t)
    \quad \text{(exact, not an approximation)}.
\]

\medskip
\noindent\textbf{Step 2: Unique decomposition via DML orthogonality.}

We use the stacked notation of~\eqref{eq:stacked}. Let $\bTheta_{\cdot i}\in\R^{np}$ be column $i$ of $\bTheta$, so that the linear part of~\eqref{eq:cond-exp} is $\bTheta_{\cdot i}^\top\bT_t$. Entry $r:=(k-1)n+j$ of $\bTheta_{\cdot i}$ is $\bTheta_{ri}=\bA_k[i,j]$, the coefficient on $(\bT_t)_r=Y^{(j)}_{t-k}$.
Let $\bV_t:=\bT_t-g_T(\bC_t)$ be the treatment residual of~\eqref{eq:dml-T}: the part of the lags that the confounders cannot predict. Since $\bT_t$ is square-integrable (and so, by stationarity, is $\bY_t$), every expectation below is finite. By construction $\bV_t$ is orthogonal to every square-integrable function of $\bC_t$, and its second-moment matrix is the $J$ of Proposition~\ref{prop:granger}:
\[
        \E\bigl[\bV_t \mid \bC_t\bigr] = \bzero,
    \qquad
    J = \E\bigl[\bV_t\bV_t^\top\bigr].
\]
This orthogonality is critical because it ensures a \emph{unique decomposition} of the conditional expectation into its linear and nonparametric parts.
To see why, suppose two coefficient vectors $\bTheta_{\cdot i}$, $\bTheta_{\cdot i}'$ and two nuisance functions $g^{(i)}$, $g'^{(i)}$ yield the same conditional expectation:
\[
        \bTheta_{\cdot i}^\top\bT_t + g^{(i)}(\bC_t)
    = \bTheta_{\cdot i}'^{\top}\bT_t + g'^{(i)}(\bC_t)
    \quad \text{a.s.}
\]
Rearranging, multiplying on the left by the vector $\bV_t$ and taking expectations gives
\[
        \E\bigl[\bV_t\bT_t^\top\bigr]\bigl(\bTheta_{\cdot i}-\bTheta_{\cdot i}'\bigr)
    = \E\Bigl[\bV_t\bigl(g'^{(i)}(\bC_t) - g^{(i)}(\bC_t)\bigr)\Bigr] = \bzero,
\]
where the last equality uses $\E[\bV_t \mid \bC_t] = \bzero$ (the orthogonality property): a residual with conditional mean zero is uncorrelated with every square-integrable function of $\bC_t$, and $g'^{(i)}(\bC_t)-g^{(i)}(\bC_t)=(\bTheta_{\cdot i}-\bTheta_{\cdot i}')^\top\bT_t$ is one.
The same property simplifies the left-hand side. Since $\bT_t=g_T(\bC_t)+\bV_t$,
\[
        \E\bigl[\bV_t\bT_t^\top\bigr]
    = \underbrace{\E\bigl[\bV_t\,g_T(\bC_t)^\top\bigr]}_{=\,\bzero}
      {}+ \E\bigl[\bV_t\bV_t^\top\bigr]
    = J,
    \qquad\text{so}\qquad
    J\,\bigl(\bTheta_{\cdot i}-\bTheta_{\cdot i}'\bigr)=\bzero .
\]
By the assumption of Proposition~\ref{prop:granger}, $J$ is positive definite, hence invertible, which forces $\bTheta_{\cdot i}=\bTheta_{\cdot i}'$ (and then $g^{(i)}(\bC_t)=g'^{(i)}(\bC_t)$ a.s.).
Hence the whole coefficient vector of equation $i$, and in particular its entry $\bA_k[i,j]$, is \emph{uniquely identified}.
We project on all of $\bV_t$ at once because projecting on the single component $(\bV_t)_r$ would leave cross terms with the other components, which need not vanish unless $J$ is diagonal.

\medskip
\noindent\textbf{Step 3: Granger non-causality forces the coefficient to zero.}

Now suppose that $Y^{(j)}$ does not Granger-cause $Y^{(i)}$ at lag $k$ (Definition~\ref{def:granger-noncause}):
\[
    \E\bigl[Y^{(i)}_t \mid \mathcal{F}_{t-1}\bigr]
    = \E\bigl[Y^{(i)}_t \mid \mathcal{F}_{t-1} \setminus \{Y^{(j)}_{t-k}\}\bigr].
\]
The right-hand side is a function of the variables in $\mathcal{F}_{t-1} \setminus \{Y^{(j)}_{t-k}\}$ only, that is, of the past without $Y^{(j)}_{t-k}$.
By Step~1, the left-hand side equals $\bTheta_{\cdot i}^\top\bT_t + g^{(i)}(\bC_t)$.
In the linear component $\bTheta_{\cdot i}^\top\bT_t$, the only term that involves $Y^{(j)}_{t-k}$ is $\bTheta_{ri}(\bT_t)_r=\bA_k[i,j]\, Y^{(j)}_{t-k}$; every other entry of $\bT_t$ is a different variable of the past.
The confounder component $g^{(i)}(\bC_t)$ depends only on $\bC_t = (\bW_{t-1}^\top, \ldots, \bW_{t-p}^\top)^\top$, which does not include $Y^{(j)}_{t-k}$.
So the non-causality equation reads
\[
        \bA_k[i,j]\, Y^{(j)}_{t-k}
    + \underbrace{\bigl(\text{other terms of } \bTheta_{\cdot i}^\top\bT_t\bigr) + g^{(i)}(\bC_t)}_{\text{does not involve } Y^{(j)}_{t-k}}
    = \underbrace{\E\bigl[Y^{(i)}_t \mid \mathcal{F}_{t-1} \setminus \{Y^{(j)}_{t-k}\}\bigr]}_{\text{does not involve } Y^{(j)}_{t-k}}
    \quad \text{a.s.}
\]
Moving the middle term to the right-hand side, the equality says
\[
        \bA_k[i,j]\, Y^{(j)}_{t-k}
    = \text{a function of the other variables in } \mathcal{F}_{t-1}
    \quad \text{a.s.}
\]
If $\bA_k[i,j]\neq0$, dividing by it would make $Y^{(j)}_{t-k}$ itself almost surely a function of the other variables in $\mathcal{F}_{t-1}$.
The non-degeneracy condition of Proposition~\ref{prop:granger} rules this out, so $\bA_k[i,j] = 0$.
(Positive definiteness of $J$, discussed after (S)(ii) in Appendix~\ref{app:assumptions}, does not give this by itself: $Y^{(j)}_{t-k}$ could be a nonlinear function of the other lags, or of lags beyond~$p$.)

Finally, this $\bA_k[i,j]$ is also the coefficient that DML-VAR estimates.
In the stationary model the coefficients do not change over time, $\bTheta_t=\bTheta$, so the window target~\eqref{eq:A-target} is $\bar\bTheta=\bar J^{-1}\bar J\,\bTheta=\bTheta$.
Under the conditions of Theorem~\ref{thm:dml-var}, the estimator $\hbTheta$ is centred on $\bar\bTheta$, which here equals $\bTheta$, and entry $r$ of $\bTheta$ is $\bA_k[i,j]$.
\hfill $\square$ (Direction b)

\medskip
\noindent\textbf{Summary.}
Causal sufficiency (Definition~\ref{def:causal-suff}) rules out the spurious-edge failure mode:
\begin{enumerate}[leftmargin=2em]
    \item There are no unobserved common causes $\Rightarrow$ no omitted variable bias $\Rightarrow$ no spurious edges.
    \item Unobserved mediators are not ruled out, and need not be: the indirect effect shows up as a nonzero coefficient at the lag where the two steps add up, so no edge is masked (Remark~\ref{rem:causal-suff}).
\end{enumerate}
Under causal sufficiency and the non-degeneracy condition, the partially linear model~\eqref{eq:plvar} is correctly specified, the DML coefficients are uniquely identified, and the equivalence $\bA_k[i,j] = 0 \iff \text{Granger non-causality at lag } k$ holds exactly.
\end{proof}

\section{Complete Experimental Results}\label{app:results}

Table~\ref{tab:full-results} reports all evaluation metrics for every method, observability level, and first-stage learner combination in the synthetic benchmark.

\begin{table*}[!htbp]
\centering
\caption{Complete synthetic benchmark results. All metrics are averaged over the rolling evaluation window. Obs.: observability level. FDR, Power and F1 use BH at level $q=0.05$. VARX treats the observed confounders as extra series; its lag $\hat{p}_t$ is chosen on the equations of the $n$ target series only, and only their edges are scored. Bold indicates the best value in each column.}\label{tab:full-results}
\footnotesize
\setlength{\tabcolsep}{3.4pt}
\begin{tabular}{llcccccc}
\toprule
\textbf{Method} & \textbf{Obs.} & \textbf{Nonzero MAE}$\downarrow$ & \textbf{Fcst.\ MAE}$\downarrow$ & \textbf{Lag RMSE}$\downarrow$ & \textbf{FDR}$\downarrow$ & \textbf{Power}$\uparrow$ & \textbf{F1}$\uparrow$ \\
\midrule
\multicolumn{8}{l}{\emph{Non-DML methods}} \\
\midrule
VAR               & none      & 0.0605 & 0.1100 & 1.475 & 0.129 & 0.655 & 0.734 \\
ACLE-VAR           & none      & 0.0655 & 0.1094 & 1.006 & 0.147 & 0.661 & 0.730 \\
VARX               & all       & 0.0627 & 0.1092 & 1.239 & 0.076 & 0.650 & 0.755 \\
ACLE-VARX           & all       & 0.0691 & 0.1081 & 1.099 & 0.080 & 0.651 & 0.754 \\
VARX               & partial-2 & 0.0604 & 0.1098 & 1.292 & 0.096 & 0.647 & 0.745 \\
ACLE-VARX           & partial-2 & 0.0679 & 0.1091 & 1.060 & 0.108 & 0.651 & 0.743 \\
VARX               & partial-1 & 0.0601 & 0.1093 & 1.386 & 0.108 & 0.648 & 0.739 \\
ACLE-VARX           & partial-1 & 0.0659 & 0.1084 & 1.030 & 0.122 & 0.654 & 0.738 \\
\midrule
\multicolumn{8}{l}{\emph{DML methods: Extra Trees}} \\
\midrule
OR-VARX            & all       & 0.0627 & 0.0945 & 1.212 & 0.094 & 0.649 & 0.749 \\
ORACLE-VARX        & all       & 0.0626 & 0.0943 & 1.073 & 0.101 & 0.652 & 0.749 \\
OR-VARX            & partial-2 & 0.0595 & 0.0988 & 1.232 & 0.085 & 0.652 & 0.753 \\
ORACLE-VARX        & partial-2 & 0.0607 & 0.0984 & 1.033 & 0.096 & 0.656 & 0.751 \\
OR-VARX            & partial-1 & \textbf{0.0572} & 0.1097 & 1.532 & 0.152 & 0.662 & 0.724 \\
ORACLE-VARX        & partial-1 & 0.0620 & 0.1088 & 1.045 & 0.175 & 0.671 & 0.719 \\
\midrule
\multicolumn{8}{l}{\emph{DML methods: Random Forest}} \\
\midrule
OR-VARX            & all       & 0.0715 & 0.0961 & 1.247 & 0.103 & 0.650 & 0.741 \\
ORACLE-VARX        & all       & 0.0647 & 0.0958 & 1.107 & 0.117 & 0.652 & 0.736 \\
OR-VARX            & partial-2 & 0.0621 & 0.1013 & 1.371 & 0.100 & 0.651 & 0.746 \\
ORACLE-VARX        & partial-2 & 0.0617 & 0.1004 & 1.029 & 0.114 & 0.655 & 0.743 \\
OR-VARX            & partial-1 & 0.0606 & 0.1129 & 1.619 & 0.155 & 0.658 & 0.718 \\
ORACLE-VARX        & partial-1 & 0.0642 & 0.1125 & 1.164 & 0.186 & 0.665 & 0.707 \\
\midrule
\multicolumn{8}{l}{\emph{DML methods: LightGBM}} \\
\midrule
OR-VARX            & all       & 0.0665 & 0.1023 & 1.288 & \textbf{0.045} & 0.666 & 0.779 \\
ORACLE-VARX        & all       & 0.0579 & 0.1007 & 0.962 & 0.047 & \textbf{0.673} & \textbf{0.784} \\
OR-VARX            & partial-2 & 0.0687 & 0.1068 & 1.311 & 0.072 & 0.651 & 0.757 \\
ORACLE-VARX        & partial-2 & 0.0643 & 0.1048 & 1.052 & 0.084 & 0.655 & 0.755 \\
OR-VARX            & partial-1 & 0.0582 & 0.1130 & 1.303 & 0.169 & 0.651 & 0.713 \\
ORACLE-VARX        & partial-1 & 0.0661 & 0.1121 & 1.068 & 0.180 & 0.660 & 0.714 \\
\midrule
\multicolumn{8}{l}{\emph{DML methods: XGBoost}} \\
\midrule
OR-VARX            & all       & 0.1057 & 0.1111 & 1.609 & 0.081 & 0.667 & 0.762 \\
ORACLE-VARX        & all       & 0.0797 & 0.1103 & 1.037 & 0.103 & 0.671 & 0.752 \\
OR-VARX            & partial-2 & 0.0958 & 0.1205 & 1.745 & 0.084 & 0.652 & 0.753 \\
ORACLE-VARX        & partial-2 & 0.0802 & 0.1194 & 0.937 & 0.106 & 0.650 & 0.743 \\
OR-VARX            & partial-1 & 0.0777 & 0.1354 & 2.165 & 0.131 & 0.654 & 0.732 \\
ORACLE-VARX        & partial-1 & 0.0660 & 0.1384 & \textbf{0.915} & 0.173 & 0.661 & 0.717 \\
\midrule
\multicolumn{8}{l}{\emph{DML methods: TabPFN}} \\
\midrule
OR-VARX            & all       & 0.0671 & 0.0913 & 1.159 & 0.056 & 0.660 & 0.770 \\
ORACLE-VARX        & all       & 0.0601 & \textbf{0.0910} & 0.981 & 0.066 & \textbf{0.673} & 0.775 \\
OR-VARX            & partial-2 & 0.0594 & 0.0980 & 1.222 & 0.064 & 0.650 & 0.761 \\
ORACLE-VARX        & partial-2 & 0.0606 & 0.0973 & 1.023 & 0.076 & 0.654 & 0.758 \\
OR-VARX            & partial-1 & 0.0588 & 0.1080 & 1.423 & 0.150 & 0.659 & 0.727 \\
ORACLE-VARX        & partial-1 & 0.0618 & 0.1078 & 1.023 & 0.167 & 0.667 & 0.723 \\
\bottomrule
\end{tabular}
\end{table*}

\subsection{Partial Observability}\label{app:partial-obs}
At Partial-2 and Partial-1 the data are still generated by all three
confounders, but the method is given only two or one of them. The hidden
confounders then act like the hidden common cause of
Figure~\ref{fig:causal-suff}: they are autocorrelated and drive every
series at lag~1, so they create spurious lagged links between the
observed series. This violates causal sufficiency, one of the assumptions under which
edges are causal (Remark~\ref{rem:causal-suff}) and BH's FDR guarantee
holds (Section~\ref{sec:theory}). These levels therefore test robustness
outside the method's assumptions, not its guarantees. The two
causal-discovery baselines of Appendix~\ref{app:cd-baselines} (PCMCI and
VAR-LiNGAM) rely on the same assumption.

Lag recovery barely changes: across learners and levels, the lag RMSE of
ORACLE-VARX stays within 0.92--1.16 (LightGBM: 0.962 with all
confounders, 1.052 with two, 1.068 with one), and the adaptive methods
stay well ahead of their non-adaptive counterparts (Extra Trees with one
confounder: OR-VARX 1.532 $\to$ ORACLE-VARX 1.045, a 32\% improvement).
Edge discovery degrades, as expected once causal sufficiency fails: the
FDR of ORACLE-VARX (LightGBM) rises from 0.047 (All) to 0.084 (Partial-2)
and 0.180 (Partial-1), at similar power (Table~\ref{tab:full-results}).

\subsection{Causal-Discovery Baselines}\label{app:cd-baselines}
\paragraph{Setup.}
We run PCMCI~\citep{runge2019detecting} and
VAR-LiNGAM~\citep{hyvarinen2010estimation} from scratch in every rolling
window. PCMCI (tigramite 5.2.10.1) first prunes each variable's candidate
parents with a screening step at the default level $0.2$, then tests each
remaining link conditional on the selected parents; we use its linear
partial-correlation test (ParCorr). VAR-LiNGAM (lingam 1.13.0, defaults)
fits a VAR whose lag order is chosen by BIC within $1,\dots,5$, and then
prunes it with the adaptive lasso, a penalised refit that sets small
coefficients to exactly zero. Both use lags $1,\dots,p_{\max}=5$, the
$n=3$ series and the observed confounders as variables, and only rows
before output day $t$: PCMCI's regression targets are rows
$t-200,\dots,t-1$, the targets of VAR at $p_{\max}$, and VAR-LiNGAM fits
VAR's 205-row window. As for every method, only links among the $n$
series are scored.
\begin{itemize}
    \item \emph{Edges.} PCMCI's $p$-values for the $n^2p_{\max}$ links go
    through the same per-window BH step at $q=0.05$. VAR-LiNGAM's edges are
    its nonzero pruned lag coefficients; adaptive-lasso pruning has no FDR
    control.
    \item \emph{Coefficients and forecasts.} PCMCI returns a graph, not a
    model, so we refit each series by OLS on its discovered parents (plus
    confounder links with $p\le 0.05$), and take its lag order to be the
    largest lag with a discovered edge. VAR-LiNGAM forecasts from its
    fitted model; its VAR has no intercept, so each window is centred
    first.
\end{itemize}
Table~\ref{tab:cd-baselines} scores them on the DML methods' days
(Table~\ref{tab:full-results}); the code is \texttt{scripts/run\_cd\_baselines.py}.

\begin{table*}[!htbp]
\centering
\caption{Causal-discovery baselines on the synthetic benchmark, with ORACLE-VARX at full observability for reference. Same days and metrics as Table~\ref{tab:full-results}.}\label{tab:cd-baselines}
\footnotesize
\setlength{\tabcolsep}{3.4pt}
\begin{tabular}{llcccccc}
\toprule
\textbf{Method} & \textbf{Obs.} & \textbf{Nonzero MAE}$\downarrow$ & \textbf{Fcst.\ MAE}$\downarrow$ & \textbf{Lag RMSE}$\downarrow$ & \textbf{FDR}$\downarrow$ & \textbf{Power}$\uparrow$ & \textbf{F1}$\uparrow$ \\
\midrule
PCMCI + OLS refit & all       & 0.0646 & 0.1081 & 1.214 & 0.045 & 0.685 & 0.792 \\
PCMCI + OLS refit & partial-2 & 0.0664 & 0.1089 & 1.223 & 0.052 & 0.672 & 0.780 \\
PCMCI + OLS refit & partial-1 & 0.0664 & 0.1083 & 1.211 & 0.052 & 0.672 & 0.780 \\
PCMCI + OLS refit & none      & 0.0657 & 0.1091 & 1.136 & 0.066 & 0.673 & 0.775 \\
\midrule
VAR-LiNGAM        & all       & 0.0704 & 0.1074 & 1.132 & 0.187 & 0.641 & 0.703 \\
VAR-LiNGAM        & partial-2 & 0.0710 & 0.1082 & 1.132 & 0.196 & 0.641 & 0.697 \\
VAR-LiNGAM        & partial-1 & 0.0709 & 0.1081 & 1.132 & 0.209 & 0.641 & 0.692 \\
VAR-LiNGAM        & none      & 0.0703 & 0.1091 & 1.132 & 0.216 & 0.641 & 0.690 \\
\midrule
ORACLE-VARX, LightGBM & all   & 0.0579 & 0.1007 & 0.962 & 0.047 & 0.673 & 0.784 \\
ORACLE-VARX, TabPFN   & all   & 0.0601 & 0.0910 & 0.981 & 0.066 & 0.673 & 0.775 \\
\bottomrule
\end{tabular}
\end{table*}

\paragraph{Results.}
With all confounders observed, PCMCI and ORACLE-VARX (LightGBM) find edges
about equally well (FDR 0.045 vs 0.047, power 0.685 vs 0.673), and
ORACLE-VARX is better at what it is built for: forecasts (MAE 0.091 with
TabPFN vs 0.108), the lag order (RMSE 0.96--0.98 vs 1.21) and the
coefficients (nonzero MAE 0.058--0.060 vs 0.065). PCMCI tests every lag up
to $p_{\max}$, while our edges are tested only within the selected lag
order $\hat p_t$, so an under-selected lag also hides its edges; running
our BH step on the saved $z$-statistics over all $p_{\max}$ lags instead
gives FDR 0.023, power 0.676 and F1 0.797 (LightGBM). Below All, PCMCI's
FDR stays near 0.05 (0.052 at Partial-2 and Partial-1; 0.062 at None on
VAR's days, against 0.129 for VAR), while ours rises
(Appendix~\ref{app:partial-obs}); both methods are then outside their
causal-sufficiency assumption. VAR-LiNGAM's BIC picks lag~1 on every
scored day, its pruned edges have FDR 0.19--0.22, and its forecasts are
about as accurate as VARX's. Its distinctive step, identifying same-day
effects from non-Gaussian noise, has nothing to work on here: the
benchmark has no same-day effects and Gaussian structural noise, so its
lagged edges are pruned VAR coefficients without FDR control.

\section{ETF Trading Results}\label{app:etf-trading}

\subsection{Portfolio Construction}

We evaluate the economic value of the discovered causal structure through a
simple long-short trading strategy. At each trading day~$t$, the model produces
one-day-ahead log return forecasts $\hat{r}_{i,t}$ for each ETF~$i$. We
consider three portfolio construction rules:

\paragraph{Naive (sign-based).} Positions are determined solely by the sign of the
forecast:
\[
    w_{i,t} = \frac{\operatorname{sgn}(\hat{r}_{i,t})}{\sum_{j=1}^{n} |\operatorname{sgn}(\hat{r}_{j,t})|}.
\]

\paragraph{Weighted (magnitude-scaled).} Position sizes are proportional to the
magnitude of the forecasted return, so that assets with stronger predicted moves
receive larger allocations:
\[
    w_{i,t} = \frac{\hat{r}_{i,t}}{\sum_{j=1}^{n} |\hat{r}_{j,t}|}.
\]

\paragraph{Top-$q$ (quantile-filtered).} Only assets whose absolute forecast
exceeds the $q$-th quantile of $\{|\hat{r}_{j,t}|\}_{j=1}^{n}$ are
traded, that is, the $\lfloor n(1-q)\rfloor$ assets with the largest absolute forecasts. Among the selected assets, positions are sign-based and equally
weighted. We report results for $q \in \{0.25, 0.50, 0.75\}$, which trade 6, 4 and 2 of the 9 ETFs. The labels ``Top~25\%'', ``Top~50\%'' and ``Top~75\%'' in the tables name this quantile threshold, not the share of assets traded: Top~75\% is the most concentrated book.

\medskip
All strategies have unit gross exposure (positions are normalized so that
$\sum_i |w_{i,t}| = 1$ each day) and are rebalanced daily. The net exposure
$G_t = \sum_i w_{i,t}$ is generally not zero, so we offset it with an opposite
position in SPY, which makes the combined book dollar-neutral:
\[
    r_{\mathrm{adj},t} = \sum_{i} w_{i,t}\, r_{i,t} - G_t \cdot b_t,
\]
where $r_{i,t}$ and $b_t$ are the simple open-to-close returns of ETF~$i$
and of SPY on day~$t$. This is a dollar hedge, not a beta hedge: it treats
each sector ETF as having unit beta to SPY, so some market exposure can remain. Transaction costs are not modeled. All returns
are compounded multiplicatively: the cumulative return at day~$T$ is
$\prod_{t=1}^{T}(1 + r_{\mathrm{adj},t}) - 1$.

\subsection{Results}

Table~\ref{tab:etf-pnl} reports annualized returns and Sharpe ratios for the
weighted strategy across all methods using the \texttt{all10} confounder
preset (10 macroeconomic variables). Methods without confounders (VAR, ACLE-VAR)
are included as baselines.

A notable pattern emerges: simpler methods (ACLE-VAR, VARX) achieve higher
Sharpe ratios than the DML-based methods. The DML forecast still uses
the confounders, through the learned $\hat{\E}[\bY_t\mid\bW]$ term, but
fitting that term and the lag effects from 504 noisy days adds variance;
a short-lag VAR without confounders lets its lag coefficients absorb the
predictable part of the confounding instead (Section~\ref{sec:etf}). The primary contribution of ORACLE-VARX is the
quality of the recovered causal graph, not forecast-based trading performance.

\begin{table*}[!htbp]
\centering
\caption{ETF portfolio performance (weighted strategy, market-adjusted). All
DML methods use the \texttt{all10} confounder preset (10 macroeconomic
variables). Annualized return and Sharpe ratio are reported over the common
test period (2004--2020).}\label{tab:etf-pnl}
\small
\begin{tabular}{llcc}
\toprule
\textbf{Method} & \textbf{First-Stage Learner} & \textbf{Ann.\ Return (\%)} & \textbf{Sharpe Ratio} \\
\midrule
VAR               & ---         & 5.93 & 1.08 \\
ACLE-VAR          & ---         & 7.34 & 1.32 \\
VARX              & ---         & 4.22 & 0.75 \\
ACLE-VARX         & ---         & 4.80 & 0.86 \\
\midrule
OR-VARX           & Extra Trees & 3.51 & 0.67 \\
ORACLE-VARX       & Extra Trees & 3.55 & 0.68 \\
OR-VARX           & TabPFN      & 2.40 & 0.47 \\
ORACLE-VARX       & TabPFN      & 1.82 & 0.36 \\
\bottomrule
\end{tabular}
\end{table*}

\subsection{Full Strategy Comparison}\label{app:etf-strategies}

Tables~\ref{tab:perf-vix}--\ref{tab:perf-all10} disaggregate the results
across all five trading strategies and all three confounder presets.
Each cell reports the annualized return~(\%) and Sharpe ratio, separated by
a slash. Methods without confounders (VAR, ACLE-VAR) are repeated across
tables for ease of comparison.

\begin{table*}[!htbp]
\centering
\caption{Full strategy performance: \texttt{vix} confounder preset. Each
cell shows Ann.\ Return~(\%)\,/\,Sharpe Ratio (market-adjusted,
2004--2020). VAR and ACLE-VAR use no confounders and are included as
baselines. Best Sharpe per column in \textbf{bold}.}\label{tab:perf-vix}
\footnotesize
\begin{tabular}{llccccc}
\toprule
\textbf{Method} & \textbf{Learner} & \textbf{Naive} & \textbf{Weighted} & \textbf{Top 50\%} & \textbf{Top 25\%} & \textbf{Top 75\%} \\
\midrule
VAR         & ---           & 4.32\,/\,1.04 & 5.93\,/\,1.08 & 6.16\,/\,0.95 & 5.12\,/\,1.01 & 8.41\,/\,0.86 \\
ACLE-VAR    & ---           & 4.64\,/\,\textbf{1.15} & 7.34\,/\,\textbf{1.32} & 8.54\,/\,\textbf{1.31} & 6.73\,/\,\textbf{1.31} & 10.14\,/\,\textbf{1.04} \\
\midrule
VARX        & OLS           & 3.35\,/\,0.83 & 6.00\,/\,1.11 & 7.16\,/\,1.10 & 5.59\,/\,1.09 & 8.98\,/\,0.92 \\
ACLE-VARX   & OLS           & 3.09\,/\,0.78 & 5.83\,/\,1.07 & 6.47\,/\,1.00 & 5.85\,/\,1.14 & 8.76\,/\,0.90 \\
\midrule
OR-VARX     & LightGBM      & 1.94\,/\,0.51 & 2.55\,/\,0.50 & 2.63\,/\,0.45 & 1.99\,/\,0.43 & 2.04\,/\,0.26 \\
OR-VARX     & XGBoost       & 0.98\,/\,0.27 & 1.43\,/\,0.29 & 2.07\,/\,0.36 & 1.79\,/\,0.38 & 3.21\,/\,0.38 \\
OR-VARX     & Extra Trees   & 2.99\,/\,0.77 & 3.74\,/\,0.71 & 4.02\,/\,0.65 & 2.75\,/\,0.57 & 5.59\,/\,0.59 \\
OR-VARX     & Random Forest & 2.71\,/\,0.71 & 3.70\,/\,0.72 & 4.81\,/\,0.78 & 3.51\,/\,0.73 & 5.73\,/\,0.62 \\
OR-VARX     & TabPFN        & 2.82\,/\,0.73 & 3.47\,/\,0.68 & 3.25\,/\,0.54 & 3.66\,/\,0.75 & 3.71\,/\,0.43 \\
\midrule
ORACLE-VARX & LightGBM      & 0.65\,/\,0.18 & 1.44\,/\,0.29 & 2.49\,/\,0.42 & 1.65\,/\,0.35 & 1.52\,/\,0.20 \\
ORACLE-VARX & XGBoost       & 0.16\,/\,0.06 & 0.38\,/\,0.10 & 0.78\,/\,0.15 & 0.85\,/\,0.19 & 0.85\,/\,0.14 \\
ORACLE-VARX & Extra Trees   & 2.39\,/\,0.63 & 4.33\,/\,0.83 & 4.48\,/\,0.73 & 5.05\,/\,1.02 & 5.76\,/\,0.61 \\
ORACLE-VARX & Random Forest & 2.48\,/\,0.65 & 3.67\,/\,0.70 & 4.00\,/\,0.65 & 2.85\,/\,0.60 & 3.89\,/\,0.44 \\
ORACLE-VARX & TabPFN        & 2.13\,/\,0.56 & 3.57\,/\,0.70 & 4.89\,/\,0.80 & 4.67\,/\,0.94 & 4.55\,/\,0.51 \\
\bottomrule
\end{tabular}
\end{table*}

\begin{table*}[!htbp]
\centering
\caption{Full strategy performance: \texttt{macro5} confounder preset. Each
cell shows Ann.\ Return~(\%)\,/\,Sharpe Ratio (market-adjusted,
2004--2020). Best Sharpe per column in \textbf{bold}.}\label{tab:perf-macro5}
\footnotesize
\begin{tabular}{llccccc}
\toprule
\textbf{Method} & \textbf{Learner} & \textbf{Naive} & \textbf{Weighted} & \textbf{Top 50\%} & \textbf{Top 25\%} & \textbf{Top 75\%} \\
\midrule
VAR         & ---           & 4.32\,/\,1.04 & 5.93\,/\,1.08 & 6.16\,/\,0.95 & 5.12\,/\,1.01 & 8.41\,/\,0.86 \\
ACLE-VAR    & ---           & 4.64\,/\,\textbf{1.15} & 7.34\,/\,\textbf{1.32} & 8.54\,/\,\textbf{1.31} & 6.73\,/\,\textbf{1.31} & 10.14\,/\,\textbf{1.04} \\
\midrule
VARX        & OLS           & 0.87\,/\,0.23 & 2.34\,/\,0.43 & 2.51\,/\,0.40 & 1.51\,/\,0.32 & 2.99\,/\,0.35 \\
ACLE-VARX   & OLS           & 1.99\,/\,0.48 & 4.66\,/\,0.81 & 5.25\,/\,0.79 & 3.56\,/\,0.69 & 7.87\,/\,0.82 \\
\midrule
OR-VARX     & LightGBM      & 0.11\,/\,0.05 & 0.47\,/\,0.11 & $-$0.94\,/\,$-$0.12 & $-$0.36\,/\,$-$0.05 & 0.52\,/\,0.10 \\
OR-VARX     & XGBoost       & $-$0.21\,/\,$-$0.03 & $-$0.83\,/\,$-$0.12 & $-$0.95\,/\,$-$0.12 & $-$0.16\,/\,$-$0.01 & $-$2.96\,/\,$-$0.27 \\
OR-VARX     & Extra Trees   & 2.12\,/\,0.55 & 3.72\,/\,0.71 & 5.70\,/\,0.90 & 3.37\,/\,0.70 & 6.14\,/\,0.66 \\
OR-VARX     & Random Forest & 0.21\,/\,0.07 & 2.11\,/\,0.41 & 3.51\,/\,0.57 & 2.22\,/\,0.46 & 4.76\,/\,0.53 \\
OR-VARX     & TabPFN        & 2.43\,/\,0.63 & 2.86\,/\,0.54 & 2.74\,/\,0.45 & 3.55\,/\,0.71 & 3.95\,/\,0.45 \\
\midrule
ORACLE-VARX & LightGBM      & 0.59\,/\,0.17 & 1.57\,/\,0.31 & 1.32\,/\,0.24 & 1.41\,/\,0.31 & 0.97\,/\,0.15 \\
ORACLE-VARX & XGBoost       & $-$0.67\,/\,$-$0.15 & $-$0.78\,/\,$-$0.11 & 0.06\,/\,0.04 & $-$0.23\,/\,$-$0.02 & $-$2.20\,/\,$-$0.19 \\
ORACLE-VARX & Extra Trees   & 2.02\,/\,0.54 & 4.42\,/\,0.85 & 7.18\,/\,1.15 & 4.18\,/\,0.85 & 8.55\,/\,0.92 \\
ORACLE-VARX & Random Forest & $-$0.38\,/\,$-$0.08 & 1.56\,/\,0.32 & 2.74\,/\,0.47 & 1.13\,/\,0.25 & 5.14\,/\,0.59 \\
ORACLE-VARX & TabPFN        & 2.28\,/\,0.59 & 3.91\,/\,0.74 & 4.50\,/\,0.73 & 3.74\,/\,0.77 & 6.24\,/\,0.68 \\
\bottomrule
\end{tabular}
\end{table*}

\begin{table*}[!htbp]
\centering
\caption{Full strategy performance: \texttt{all10} confounder preset. Each
cell shows Ann.\ Return~(\%)\,/\,Sharpe Ratio (market-adjusted,
2004--2020). Best Sharpe per column in \textbf{bold}.}\label{tab:perf-all10}
\footnotesize
\begin{tabular}{llccccc}
\toprule
\textbf{Method} & \textbf{Learner} & \textbf{Naive} & \textbf{Weighted} & \textbf{Top 50\%} & \textbf{Top 25\%} & \textbf{Top 75\%} \\
\midrule
VAR         & ---           & 4.32\,/\,1.04 & 5.93\,/\,1.08 & 6.16\,/\,0.95 & 5.12\,/\,1.01 & 8.41\,/\,0.86 \\
ACLE-VAR    & ---           & 4.64\,/\,\textbf{1.15} & 7.34\,/\,\textbf{1.32} & 8.54\,/\,\textbf{1.31} & 6.73\,/\,\textbf{1.31} & 10.14\,/\,\textbf{1.04} \\
\midrule
VARX        & OLS           & 1.87\,/\,0.46 & 4.22\,/\,0.75 & 5.03\,/\,0.77 & 3.45\,/\,0.66 & 7.27\,/\,0.77 \\
ACLE-VARX   & OLS           & 1.89\,/\,0.47 & 4.80\,/\,0.86 & 5.68\,/\,0.86 & 3.38\,/\,0.65 & 5.98\,/\,0.66 \\
\midrule
OR-VARX     & LightGBM      & $-$0.44\,/\,$-$0.09 & 0.13\,/\,0.05 & $-$0.33\,/\,$-$0.02 & $-$0.46\,/\,$-$0.07 & 1.44\,/\,0.19 \\
OR-VARX     & XGBoost       & 0.79\,/\,0.21 & 0.29\,/\,0.08 & 0.64\,/\,0.13 & 0.15\,/\,0.06 & $-$2.95\,/\,$-$0.27 \\
OR-VARX     & Extra Trees   & 2.52\,/\,0.65 & 3.51\,/\,0.67 & 4.52\,/\,0.74 & 3.26\,/\,0.68 & 4.35\,/\,0.49 \\
OR-VARX     & Random Forest & 0.44\,/\,0.13 & 1.21\,/\,0.25 & 1.44\,/\,0.26 & 0.60\,/\,0.15 & 6.74\,/\,0.73 \\
OR-VARX     & TabPFN        & 1.94\,/\,0.52 & 2.40\,/\,0.47 & 3.25\,/\,0.53 & 2.01\,/\,0.43 & 1.55\,/\,0.20 \\
\midrule
ORACLE-VARX & LightGBM      & 0.00\,/\,0.02 & 0.66\,/\,0.15 & 0.97\,/\,0.18 & 0.39\,/\,0.10 & 2.13\,/\,0.26 \\
ORACLE-VARX & XGBoost       & 1.02\,/\,0.27 & 0.90\,/\,0.19 & 1.49\,/\,0.26 & 1.16\,/\,0.26 & $-$1.84\,/\,$-$0.16 \\
ORACLE-VARX & Extra Trees   & 2.08\,/\,0.54 & 3.55\,/\,0.68 & 4.23\,/\,0.71 & 3.51\,/\,0.72 & 7.27\,/\,0.80 \\
ORACLE-VARX & Random Forest & 0.48\,/\,0.15 & 0.90\,/\,0.20 & 1.08\,/\,0.21 & 0.50\,/\,0.13 & 3.74\,/\,0.44 \\
ORACLE-VARX & TabPFN        & 0.86\,/\,0.24 & 1.82\,/\,0.36 & 2.73\,/\,0.47 & 2.03\,/\,0.44 & 2.58\,/\,0.31 \\
\bottomrule
\end{tabular}
\end{table*}

\subsection{Cumulative PnL}\label{app:etf-pnl-plots}

Figure~\ref{fig:pnl-comparison} shows cumulative market-adjusted returns
for three representative configurations, illustrating the five trading
strategies over the full test period.

\begin{figure*}[htbp]
\centering
\includegraphics[width=0.75\textwidth]{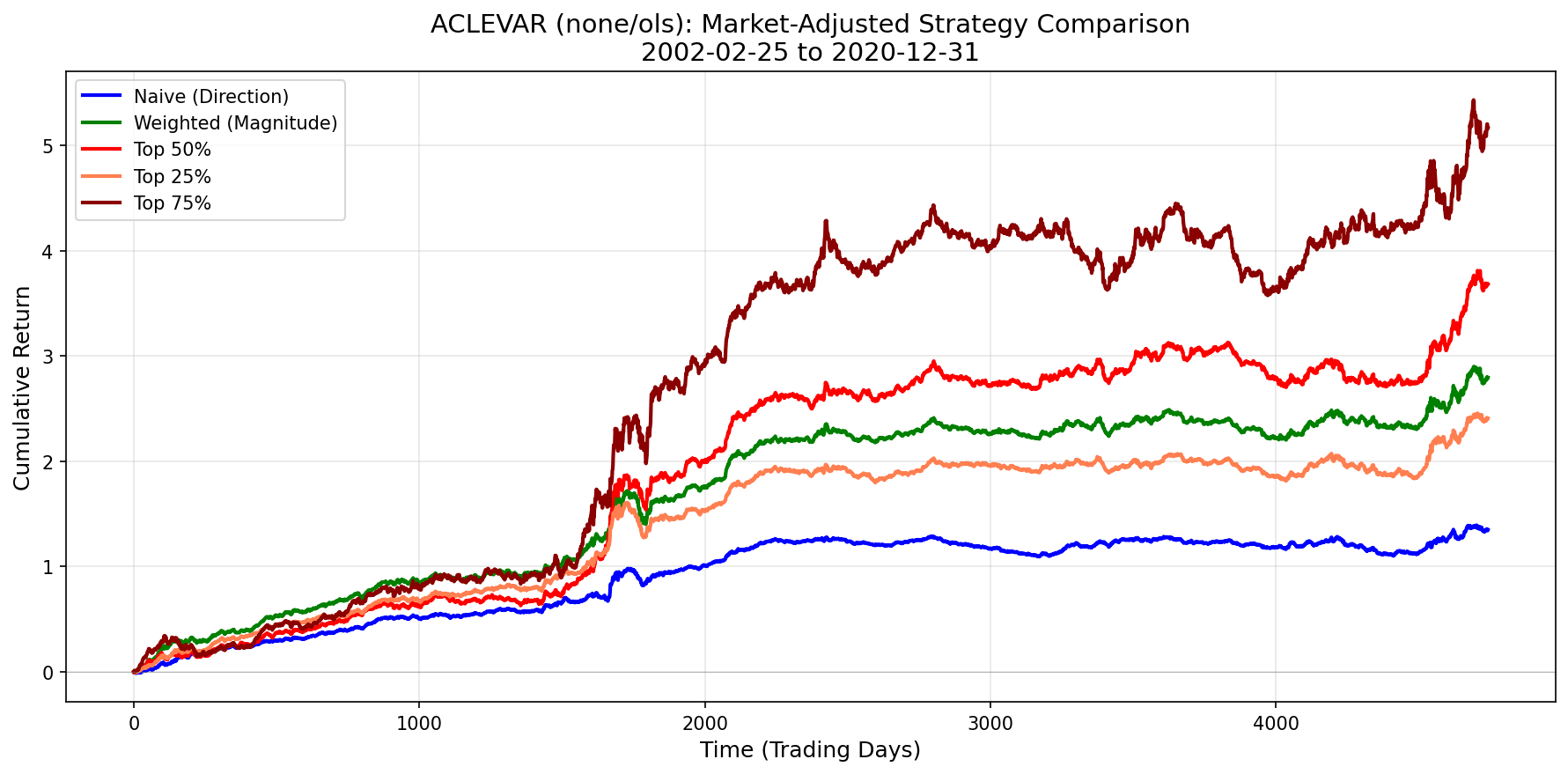}
\\[\smallskipamount]
{\small (a) ACLE-VAR (no confounders)}
\\[\medskipamount]
\includegraphics[width=0.75\textwidth]{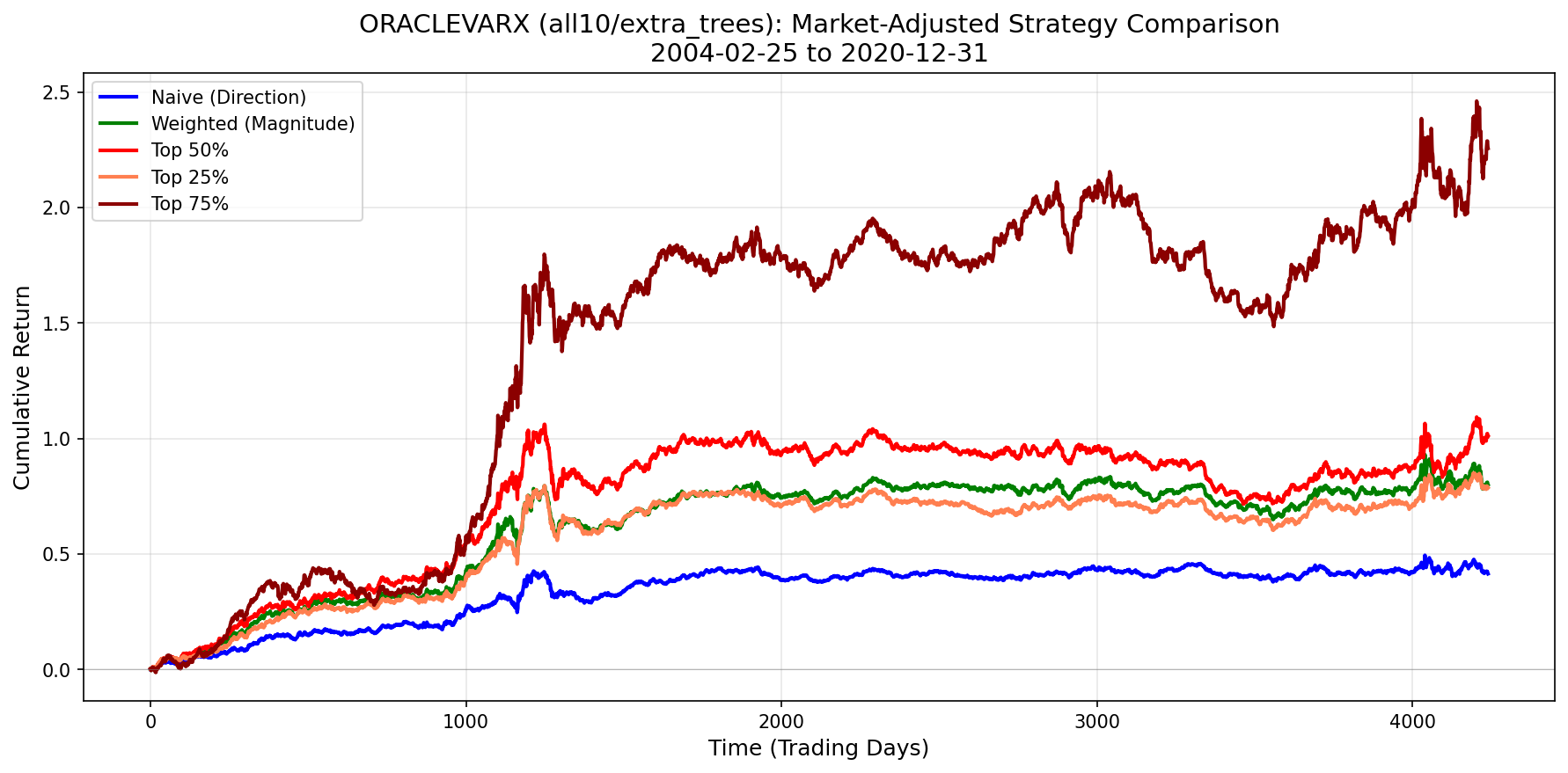}
\\[\smallskipamount]
{\small (b) ORACLE-VARX (Extra Trees, all10)}
\\[\medskipamount]
\includegraphics[width=0.75\textwidth]{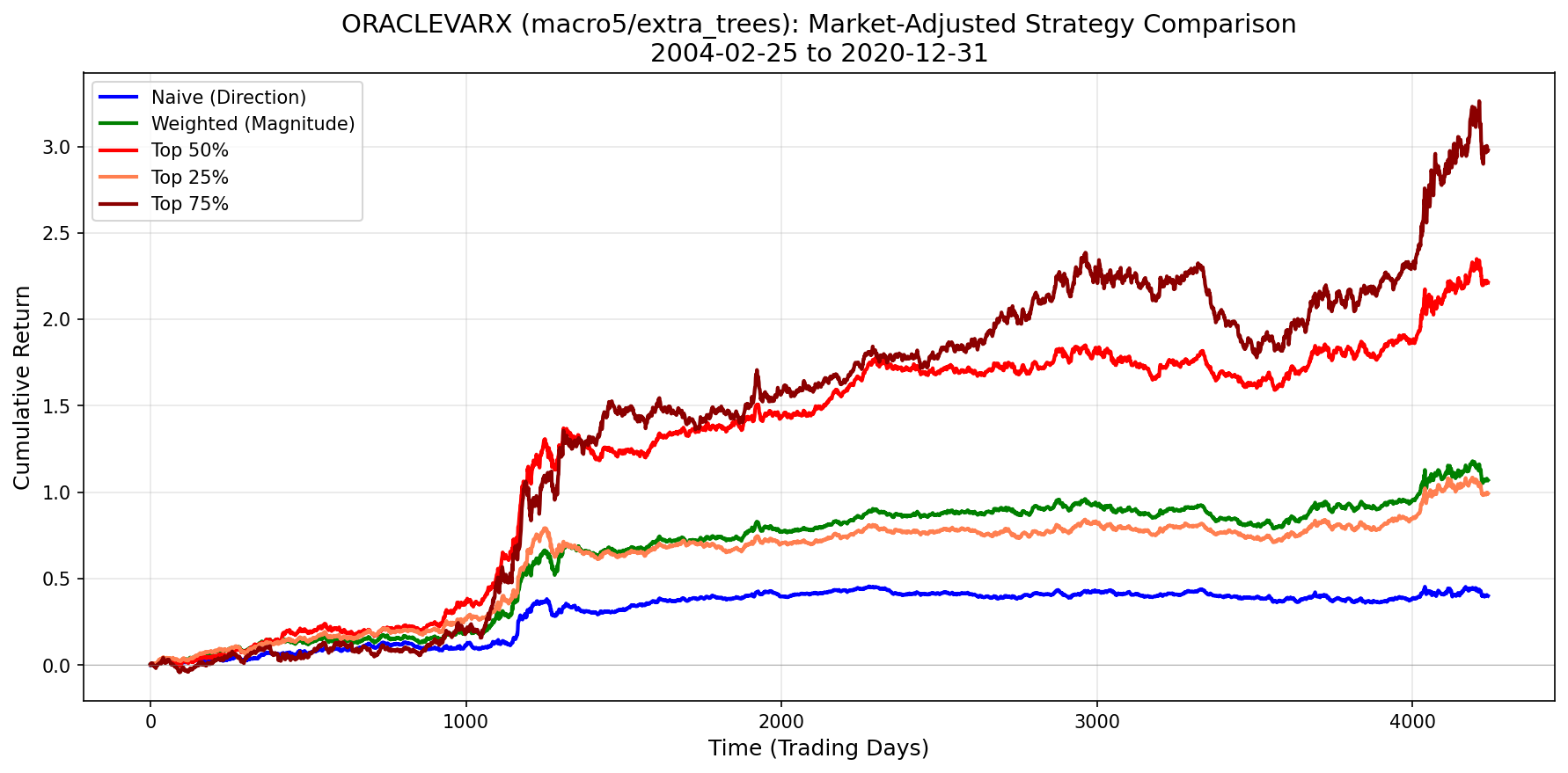}
\\[\smallskipamount]
{\small (c) ORACLE-VARX (Extra Trees, macro5)}
\caption{Cumulative market-adjusted returns for three representative
configurations. (a)~The ACLE-VAR baseline achieves the highest Sharpe ratios
overall (up to 1.32). (b)~ORACLE-VARX with Extra Trees and the
\texttt{all10} preset---the configuration used for the lag-volatility
analysis in Figure~\ref{fig:etf-lag}---produces stable, positive returns
across all strategies. (c)~The \texttt{macro5} preset with Extra Trees
yields the highest DML Sharpe ratio (1.15, Top~50\% strategy) among all DML
experiments.}\label{fig:pnl-comparison}
\end{figure*}

\paragraph{Discussion.}
Several patterns emerge from the full results.
\emph{First}, the OLS-based baselines (VAR, ACLE-VAR) consistently achieve
the highest Sharpe ratios, confirming that DML debiasing trades predictive
power for causal validity (Section~\ref{sec:experiments}).
\emph{Second}, among DML first-stage learners, Extra Trees is the most consistent, and TabPFN is usually next: LightGBM and XGBoost frequently produce near-zero or negative
returns, suggesting their nuisance estimates are less stable in the
rolling-window setting.
\emph{Third}, the \texttt{macro5} preset is generally the weakest for DML
methods---with the notable exception of ORACLE-VARX with Extra Trees, which
achieves the highest DML Sharpe ratio of 1.15 (Top~50\% strategy under
\texttt{macro5}).
\emph{Fourth}, the quantile-filtered strategies (Top~25\%, Top~50\%,
Top~75\%) often outperform the naive and weighted strategies in annualized
return, indicating that concentration on high-conviction signals is
beneficial.

\section{Implementation Details}\label{app:implementation}

\subsection{Experiment Configuration and Learner Hyperparameters}

Table~\ref{tab:config} summarizes the experiment configuration and first-stage learner hyperparameters for both benchmarks.

We use Extra Trees~\citep{geurts2006extremely} and Random Forest~\citep{breiman2001random} from scikit-learn~\citep{pedregosa2011scikit}, LightGBM~\citep{ke2017lightgbm}, XGBoost~\citep{chen2016xgboost} and TabPFN-2.5~\citep{hollmann2025tabpfn,grinsztajn2025tabpfn}. The causal-discovery baselines use the \texttt{tigramite} package for PCMCI~\citep{runge2019detecting} and the \texttt{lingam} package for VAR-LiNGAM~\citep{ikeuchi2023python,hyvarinen2010estimation}. The ETF open-to-close returns come from CRSP~\citep{crsp}. Calculated based on data from CRSP US Stock Databases \textcopyright 2024 Center for Research in Security Prices, LLC, An Affiliate of the University of Chicago Booth School of Business. All ten confounders come from FRED~\citep{fred}: the Economic Policy Uncertainty index is that of \citet{baker2016measuring}, the BBB spread is an ICE BofA index (ICE Data Indices, LLC), and VIX and the gold volatility index are copyright of the Chicago Board Options Exchange. Appendix~\ref{app:compute} lists the versions we ran, and the code repository lists all dependencies in its requirements files.

\begin{table*}[!htbp]
\centering
\caption{Experiment configuration and first-stage learner hyperparameters.}\label{tab:config}
\begin{tabular}{l|c|c}
\toprule
\multicolumn{3}{l}{\textbf{Panel A: Experiment Configuration}} \\
\midrule
\textbf{Parameter} & \textbf{Synthetic} & \textbf{ETF} \\
\midrule
$p_{\max}$ & 5 & 10 \\
$\alpha$-grid & \multicolumn{2}{c}{$\{0.01, 0.05, 0.10, 0.15, 0.20, 0.25, 0.30\}$} \\
OLS window & 200 days & 504 days \\
Tree training window & 200 days & 504 days \\
Validation days & 20 & 21 \\
Block length $n_{\mathrm{test}}$ (refit interval) & 20 days & 21 days \\
\midrule
\multicolumn{3}{l}{\textbf{Panel B: First-Stage Learner Hyperparameters}} \\
\midrule
\textbf{Learner} & \multicolumn{2}{l}{\textbf{Key Hyperparameters}} \\
\midrule
Extra Trees & \multicolumn{2}{l}{n\_estimators $= 100$, max\_features $= 1.0$, max\_depth $= 5$} \\
\addlinespace
Random Forest & \multicolumn{2}{l}{n\_estimators $= 100$, max\_features $= 1.0$, max\_depth $= 5$} \\
\addlinespace
LightGBM & \multicolumn{2}{l}{n\_estimators $= 100$, learning\_rate $= 0.1$, max\_depth $= 5$, num\_leaves $= 31$} \\
\addlinespace
XGBoost & \multicolumn{2}{l}{n\_estimators $= 100$, learning\_rate $= 0.3$, max\_depth $= 5$} \\
\addlinespace
TabPFN & \multicolumn{2}{l}{Pretrained transformer; no tuning required.} \\
       & \multicolumn{2}{l}{v2.5 weights (\texttt{tabpfn} 6.3.2); one forward pass, no ensembling.} \\
\bottomrule
\end{tabular}
\end{table*}

\subsection{Full Algorithm}\label{app:algorithm}

\begin{algorithm}[H]
\SetAlgoLined
\KwIn{Time series $\{\bY_t,\bW_t\}_{t=1}^T$;
  $p_{\max}$; window size $T_w$; $\alpha$-grid $\{\alpha_1,\ldots,\alpha_M\}$;
  validation window $v$; learner $\mathcal{L}$; edge FDR level $q$}
\KwOut{Sequence of causal lag graphs $\{\hat{\mathcal{G}}_t\}$ and forecasts $\{\hat{\bY}_t\}$}
\BlankLine
\For{each rolling window ending at $t = T_w, T_w{+}1, \ldots, T$}{
  \For{$p = 1,\ldots,p_{\max}$}{
    Build stacked treatment $\bT^{(p)}$ and confounders
    $\bC^{(p)}$\;
    Run DML cross-fitting with learner $\mathcal{L}$
    to obtain $\hg_Y^{(p)},\,\hg_T^{(p)}$\;
    Compute residuals
    $\tilde{\bY}^{(p)},\tilde{\bT}^{(p)}$
    via~\eqref{eq:residuals}\;
    Estimate $\hbTheta^{(p)}$ and standard errors
    via~\eqref{eq:theta-hat}--\eqref{eq:ols-se}\;
    Forecast $\hat{\bY}_t^{(p)} \leftarrow \hg_Y^{(p)}(\bC_t)
    + \hbTheta^{(p)\top}\bigl(\bT_t^{(p)} - \hg_T^{(p)}(\bC_t)\bigr)$\;
  }
  \tcp{Lag selection by ACLE (OR-VARX: RMSE~\eqref{eq:rmse-select} instead)}
  \For{$\alpha \in \{\alpha_1,\ldots,\alpha_M\}$}{
    $p_\alpha \leftarrow$ sequential significance testing
    at level $\alpha$ for $p = 2,\ldots,p_{\max}$ (lag 1 is always kept), using the $z$-scores of $\hbTheta^{(p_{\max})}$\;
  }
  $\alpha^*_t \leftarrow \argmin_\alpha \mathrm{RMSE}_v(p_\alpha)$
  via~\eqref{eq:alpha-select}\;
  $\hat{p}_t \leftarrow p_{\alpha^*_t}$;\quad $\hat{\bY}_t \leftarrow \hat{\bY}_t^{(\hat{p}_t)}$\;
  \BlankLine
  \tcp{Edge discovery with FDR control}
  Compute $z_{k,i,j} = \hat{\bA}_k[i,j] \,/\, \hat s(\hat{\bA}_k[i,j])$
  for all $(k,i,j)$ with $k\le\hat{p}_t$, from the $p_{\max}$ fit\;
  Apply BH-FDR at level $q$ to the $n^2\hat{p}_t$ $z$-scores\;
  $\hat{\mathcal{G}}_t \leftarrow$ edges surviving FDR threshold\;
}
\Return $\{\hat{\mathcal{G}}_t,\,\hat{\bY}_t\}_{t}$\;
\caption{OR(ACLE)-VARX: Orthogonalized Causal Lag Discovery}
\label{alg:oracle}
\end{algorithm}

\subsection{Computational Cost}\label{app:compute}

\paragraph{Reuse across windows.}
Algorithm~\ref{alg:oracle} is written one window at a time, but the implementation shares work across windows. The blocked split of Section~\ref{sec:dml-var} puts the cross-fitting blocks on one grid of days that does not depend on the window: block~$k$ trains on $[s_k,\,s_k+n_{\mathrm{train}})$, predicts the next $n_{\mathrm{test}}$ days, and $s_{k+1}=s_k+n_{\mathrm{test}}$. So, for each lag order $p$:
\begin{enumerate}
    \item each block's nuisance models are fit once, predict their $n_{\mathrm{test}}$ days, and are discarded;
    \item the out-of-sample residuals $(\tilde{\bY}_t,\tilde{\bT}_t^{(p)})$ are stored once, one row per day;
    \item the window ending at day $t$ reads rows $t-T_w,\dots,t-1$ of this array, so neighbouring windows share $T_w-1$ of their $T_w$ rows;
    \item batched least-squares calls solve~\eqref{eq:theta-hat} and~\eqref{eq:ols-se} for all windows together.
\end{enumerate}
The reuse is exact: recomputing these residuals separately inside each window would give the same numbers. OR-VARX and ORACLE-VARX share this first stage, and every $\alpha$ in the grid reuses the same $(\hbTheta,\widehat{\mathrm{SE}})$; only the $z$-tests and BH are rerun. There is no hyperparameter search (Table~\ref{tab:config}).

\paragraph{Number of first-stage fits.}
For each block and each $p=1,\dots,p_{\max}$, the first stage fits $n$ single-output learners for $\bY_t$ and $np$ for the lags $\bT_t^{(p)}$, so $S=\sum_{p=1}^{p_{\max}}n(p+1)=n\,p_{\max}(p_{\max}+3)/2$ fits per block. With $F$ blocks, the implementation makes $F\cdot S$ fits. Table~\ref{tab:fit-counts} compares this with two per-window alternatives:
\begin{itemize}
    \item \emph{(A) one fresh split per window}: train once on the $n_{\mathrm{train}}$ days before the window and predict its $T_w$ days. This is a cheaper, different estimator, and costs about $n_{\mathrm{test}}$ times more fits.
    \item \emph{(B) the same estimator, recomputed per window}: refit every block that covers the window. Its output equals ours, and it costs roughly $T_w$ times more fits.
\end{itemize}

\begin{table}[!htbp]
\centering
\caption{First-stage fits per run (one learner, one observation level). Ratios are relative to our implementation.}\label{tab:fit-counts}
\small
\begin{tabular}{lrr}
\toprule
 & Synthetic & ETF \\
\midrule
Windows & 2,595 & 4,261 \\
Blocks $F$ & 140 & 228 \\
Fits per block $S$ & 60 & 585 \\
\midrule
Ours ($F\cdot S$) & 8,400 & 133,380 \\
(A) fresh split per window & 155,700 (18.5$\times$) & 2,492,685 (18.7$\times$) \\
(B) same estimator per window & 1,704,960 (203$\times$) & 62,198,370 (466$\times$) \\
\bottomrule
\end{tabular}
\end{table}

\paragraph{Wall-clock time.}
Table~\ref{tab:runtime} times the full synthetic run (seed 42, all three confounders observed) on an Apple M5 laptop (10 CPU cores, 24~GB), one learner at a time with 9 threads (Python 3.12, PyTorch 2.14.1, scikit-learn 1.6.1, LightGBM 4.7, XGBoost 3.4). ``Ours'' covers the shared first stage, the batched OLS and OR-VARX over all 2,595 windows. Column (A) is measured on 10 evenly spaced windows and scaled to 2,595. Column (B) is computed as the measured time per block times the number of block fits (B) needs; it was not run. ORACLE-VARX adds 0.11--0.18\,s on top of the shared first stage, and the four OLS baselines take 0.10--0.31\,s each for all windows.

\begin{table}[!htbp]
\centering
\caption{Wall-clock time of one synthetic run (Apple M5 laptop, CPU).}\label{tab:runtime}
\small
\begin{tabular}{lrrrr}
\toprule
Learner & Ours & Per window & (A) & (B) \\
\midrule
Extra Trees & 6.1 min & 0.14 s & 1.9 h & 20.8 h \\
Random Forest & 9.4 min & 0.22 s & 2.4 h & 31.8 h \\
LightGBM & 13.6 min & 0.31 s & 3.7 h & 45.9 h \\
XGBoost & 17.1 min & 0.40 s & 6.1 h & 57.9 h \\
\bottomrule
\end{tabular}
\end{table}

\paragraph{Estimating other settings.}
The first stage dominates, so the run time is roughly
\[
    \text{time} \;\approx\; F\cdot S\cdot c_{\mathrm{fit}},\qquad F\approx\frac{T-n_{\mathrm{train}}-p_{\max}}{n_{\mathrm{test}}},
\]
where $c_{\mathrm{fit}}$ is the time of one single-output fit on $n_{\mathrm{train}}$ rows. From Table~\ref{tab:runtime}, $c_{\mathrm{fit}}$ is about 44\,ms (Extra Trees), 67\,ms (Random Forest), 97\,ms (LightGBM) and 122\,ms (XGBoost) on this laptop. For example, the 12 tree-based synthetic runs (4 learners $\times$ 3 observation levels) take about 2.3 hours in total, if the partially observed runs cost the same as the fully observed one. Assuming the same $c_{\mathrm{fit}}$ at $n_{\mathrm{train}}=504$, one ETF run with Extra Trees would take about $133{,}380\times 44\,\mathrm{ms}\approx 1.6$ hours; this is an extrapolation, not a measurement.

\paragraph{TabPFN.}
TabPFN needs no training, only a forward pass of the pretrained transformer. For each $p$, all (block, output) problems with the same test length are stacked into the transformer's batch dimension, and split only as GPU memory requires. With 80~GB of GPU memory, the 8,400 single-output problems of a synthetic run need 10--12 forward passes, and the 133,380 of an ETF run need 104--496, depending on the number of confounders. The same predictions give both the residuals and the forecasts, so forecasting needs no extra passes. On one cloud NVIDIA A100 GPU (80~GB; PyTorch 2.14.1, TabPFN 6.3.2), a full synthetic run takes 48--57\,s per observation level, 34--40\,s of it in TabPFN. The ETF runs take 17\,min (\texttt{vix}), 49\,min (\texttt{macro5}) and 88\,min (\texttt{all10}), almost all of it in TabPFN: more confounders mean more features per problem, so fewer problems fit in one pass.

\paragraph{Causal-discovery baselines.}
PCMCI and VAR-LiNGAM have no first stage to share across windows: each
window is a new search, so a run costs the time per window times the
number of windows. Table~\ref{tab:cd-runtime} times three evenly spaced
windows of each (one for CMIknn on the ETF data), one window at a time in
a single process on the same laptop, with all series and confounders as
variables, lags $1,\dots,p_{\max}$ and default settings; VAR-LiNGAM keeps
the lag fixed at $p_{\max}$. The ETF
problem has 19 variables (the 9 ETFs and the 10 \texttt{all10}
confounders), $T_w=504$ and $p_{\max}=10$. The code is
\texttt{scripts/time\_cd\_baselines.py}.

\begin{table}[!htbp]
\centering
\caption{Time of the causal-discovery baselines on the Apple M5 laptop CPU (TabPFN: one A100 GPU). Baselines: one window at a time in one process, scaled to a full run (2,595 synthetic and 4,261 ETF windows); windows are independent, so these runs split evenly over cores. Ours uses 9 threads. CMIknn is PCMCI with the nonparametric conditional-mutual-information test.}\label{tab:cd-runtime}
\small
\setlength{\tabcolsep}{3pt}
\begin{tabular}{lrrrr}
\toprule
 & \multicolumn{2}{c}{Synthetic} & \multicolumn{2}{c}{ETF (\texttt{all10})} \\
Method & Window & Run & Window & Run \\
\midrule
PCMCI (ParCorr) & 0.16 s & 7 min & 4.4--9.1 s & 5--11 h \\
PCMCI (CMIknn) & 6.2--6.6 min & 268--284 h & $>$1 h & $>$4,000 h \\
VAR-LiNGAM & 0.02 s & 1 min & 9--15 s & 11--17 h \\
\midrule
Ours, Extra Trees & 0.14 s & 6.1 min & -- & $\approx$1.6 h$^\ast$ \\
Ours, TabPFN (A100) & -- & 48--57 s & -- & 88 min \\
\bottomrule
\end{tabular}

\smallskip
{\footnotesize $^\ast$Extrapolated, not measured (see above).}
\end{table}

With linear tests the baselines are cheap: spread over 9 cores like
ours, linear PCMCI and VAR-LiNGAM would take about a minute on the
synthetic data and 0.6--1.9 hours on the ETF data, against 6.1 minutes
and an extrapolated 1.6 hours for ORACLE-VARX with Extra Trees on the same
laptop. The gap is in the nonlinear setting that ORACLE-VARX targets.
PCMCI with the nonparametric CMIknn test, the variant suited to nonlinear
confounding, takes over 6 minutes per synthetic window (about 30 hours
for a synthetic run even on 9 cores) and did not finish one ETF window
within an hour (over 450 hours for an ETF run on 9 cores). ORACLE-VARX
handles nonlinear confounding with the reused first stage instead.

\subsection{BH-FDR Procedure}

Edge discovery converts each debiased coefficient $\hat{\bA}_k[i,j]$ into a $z$-score:
\[
    z_{k,i,j} = \frac{\hat{\bA}_k[i,j]}{\hat s(\hat{\bA}_k[i,j])},
\]
where $\hat s$ is from the residualized OLS standard errors~\eqref{eq:ols-se}. Two-sided $p$-values are $p_{k,i,j} = 2(1 - \Phi(|z_{k,i,j}|))$. The Benjamini--Hochberg procedure is applied at level $q = 0.05$ to all $n^2 \cdot \hat{p}_t$ $p$-values per window:
\[
    r^* = \max\!\left\{r : p_{(r)} \leq \frac{r}{m} \cdot q\right\},
\]
where $m = n^2 \cdot \hat{p}_t$ is the number of tests and $p_{(1)} \leq \cdots \leq p_{(m)}$ are the sorted $p$-values. The edges with the $r^*$ smallest $p$-values are kept; if no $r$ satisfies the inequality, no edge is kept.

\section{Use of AI in the Theory}\label{app:ai-use}
The result to prove came from the authors: we needed a locally stationary version of the DML central limit theorem, and expected its proof to follow the DML argument of \citet{chernozhukov2018double}, with extra error terms from the time series (the drift of the coefficients inside the window, and dependence handled through martingale differences). The proofs in Appendix~\ref{app:proofs} were then produced by a multi-agent setup: Claude models (Opus~5.5 and Fable~5.1, by Anthropic) drafted the proofs and checked each other's drafts adversarially, and Codex (by OpenAI, with the Astra model) gave an independent adversarial review. Opus was instructed to write the proofs so that they spell out every step and skip none, so that a reader can verify each step. The one external theorem used, the martingale central limit theorem, was checked against the printed page of \citet{billingsley1995probability}.

An author with a background in mathematical analysis then checked the proofs line by line, verifying every step by hand, including all bounds and constants, and asking questions wherever a step was unclear; answers that exposed a problem led to a fix. Apart from the example below and presentation changes, this review found no mathematical errors.

\paragraph{Example: an inconsistency found during review.} In the proof of Theorem~\ref{thm:dml-var} (Step~1 of its completion in Appendix~\ref{app:se}), a display stated $Q_T=X_TY_T$, while the next sentence used $Q_T=X_TY_T+U_T$ with $U_T\convp0$, where $X_T\convd\N(0,1)$ and $Y_T\convp1$. While checking this step, the author asked why $Q_T$ was first written as $X_TY_T$ and then as $X_TY_T+U_T$:
\begin{quote}\itshape
``I thought $Q_T=X_TY_T$, how is it then [$X_TY_T$] $+\,U_T$ again now?''
\end{quote}
The first identity holds only on a high-probability event $G_T$, and $U_T$ is the correction that makes the second hold at every outcome; the display did not say so. The fix labels the display ``on $G_T$'' and adds one sentence explaining $U_T$. In resolving this, the assistant also noticed that the proof applied Slutsky's lemma to the product $X_TY_T$, whereas the appendix stated the lemma only in its additive form; the step was rewritten as $Q_T=X_T+[X_T(Y_T-1)+U_T]$ with $X_T(Y_T-1)=O_p(1)\cdot o_p(1)=o_p(1)$, so the stated form applies.

\end{document}